\documentclass[pra,showpacs,showkeys,twocolumn]{revtex4}%
\pdfoutput=1
\usepackage{amsfonts}
\usepackage{amsmath}
\usepackage{amssymb}
\usepackage{graphicx}
\usepackage{diagrams}%
\providecommand{\U}[1]{\protect\rule{.1in}{.1in}}
\newtheorem{theorem}{Theorem}

\newtheorem{corollary}{Corollary}

\newtheorem{lemma}{Lemma}

\newenvironment{proof}[1][Proof]{\noindent\textbf{#1.} }{\ \rule{0.5em}{0.5em}}
\newcommand{\nc}{\newcommand}
\nc{\rnc}{\renewcommand} \nc{\beq}{\begin{equation}}
\nc{\eeq}{{\end{equation}}} \nc{\bea}{\begin{eqnarray}}
\nc{\eea}{\end{eqnarray}} \nc{\beqa}{\begin{eqnarray}}
\nc{\eeqa}{\end{eqnarray}} \nc{\lbar}[1]{\overline{#1}}
\nc{\bra}[1]{\langle#1|} \nc{\ket}[1]{|#1\rangle}
\nc{\ketbra}[2]{|#1\rangle\!\langle#2|}
\nc{\braket}[2]{\langle#1|#2\rangle} \nc{\proj}[1]{|
#1\rangle\!\langle #1 |} \nc{\avg}[1]{\langle#1\rangle}

\rnc{\max}{\operatorname{max}} \nc{\rank}{\operatorname{rank}}
\nc{\conv}{\operatorname{conv}}
\nc{\smfrac}[2]{\mbox{$\frac{#1}{#2}$}} \nc{\Tr}{\operatorname{Tr}}
\nc{\ox}{\otimes} \nc{\dg}{\dagger} \nc{\dn}{\downarrow}
\nc{\lmax}{\lambda_{\text{max}}}
\nc{\lmin}{\lambda_{\text{min}}}
\nc{\cA}{{\cal A}} \nc{\cB}{{\cal B}} \nc{\cC}{{\cal C}}
\nc{\cD}{{\cal D}} \nc{\cE}{{\cal E}} \nc{\cF}{{\cal F}}
\nc{\cG}{{\cal G}} \nc{\cH}{{\cal H}} \nc{\cI}{{\cal I}}
\nc{\cJ}{{\cal J}} \nc{\cK}{{\cal K}} \nc{\cL}{{\cal L}}
\nc{\cM}{{\cal M}} \nc{\cN}{{\cal N}} \nc{\cO}{{\cal O}}
\nc{\cP}{{\cal P}} \nc{\cQ}{{\cal Q}} \nc{\cR}{{\cal R}} \nc{\cS}{{\cal S}}
\nc{\cT}{{\cal T}} \nc{\cU}{{\cal U}} \nc{\cV}{{\cal V}}
\nc{\cX}{{\cal X}} \nc{\cW}{{\cal W}} \nc{\cZ}{{\cal Z}}
\nc{\CA}{{\cal A}} \nc{\CB}{{\cal B}} \nc{\CC}{{\cal C}}
\nc{\CD}{{\cal D}} \nc{\CE}{{\cal E}} \nc{\CF}{{\cal F}}
\nc{\CG}{{\cal G}} \nc{\CH}{{\cal H}} \nc{\CI}{{\cal I}}
\nc{\CJ}{{\cal J}} \nc{\CK}{{\cal K}} \nc{\CL}{{\cal L}}
\nc{\CM}{{\cal M}} \nc{\CN}{{\cal N}} \nc{\CO}{{\cal O}}
\nc{\CP}{{\cal P}} \nc{\CQ}{{\cal Q}} \nc{\CR}{{\cal R}} \nc{\CS}{{\cal S}}
\nc{\CT}{{\cal T}} \nc{\CU}{{\cal U}} \nc{\CV}{{\cal V}}
\nc{\CX}{{\cal X}} \nc{\CW}{{\cal W}} \nc{\CZ}{{\cal Z}}
\nc{\csupp}{{\operatorname{csupp}}}
\nc{\qsupp}{{\operatorname{qsupp}}} \nc{\var}{\operatorname{var}}
\nc{\rar}{\rightarrow} \nc{\lrar}{\longrightarrow}
\nc{\poly}{\operatorname{poly}}
\nc{\polylog}{\operatorname{polylog}} \nc{\Lip}{\operatorname{Lip}}
\nc{\mb}[1]{\mathbf{#1}}
\nc{\ep}{\epsilon}
\nc{\Om}{\Omega}
\nc{\wt}[1]{\widetilde{#1}}
\def\>{\rangle}
\def\<{\langle}

\def\ph{\varphi}

\begin{document}
\preprint{ }
\title{Trade-off capacities of the quantum Hadamard channels}
\author{Kamil Br\'{a}dler, Patrick Hayden, Dave Touchette, and Mark M. Wilde}
\affiliation{School of Computer Science, McGill University, Montreal, Qu\'{e}bec, H3A 2A7, Canada}
\keywords{quantum Shannon theory, trading resources, entanglement-assisted classical and
quantum communication, Hadamard channel, cloning channel, Unruh channel}
\pacs{03.67.Hk, 03.67.Pp, 04.62.+v}

\begin{abstract}
Coding theorems in quantum Shannon theory express the ultimate rates at which
a sender can transmit information over a noisy quantum channel. More often
than not, the known formulas expressing these transmission rates are
intractable, requiring an optimization over an infinite number of uses of the
channel. Researchers have rarely found quantum channels with a tractable
classical or quantum capacity, but when such a finding occurs, it demonstrates
a complete understanding of that channel's capabilities for transmitting
classical or quantum information. Here, we show that the three-dimensional
capacity region for entanglement-assisted transmission of classical and
quantum information is tractable for the Hadamard class of channels. Examples
of Hadamard channels include generalized dephasing channels, cloning channels,
and the Unruh channel. The generalized dephasing channels and the cloning
channels are natural processes that occur in quantum systems
through the loss of quantum coherence or stimulated emission, respectively.
The Unruh channel is a noisy process that occurs in relativistic quantum
information theory as a result of the Unruh effect and bears a strong
relationship to the cloning channels. We give exact formulas for the
entanglement-assisted classical and quantum communication capacity regions of
these channels. The coding strategy for each of these examples is superior to
a na\"{\i}ve time-sharing strategy, and we introduce a measure to determine
this improvement.

\end{abstract}
\volumeyear{2010}
\volumenumber{ }
\issuenumber{ }
\eid{ }
\date{\today}
%\received{}

%\revised{}

%\accepted{}

%\published{}

\maketitle

\section{Introduction}

One of the aims of quantum Shannon theory is to characterize the ultimate
limits on the transmission of information over a noisy quantum channel.
Holevo, Schumacher, and Westmoreland contributed the first seminal result in
this direction by providing a lower bound for the ultimate limit of a noisy
quantum channel to transmit classical
information~\cite{ieee1998holevo,PhysRevA.56.131}, a result now known as the
HSW coding theorem. Lloyd, Shor, and Devetak then contributed increasingly
rigorous proofs of the quantum channel coding
theorem~\cite{PhysRevA.55.1613,capacity2002shor,ieee2005dev}, now known as the
LSD\ coding theorem, that provides a lower bound on the ultimate limit for a
noisy quantum channel to transmit quantum information. Other expository proofs
appeared later, providing insight into the nature of quantum
coding~\cite{PhysRevLett.93.080501,qcap2008first,qcap2008second,qcap2008third,qcap2008fourth}%
. Bennett \textit{et al}.~\cite{PhysRevA.54.3824} and Barnum \textit{et
al}.~\cite{ieee2000barnum}\ also showed that the capacity of a quantum channel
for transmitting quantum information is the same as that channel's capacity
for generating shared entanglement between sender and receiver. These three
results form the core of the dynamic, single-resource quantum Shannon theory,
where a sender exploits a noisy quantum channel to establish a single
noiseless resource, namely, classical communication, quantum communication, or
shared entanglement, with a receiver.

A formula for the capacity of a channel gives a \textquotedblleft
single-letter\textquotedblright\ characterization if the computation of the
capacity requires an optimization over only a single use of the channel, and
the formula gives a \textquotedblleft multi-letter\textquotedblright%
\ characterization otherwise. A single-letter characterization implies that
the computation of the capacity is tractable for a fixed input dimension of the channel, whereas a multi-letter
characterization typically requires an optimization over an infinite number of
uses of the channel and is therefore intractable in this case. This
\textquotedblleft single-letterization\textquotedblright\ issue does not play
a central role in classical information theory for the most basic task of
communication over a noisy classical channel, because single-letterization
occurs naturally in Shannon's original analysis for classical, memoryless
channels~\cite{bell1948shannon}. But this issue plays a prominent role in the
domain of quantum Shannon theory even for the most basic communication tasks.
Our knowledge so far indicates that the computation of the classical capacity
is intractable in the general case~\cite{H09,HW08,FKM09,BH09}, with the same
seeming to hold generally for the quantum
capacity~\cite{PhysRevA.57.830,smith:030501}. These results underscore our
incomplete\ understanding of the nature of quantum information, but they also
have the surprising and \textquotedblleft uniquely quantum\textquotedblright%
\ respective consequences that the strong correlations present in entangled
uses of a quantum channel boost the classical capacity and that the degeneracy
property of quantum codes can enhance the quantum capacity.

Thus, in hindsight, we might now say that any channel with a single-letter
capacity formula is a \textquotedblleft rare gem\textquotedblright\ in quantum
Shannon theory, given that the known formulas for capacities generally give
multi-letter characterizations. In fact,
researchers once conjectured that the HSW\ formula for the classical
capacity would generally give a single-letter characterization \cite{S04,S07}%
,\ until the recent result of Hastings~\cite{H09}. Researchers have found
several examples of these gems for the classical capacity:\ the identity
channel~\cite{AHW00}, unital qubit channels~\cite{king:4641}, erasure
channels~\cite{PhysRevLett.78.3217}, Hadamard channels~\cite{KMNR07},
entanglement-breaking channels~\cite{shor:4334}, depolarizing
channels~\cite{K03}, transpose-depolarizing channels~\cite{DHS06}, shifted
depolarizing channels~\cite{F05}, cloning channels~\cite{B09}, and the
so-called \textquotedblleft Unruh\textquotedblright\ channel~\cite{B09}.
Researchers have found fewer such exemplary single-letter gems for the quantum
capacity:\ erasure channels~\cite{PhysRevLett.78.3217}, degradable
channels~\cite{cmp2005dev,cubitt:102104}, conjugate degradable
channels~\cite{BDHM09}, and amplitude damping
channels~\cite{PhysRevA.71.032314}. We do not yet have a general method for
determining whether a given channel's capacity admits a single-letter
characterization using known formulas---the techniques for proving single-letterization of the
above examples are all \textit{ad hoc}, varying from case to case.
Additionally, we can observe that it is an even rarer gem for a channel to
admit a single-letter characterization for both the classical and quantum
capacity. This \textquotedblleft single-letter overlap\textquotedblright%
\ occurs for erasure channels~\cite{PhysRevLett.78.3217}, generalized
dephasing channels~\cite{cmp2005dev}, cloning channels~\cite{B09,BDHM09}, and
the Unruh channel~\cite{B09,BDHM09,BHP09}, but the reasons for
single-letterization of the classical and quantum capacities have no obvious connection.

After the early work in quantum Shannon theory, several researchers considered
how different noiseless resources such as entanglement, classical
communication, or quantum communication might trade off against one another
together with a noisy channel. The first findings in this direction were those
of Bennett \textit{et al}.~\cite{PhysRevLett.83.3081,ieee2002bennett}, who
showed that unlimited, shared, noiseless entanglement can boost the classical
capacity of a noisy quantum channel, generalizing the super-dense coding
effect \cite{PhysRevLett.69.2881}. Perhaps even more surprising was that the
formula for the entanglement-assisted classical capacity gives a single-letter
characterization, marking the first time that we could say there is a problem
in quantum Shannon theory that we truly understand. Shor then refined this
result by considering the classical capacity of a channel assisted by a finite
amount of shared entanglement \cite{arx2004shor}. He calculated a trade-off
curve that determines how a sender can optimally trade the consumption of
noiseless entanglement with the generation of noiseless classical
communication. This trade-off curve also bounds a rate region consisting of
rates of entanglement consumption and generated classical communication.
Unfortunately, the formulas for the rate region do not give a single-letter
characterization in the general case.

Shor's result then inspired Devetak and Shor to consider a scenario where a
sender exploits a noisy quantum channel to simultaneously transmit both noiseless classical
and quantum information \cite{cmp2005dev}, a scenario later dubbed
\textquotedblleft classically-enhanced quantum coding\textquotedblright%
\ \cite{HW08a,HW09}\ after schemes formulated in the theory of quantum error
correction~\cite{kremsky:012341,arx2008wildeUQCC}. Devetak and Shor provided a
multi-letter characterization of the classically-enhanced quantum capacity
region for general channels, but were able to show that both generalized
dephasing channels and erasure channels admit single-letter capacity regions.
We must emphasize that single-letterization of both the classical and
quantum capacities of a noisy quantum channel does not immediately imply the
single-letterization of the classically-enhanced quantum capacity region
(though the latter does imply the former). That is, the proof that the region
single-letterizes requires a different technique that has no obvious
connection to the techniques used to prove the single-letterization of the
individual capacities. The additional benefit of the Devetak-Shor
classically-enhanced quantum coding scheme is that it beats a time-sharing
strategy \footnote{Time-sharing in this case is a simple strategy that
exploits an HSW\ code for some fraction of the channel uses and an LSD\ code
for the other fraction of the channel uses.}\ for some channels.

We might say that the above scenarios are a part of the dynamic,
double-resource quantum Shannon theory, where a sender can exploit a noisy
quantum channel to generate two noiseless resources, or a sender can exploit a
noisy quantum channel in addition to a noiseless resource to generate another
noiseless resource. This theory culminated with the landmark work of Devetak
\textit{et al}.~that provided a multi-letter characterization for virtually
every permutation of two resources and a noisy quantum channel which one can
consider~\cite{arx2005dev,PhysRevLett.93.230504}, but they neglected to search
for channels with single-letter characterizations of the double-resource
capacity regions. Other researchers concurrently considered how noiseless
resources might trade off against each other in tasks outside of the dynamic,
double-resource quantum Shannon theory, such as quantum
compression~\cite{BHJW01,hayden:4404,KI01}, remote state
preparation~\cite{BHLSW05,PhysRevA.68.062319}, and hybrid quantum
memories~\cite{K03a}.

The next natural step in this line of inquiry was to consider the dynamic,
triple-resource quantum Shannon theory. Hsieh and Wilde did so by providing a
multi-letter characterization of an entanglement-assisted quantum channel's
ability to transmit both classical and quantum information \cite{HW08a,HW09}.
In addition, they found that the formulas for a generalized dephasing channel,
an erasure channel, and the trivial completely depolarizing channel all give a
single-letter triple trade-off capacity region (though they omitted the full
proof for the generalized dephasing channels). Thus, these findings provided a
complete understanding of the ultimate performance limits of any scheme for
entanglement-assisted classical and quantum error correction
\cite{science2006brun,DBH09,hsieh:062313,kremsky:012341,arx2008wildeUQCC,arx2007wildeEAQCC}%
, at least for the aforementioned channels. Ref.~\cite{HW08a} also constructed
a new protocol, dubbed the \textquotedblleft classically-enhanced father
protocol,\textquotedblright\ that outperforms a time-sharing strategy for
transmitting both classical and quantum information over an
entanglement-assisted quantum channel.

In this paper, we contribute a class of \textquotedblleft single-letter
gems\textquotedblright\ to the dynamic, triple-resource quantum Shannon
theory. We provide single-letter formulas for the triple trade-off
capacity region of the Hadamard class of channels. We then compute and plot
examples of the triple trade-off regions for the generalized dephasing
channels, the cloning channels, and the Unruh channel, all of which are
members of the class of Hadamard channels. The generalized dephasing channel
represents a natural mechanism for decoherence in physical systems such as
superconducting qubits~\cite{BDKS08}, the cloning channel represents a natural
process that occurs during stimulated
emission~\cite{MH82,SWZ00,LSHB02}, and the Unruh channel arises in
relativistic quantum information theory~\cite{B09,BDHM09,BHP09}, bearing
connections to the process of black-hole stimulated
emission~\cite{PhysRevD.14.870}. The proof technique to determine the formulas
for the cloning channel extends naturally to the formulas for the Unruh
channel, by exploiting the insights of Br\'{a}dler in Ref.~\cite{B09}. We also
find that the coding strategy for each of these channels beats a simple
time-sharing strategy, and we introduce a measure to compute the amount by
which it beats a time-sharing strategy.

We structure this work as follows. Section~\ref{sec:defnot} reviews the
definitions of noisy quantum channels, classical-quantum states,
information-theoretic quantities, and examples of noisy quantum channels that
are relevant for our purposes here. Section~\ref{sec:region-review}\ then
reviews the capacity regions mentioned above, specifically, the
classically-enhanced quantum (CQ)\ capacity region, the entanglement-assisted
classical (CE)\ capacity region, and the entanglement-assisted classical and
quantum (CQE)\ capacity region---we abbreviate a capacity region by the
noiseless resources involved:\ classical communication (C), quantum
communication (Q), or entanglement (E). We then present our main result in
Section~\ref{sec:single-letter-CEQ-EAC-Hadamard}: the proof that the CQ\ and
CE capacity regions of Hadamard channels admit a single-letter
characterization. As first observed by Hsieh and Wilde~\cite{HW08a}, and
perhaps surprisingly, the single-letterization of these two regions
immediately implies the single-letterization of the CQE\ capacity region.
Section~\ref{sec:param-curves} computes the CQE capacity region of the qubit
dephasing channel, the $1\rightarrow N$ cloning channel, and the Unruh
channel, and the next section plots the CQE regions. The final section
measures the improvement of the optimal protocols over a time-sharing
strategy. Finally, we conclude with some remaining observations and
suggestions for future work.

\section{Definitions and Notation}

\label{sec:defnot}

\subsection{Quantum Channels and Classical-Quantum States}

We first review the notion of a noisy quantum channel, an isometric extension,
and a classical-quantum state.

A noisy quantum channel $\mathcal{N}$ is a completely-positive
trace-preserving (CPTP) convex-linear map. It admits a Kraus representation
\cite{book2000mikeandike},\ so that its action on a density operator $\sigma$
is as follows:%
\begin{equation}
\mathcal{N}\left(  \sigma\right)  =\sum_{i}N_{i}\sigma N_{i}^{\dag},
\label{eq:kraus-map}%
\end{equation}
where the Kraus operators form a resolution of the identity:\ $\sum_{i}%
N_{i}^{\dag}N_{i}=I$, ensuring that the map is trace-preserving. The notion of
an isometric extension of a noisy quantum channel proves to be useful in
quantum Shannon theory---the notion is similar to that of a purification of a
density operator. Let $U_{\mathcal{N}}^{A^{\prime}\rightarrow BE}$ denote an
isometric extension of the noisy map $\mathcal{N}$, defined such
that~$U_{\mathcal{N}}^{\dag}U_{\mathcal{N}}=I$ and%
\[
\mathcal{N}^{A^{\prime}\rightarrow B}(\sigma
)=\mathop{{\mathrm{Tr}}_{E}}\{U_{\mathcal{N}}\sigma U_{\mathcal{N}}^{\dagger
}\}.
\]
The Kraus operators provide a straightforward method for constructing an
isometric extension:%
\[
U_{\mathcal{N}}^{A^{\prime}\rightarrow BE}=\sum_{i}N_{i}^{A^{\prime
}\rightarrow B}\otimes\left\vert i\right\rangle ^{E},
\]
where the set $\{\left\vert i\right\rangle ^{E}\}$ is an orthonormal set of
states. The following relation gives the conjugation of a density operator
$\sigma$ by the isometry $U_{\mathcal{N}}$:%
\begin{equation}
U_{\mathcal{N}}\sigma U_{\mathcal{N}}^{\dagger}=\sum_{i,j}(N_{i}\sigma
N_{j}^{\dag})^{B}\otimes\left\vert i\right\rangle \left\langle j\right\vert
^{E}. \label{eq:isometric-from-Kraus}%
\end{equation}
Let $\mathcal{N}^{c}$ denote a complementary channel of $\mathcal{N}$, unique
up to isometries on the system $E$, whose action on $\sigma$ is%
\[
\mathcal{N}^{c}(\sigma)\equiv\mathop{{\mathrm{Tr}}_{B}}\{U_{\mathcal{N}}\sigma
U_{\mathcal{N}}^{\dagger}\}.
\]
Observe that the following channel is a valid complementary channel for the
channel in (\ref{eq:kraus-map}):%
\[
\mathcal{N}^{c}\left(  \sigma\right)  =\sum_{i,j}\text{Tr}\left\{  N_{i}\sigma
N_{j}^{\dag}\right\}  \left\vert i\right\rangle \left\langle j\right\vert
^{E}.
\]
A quantum channel is degradable~\cite{cmp2005dev,cubitt:102104}, a notion
imported from classical information theory \cite{B73}, if there exists a
degrading channel $\mathcal{D}^{B\rightarrow E}$ that simulates the action of
the complementary channel $\left(  \mathcal{N}^{c}\right)  ^{A^{\prime
}\rightarrow E}$ so that%
\[
\forall\sigma\ \ \ \mathcal{D}^{B\rightarrow E}\circ\mathcal{N}^{A^{\prime
}\rightarrow B}\left(  \sigma\right)  =\left(  \mathcal{N}^{c}\right)
^{A^{\prime}\rightarrow E}\left(  \sigma\right)  .
\]

Suppose that Alice possesses an ensemble $\{(p_{X}\left(  x\right)  ,\rho
_{x}^{A^{\prime}})\}$\ of quantum states where $p_{X}\left(  x\right)  $ is the
probability density function for a random variable $X$ and $\rho
_{x}^{A^{\prime}}$ is a density operator conditional on the realization $x$ of
random variable $X$. She can augment this ensemble by correlating a classical
variable with each $\rho_{x}^{A^{\prime}}$. This procedure produces an
augmented ensemble $\{(p_{X}\left(  x\right)  ,\left\vert x\right\rangle
\left\langle x\right\vert ^{X}\otimes\rho_{x}^{A^{\prime}})\}$, where the
states $\{\left\vert x\right\rangle ^{X}\}$ form an orthonormal set. Taking
the expectation over the augmented ensemble gives the following
classical-quantum state:%
\begin{equation}
\rho^{XA^{\prime}}\equiv\sum_{x}p_{X}\left(  x\right)  |x\rangle\langle
x|^{X}\otimes\rho_{x}^{A^{\prime}}.\nonumber
\end{equation}
Let $\mathop{\left|\phi_x\right>}\nolimits^{AA^{\prime}}$ denote purifications
of each $\rho_{x}^{A^{\prime}}$. The following state is also a
classical-quantum state:%
\begin{equation}
\rho^{XAA^{\prime}}\equiv\sum_{x}p_{X}\left(  x\right)  |x\rangle\!\langle
x|^{X}\otimes|\phi_{x}\rangle\!\langle\phi_{x}|^{AA^{\prime}}.
\end{equation}
Suppose that Alice transmits the $A^{\prime}$ subsystem through a noisy
quantum channel $\mathcal{N}^{A^{\prime}\rightarrow B}$. The state output from
the channel is $\rho^{XAB}$ where%
\begin{align}
\rho^{XAB}  &  \equiv\mathcal{N}^{A^{\prime}\rightarrow B}(\rho^{XAA^{\prime}%
})\\
&  =\sum_{x}p_{X}\left(  x\right)  |x\rangle\!\langle x|^{X}\otimes
\mathcal{N}^{A^{\prime}\rightarrow B}(|\phi_{x}\rangle\!\langle\phi
_{x}|^{AA^{\prime}}).
\end{align}
It is implicit that an identity acts on a system for which there is no label
on the CPTP map $\mathcal{N}$. Note that the registers $XAB$ also form a
classical-quantum system. For such a multi-party state, we adopt the
convention that a state with a superscript unambiguously identifies a density
operator.\ For example, we define the reduced density operator $\rho^{A}%
\equiv\mathop{{\mathrm{Tr}}_{XB}}\{\rho^{XAB}\}$. We define states
$\mathop{\left|\phi_x\right>}\nolimits^{ABE}$ as%
\begin{equation}
\mathop{\left|\phi_x\right>}\nolimits^{ABE}\equiv U_{\mathcal{N}%
}\mathop{\left|\phi_x\right>}\nolimits^{AA^{\prime}}.\nonumber
\end{equation}
Observe that the states $\mathop{\left|\phi_x\right>}\nolimits^{ABE}$ purify
each $\rho_{x}^{AB}$. All of the above classical-quantum states are important
throughout this paper.

\subsection{Information-Theoretic Quantities}

We now define some standard information theoretic quantities. The von Neumann
entropy of a quantum state $\rho$ is defined as%
\[
H(\rho)\equiv-\text{Tr}\left\{  \rho\log{\rho}\right\}  .
\]
We write $H(A)_{\rho}\equiv H(\rho^{A})$ for a subsystem $A$ of $\rho$,
omitting the subscript if $\rho$ is implicitly clear. Note that the von
Neumann entropy is equal to the Shannon entropy for a classical system $X$:%
\[
H(X)\equiv-\sum_{x}p_{X}\left(  x\right)  \log\left(  {p_{X}\left(  x\right)
}\right)  .
\]
We define the conditional entropy$~H(A|B)$, the mutual information~$I(A;B)$,
and the coherent information~$I(A\rangle B)$ as follows for a bipartite state
$\rho^{AB}$:%
\begin{align*}
H(A|B)  &  \equiv H(AB)-H(B),\\
I(A;B)  &  \equiv H(A)-H(A|B),\\
I(A\rangle B)  &  \equiv-H(A|B).
\end{align*}
The Holevo quantity $\chi$\ for the state $\rho^{XB}$ of a classical-quantum
system $XB$ is%
\[
\chi(\{(p_{X}\left(  x\right)  ,\rho_{x}^{B})\})\equiv I\left(  X;B\right)
_{\rho}.
\]
The following relation holds by the joint entropy theorem
\cite{book2000mikeandike} when the conditioning system in a conditional
entropy is classical:%
\[
H(B|X)=\sum_{x}p_{X}\left(  x\right)  H(B)_{\rho_{x}^{B}}.
\]

\subsection{Quantum Channels}

\subsubsection{Entanglement-breaking Channels}

A noisy quantum channel $\mathcal{N}$\ is entanglement-breaking if it outputs
a separable state whenever half of any entangled state is the input to the
channel \cite{HSR03}. By the methods in Ref.~\cite{HSR03}, a channel is
entanglement-breaking if it produces a separable state when its input is half
of a maximally entangled state. More precisely, suppose that $\left\vert
\Phi\right\rangle ^{AA^{\prime}}$ is the maximally entangled qudit state:%
\begin{equation}
\left\vert \Phi\right\rangle ^{AA^{\prime}}\equiv\frac{1}{\sqrt{D}}\sum
_{i=0}^{D-1}\left\vert i\right\rangle ^{A}\left\vert i\right\rangle
^{A^{\prime}}. \label{eq:max-entangled}%
\end{equation}
A channel $\mathcal{N}^{A^{\prime}\rightarrow B}$ is entanglement-breaking if
and only if%
\[
\mathcal{N}^{A^{\prime}\rightarrow B}(\Phi^{AA^{\prime}})=\sum_{x}p_{X}\left(
x\right)  \rho_{x}^{A}\otimes\sigma_{x}^{B},
\]
where $p_{X}\left(  x\right)  $ is an arbitrary probability distribution and
$\rho_{x}^{A}$ and $\sigma_{x}^{B}$ are arbitrary density operators.

We can also think of an entanglement-breaking channel as a
quantum-classical-quantum channel \cite{HSR03}, in the sense that it first
applies a noisy channel $\mathcal{N}_{1}$, it performs a complete von Neumann
measurement of the resulting state, and it finally applies another noisy
channel $\mathcal{N}_{2}$. Consequently, we can write the action of an
entanglement-breaking channel $\mathcal{N}_{\text{EB}}$ as follows:%
\begin{equation}
\mathcal{N}_{\text{EB}}\left(  \rho\right)  =\sum_{x}\text{Tr}\left\{
\Lambda_{x}\rho\right\}  \sigma_{x}, \label{eq:ent-break-form}%
\end{equation}
where $\left\{  \Lambda_{x}\right\}  $ is a positive-operator-valued
measurement (POVM) and the states $\sigma_{x}$ are arbitrary. Additionally,
any entanglement-breaking channel admits a Kraus representation whose Kraus
operators are unit rank: $N_{i}=\left\vert \xi_{i}\right\rangle ^{B}%
\left\langle \varsigma_{i}\right\vert ^{A^{\prime}}$ \cite{HSR03}. Note that
the sets $\{\left\vert \xi_{i}\right\rangle ^{B}\}$ and $\{\left\vert
\varsigma_{i}\right\rangle ^{A^{\prime}}\}$ each do not necessarily form an
orthonormal set.

The classical capacity of an entanglement-breaking channel single-letterizes
\cite{shor:4334}, and its quantum capacity vanishes because it destroys
entanglement and thus cannot transmit quantum information.

\subsubsection{Hadamard Channels}

\label{sec:Hadamard-channel}A Hadamard channel is a quantum channel whose
complementary channel is entanglement-breaking~\cite{KMNR07}. We can write its
output as the Hadamard product (element-wise multiplication) of a
representation of the input density operator with another operator
\footnote{Strictly speaking, King~\textit{et~al}.~did not refer to such
channels as \textquotedblleft Hadamard channels,\textquotedblright\ but rather
as \textquotedblleft channels of the Hadamard form.\textquotedblright\ Here,
we loosely refer to them simply as \textquotedblleft Hadamard
channels.\textquotedblright}. To briefly review how this comes about, suppose
that the complementary channel$~\left(  \mathcal{N}^{c}\right)  ^{A^{\prime
}\rightarrow E}$\ of a channel$~\mathcal{N}^{A^{\prime}\rightarrow B}$ is
entanglement-breaking. Then, using the fact in the previous section that its
Kraus operators $\left\vert \xi_{i}\right\rangle ^{E}\left\langle \zeta
_{i}\right\vert ^{A^{\prime}}$ are unit rank and the construction in
(\ref{eq:isometric-from-Kraus}) for an isometric extension, we can write an
isometric extension$~U_{\mathcal{N}^{c}}$\ for $\mathcal{N}^{c}$ as%
\begin{align}
U_{\mathcal{N}^{c}}\sigma U_{\mathcal{N}^{c}}^{\dag}  &  =\sum_{i,j}\left\vert
\xi_{i}\right\rangle ^{E}\left\langle \zeta_{i}\right\vert ^{A^{\prime}}%
\sigma\left\vert \zeta_{j}\right\rangle ^{A^{\prime}}\left\langle \xi
_{j}\right\vert ^{E}\otimes\left\vert i\right\rangle ^{B}\left\langle
j\right\vert ^{B}\nonumber\\
&  =\sum_{i,j}\left\langle \zeta_{i}\right\vert ^{A^{\prime}}\sigma\left\vert
\zeta_{j}\right\rangle ^{A^{\prime}}\left\vert \xi_{i}\right\rangle
^{E}\left\langle \xi_{j}\right\vert ^{E}\otimes\left\vert i\right\rangle
^{B}\left\langle j\right\vert ^{B}. \label{eq:iso-Hadamard}%
\end{align}
The sets $\{\left\vert \xi_{i}\right\rangle ^{E}\}$ and $\{\left\vert
\zeta_{i}\right\rangle ^{A^{\prime}}\}$ each do not necessarily consist of
orthonormal states, but the set $\{\left\vert i\right\rangle ^{B}\}$ does.
Tracing over the system $E$ gives the original channel from system $A^{\prime
}$ to $B$:%
\begin{equation}
\mathcal{N}_{\text{H}}^{A^{\prime}\rightarrow B}\left(  \sigma\right)
=\sum_{i,j}\left\langle \zeta_{i}\right\vert ^{A^{\prime}}\sigma\left\vert
\zeta_{j}\right\rangle ^{A^{\prime}}\left\langle \xi_{j}|\xi_{i}\right\rangle
^{E}\left\vert i\right\rangle ^{B}\left\langle j\right\vert ^{B}.
\label{eq:hadamard-product}%
\end{equation}
Let $\Sigma$ denote the matrix with elements$~\left[  \Sigma\right]
_{i,j}=\left\langle \zeta_{i}\right\vert ^{A^{\prime}}\sigma\left\vert
\zeta_{j}\right\rangle ^{A^{\prime}}$, a representation of the input state
$\sigma$, and let $\Gamma$ denote the matrix with elements$~\left[
\Gamma\right]  _{i,j}=\left\langle \xi_{j}|\xi_{i}\right\rangle ^{E}$. Then,
from (\ref{eq:hadamard-product}), it is clear that the output of the channel
is the Hadamard product $\ast$\ of $\Sigma$ and $\Gamma^{\dag}$ with respect
to the basis $\{\left\vert i\right\rangle ^{B}\}$:%
\[
\mathcal{N}_{\text{H}}^{A^{\prime}\rightarrow B}\left(  \sigma\right)
=\Sigma\ast\Gamma^{\dag}.
\]
For this reason, such a channel is known as a Hadamard channel.

Hadamard channels are degradable. If Bob performs a von Neumann measurement of
his state in the basis$~\{\left\vert i\right\rangle ^{B}\}$ and prepares the
state $\left\vert \xi_{i}\right\rangle ^{E}$ conditional on the outcome of the
measurement, this procedure simulates the complementary channel$~\left(
\mathcal{N}^{c}\right)  ^{A^{\prime}\rightarrow E}$ and also implies that the
degrading map~$\mathcal{D}^{B\rightarrow E}$ is entanglement-breaking. To be
more precise, the Kraus operators of the degrading map$~\mathcal{D}%
^{B\rightarrow E}$\ are $\{\left\vert \xi_{i}\right\rangle ^{E}\left\langle
i\right\vert ^{B}\}$ so that%
\begin{align*}
\mathcal{D}^{B\rightarrow E}(\mathcal{N}_{\text{H}}^{A^{\prime}\rightarrow
B}\left(  \sigma\right)  )  &  =\sum_{i}\left\vert \xi_{i}\right\rangle
^{E}\left\langle i\right\vert ^{B}\mathcal{N}^{A^{\prime}\rightarrow B}\left(
\sigma\right)  \left\vert i\right\rangle ^{B}\left\langle \xi_{i}\right\vert
^{E}\\
&  =\sum_{i}\left\langle i\right\vert ^{A^{\prime}}\sigma\left\vert
i\right\rangle ^{A^{\prime}}\left\vert \xi_{i}\right\rangle \left\langle
\xi_{i}\right\vert ^{E},
\end{align*}
demonstrating that this degrading map effectively simulates the complementary
channel$~\left(  \mathcal{N}_{\text{H}}^{c}\right)  ^{A^{\prime}\rightarrow E}$. Note
that we can view this degrading map as the composition of two maps:\ a first
map $\mathcal{D}_{1}^{B\rightarrow Y}$ performs the von Neumann measurement, leading to a classical variable $Y$,
and a second map $\mathcal{D}_{2}^{Y\rightarrow E}$
performs the state preparation, conditional on the value of
the classical variable $Y$. We can therefore write
$\mathcal{D}^{B\rightarrow E}=\mathcal{D}_{2}^{Y\rightarrow E%
}\circ\mathcal{D}_{1}^{B\rightarrow Y}$. This observation is crucial to our proof of the
single-letterization of both the CQ\ and CE\ capacity regions of Hadamard channels.
%TCIMACRO{\TeXButton{TeX field}%
%{These structural relationships are summarized in the following commutative diagram:
%\begin{diagram}
%&            &            &  B                 \\
%&            & \ruTo^{\mathcal{N}_{\text{H}}}  &  \dDotsto_{\mathcal{D}_1}         \\
%& A^{\prime}           &            &  \; \; \quad\quad\quad\quad Y \; \mbox
%{(classical)}            \\
%&            & \rdTo_{\mathcal{N}_{\text{H}}^c} &  \dDotsto_{\mathcal{D}_2}
%\\
%&            &              &  E
%\end{diagram}}}%
%BeginExpansion
These structural relationships are summarized in the following commutative diagram:
\begin{diagram}
&            &            &  B                 \\
&            & \ruTo^{\mathcal{N}_{\text{H}}}  &  \dDotsto_{\mathcal{D}_1}         \\
& A^{\prime}           &            &  \; \; \quad\quad\quad\quad Y \; \mbox
{(classical)}            \\
&            & \rdTo_{\mathcal{N}_{\text{H}}^c} &  \dDotsto_{\mathcal{D}_2}
\\
&            &              &  E
\end{diagram}%
%EndExpansion

We show in the next two sections that generalized dephasing channels and
cloning channels are both members of the Hadamard class because their
complementary channels are entanglement-breaking.

\subsubsection{Generalized Dephasing Channels}

\label{sec:generalized-dephasing}Generalized dephasing channels model physical
processes where there is no loss of energy but there is a loss of quantum
coherence \cite{book2000mikeandike}, a type of quantum noise that dominates for example in
superconducting qubits~\cite{BDKS08}. The respective input and output systems
$A^{\prime}$ and $B$ of such channels are of the same dimension. Let
$\{\left\vert i\right\rangle ^{A^{\prime}}\}$ and $\{\left\vert i\right\rangle
^{B}\}$ be some respective preferred orthonormal bases for these systems, the
first of which we call the dephasing basis. The channel does not affect any
state that is diagonal in the dephasing basis, but it mixes coherent
superpositions of these basis states.

An isometric extension $U_{\mathcal{N}_{\text{GD}}}^{A^{\prime}\rightarrow
BE}$ of a generalized dephasing channel $\mathcal{N}_{\text{GD}}^{A^{\prime
}\rightarrow B}$ has the form:%
\begin{equation}
U_{\mathcal{N}_{\text{GD}}}^{A^{\prime}\rightarrow BE}\equiv\sum_{i}\left\vert
i\right\rangle ^{B}\left\langle i\right\vert ^{A^{\prime}}\otimes\left\vert
\vartheta_{i}\right\rangle ^{E}, \label{eq:iso-gen-dephasing}%
\end{equation}
where the set $\{\left\vert \vartheta_{i}\right\rangle ^{E}\}$ is not
necessarily an orthonormal set. The output of a generalized dephasing channel
is as follows:%
\begin{equation}
\mathcal{N}_{\text{GD}}(\sigma)=\sum_{i,j}\langle\vartheta_{j}|\vartheta
_{i}\rangle^{E}\mathop{\left<i\,\right|}\nolimits\sigma
\mathop{\left|j\right>}^{A^{\prime}}\nolimits|i\rangle\langle j|^{B}.\nonumber
\end{equation}
If the states $\left\vert \vartheta_{i}\right\rangle ^{E}$ are orthonormal,
the channel is a completely dephasing channel $\Delta^{A^{\prime}\rightarrow
B}$:%
\begin{equation}
\Delta^{A^{\prime}\rightarrow B}(\sigma)\equiv\sum_{i}|i\rangle^{B}\langle
i|^{A^{\prime}}\sigma|i\rangle^{A^{\prime}}\langle i|^{B}.\nonumber
\end{equation}

We obtain the complementary channel of a generalized dephasing channel by
tracing over Bob's system $B$:%
\[
\text{Tr}_{B}\left\{  U_{\mathcal{N}_{\text{GD}}}\sigma U_{\mathcal{N}%
_{\text{GD}}}^{\dag}\right\}  =\sum_{i}\left\langle i\right\vert ^{A^{\prime}%
}\sigma\left\vert i\right\rangle ^{A^{\prime}}\left\vert \vartheta
_{i}\right\rangle \left\langle \vartheta_{i}\right\vert ^{E},
\]
which we recognize to be an entanglement-breaking channel of the form in
(\ref{eq:ent-break-form}). Thus, a generalized dephasing channel is a Hadamard
channel because its complementary channel is entanglement-breaking. In fact,
the isometric extension in~(\ref{eq:iso-gen-dephasing}) of the generalized
dephasing channel appears remarkably similar to the isometric extension
in~(\ref{eq:iso-Hadamard}) of the Hadamard channel, with the exception that
the states $\{\left\vert i\right\rangle ^{A^{\prime}}\}$ of the generalized
dephasing channel form an orthonormal basis.

A completely dephasing channel $\Delta$\ commutes with a generalized dephasing
channel $\mathcal{N}_{\text{GD}}$ because%
\[
(\mathcal{N}_{\text{GD}}\circ\Delta)(\sigma)=(\Delta\circ\mathcal{N}%
_{\text{GD}})(\sigma)=\sum_{i}\left\langle i\right\vert \sigma\left\vert
i\right\rangle ^{A^{\prime}}|i\rangle\langle i|^{B}.
\]
The property $\mathcal{N}_{\text{GD}}^{c}=\mathcal{N}_{\text{GD}}^{c}%
\circ\Delta$ also holds for the complementary channel:%
\[
\mathcal{N}_{\text{GD}}^{c}(\sigma)=(\mathcal{N}_{\text{GD}}^{c}\circ
\Delta)(\sigma)=\sum_{i}\mathop{\left<i\,\right|}\nolimits\rho
\mathop{\left|i\right>}\nolimits^{A^{\prime}}|\vartheta_{i}\rangle
\!\langle\vartheta_{i}|^{E}.
\]

The simplest example of a generalized dephasing channel is a qubit dephasing
channel. The action of the $p$-dephasing qubit channel is%
\begin{equation}
\mathcal{N}\left(  \sigma\right)  =(1-p)\sigma+p\Delta\left(  \sigma\right)  ,
\label{eq:qubit-dephasing}%
\end{equation}
where $\Delta$ in this case is%
\begin{equation}
\Delta(\sigma)=\frac{1}{2}\left(  \sigma+Z\sigma Z\right)  ,\nonumber
\end{equation}
and $Z$ is the Pauli matrix $\sigma_{Z}$. Hence an isometric extension
$U_{\mathcal{N}}^{A^{\prime}\rightarrow BE}$ of the qubit dephasing channel
has the form:%
\[
U_{\mathcal{N}}^{A^{\prime}\rightarrow BE}=\sqrt{1-\frac{p}{2}}\ I\otimes
\left\vert 0\right\rangle ^{E}+\sqrt{\frac{p}{2}}\ Z\otimes\left\vert
1\right\rangle ^{E},
\]
where $I$ is the identity operator. Therefore, the following is a
complementary channel $\mathcal{N}^{c}$\ of a qubit dephasing channel:%
\begin{multline*}
\mathcal{N}^{c}\left(  \sigma\right)  =\frac{p}{2}\left\vert 0\right\rangle
\left\langle 0\right\vert ^{E}+\left(  1-\frac{p}{2}\right)  \left\vert
1\right\rangle \left\langle 1\right\vert ^{E}\\
+\sqrt{\left(  1-\frac{p}{2}\right)  \frac{p}{2}}\text{Tr}\left\{  \sigma
Z\right\}  \left(  \left\vert 0\right\rangle \left\langle 1\right\vert
^{E}+\left\vert 1\right\rangle \left\langle 0\right\vert ^{E}\right)  ,
\end{multline*}
and we observe that a bit flip on the input state does not change the
eigenvalues of the resulting environment output state.

\subsubsection{Cloning Channels}

\label{sec:cloning}The no-cloning theorem forbids the cloning of arbitrary
quantum states \cite{nat1982}. However, nothing prevents one from performing
approximate cloning provided the fidelity of cloning is not too high. A
universal $1\rightarrow N$ cloner is a device that approximately copies the
input with maximal copy fidelity independent of the input state \cite{GM97}.
We refer to such a device as a cloning channel \cite{B09,BDHM09}. Such a
decoherence process occurs naturally during stimulated
emission~\cite{MH82,SWZ00,LSHB02}.

We focus on $1\rightarrow N$ qubit cloning channels where a single qubit
serves as the input, and the output is $N$ identical approximate copies on $N$
respective qubit systems. These universal cloners are unitarily
covariant~\cite{H02}, in the sense that a\ unitary $V$ on the input qubit maps
to an irreducible representation $R_{V}$ of $V$ on the output state:%
\begin{equation}
\mathcal{N}(V\rho V)=R_{V}\mathcal{N}(\rho)R_{V}^{\dagger}.\nonumber
\end{equation}

We present the action of an isometric extension of a $1\rightarrow N$ cloning
channel on a basis
$\{\mathop{\left|0\right>}\nolimits,\mathop{\left|1\right>}\nolimits\}$ for a
qubit input system $A^{\prime}$. Let $\mathop{\left|j\right>}\nolimits^{B}$ be
an orthonormal basis of normalized completely symmetric states for the output
system $B$ that consists of $N$ qubits:%
\[
\{\mathop{\left|j\right>}\nolimits^{B}\equiv
\mathop{\left|N-j, j\right>}\nolimits\}_{j=0}^{N},
\]
where $\left\vert N-j,j\right\rangle ^{B}$ denotes a normalized state on an
$N$-qubit system that is a uniform superposition of computational basis states
with $N-j$ \textquotedblleft zeros\textquotedblright\ and $j$
\textquotedblleft ones.\textquotedblright\ Let
$\mathop{\left|i\right>}\nolimits^{E}$ be an orthonormal basis for the
environment$~E$:%
\[
\{\mathop{\left|i\right>}\nolimits^{E}\equiv
\mathop{\left|N-i-1, i\right>}\nolimits\}_{i=0}^{N-1},
\]
where $\left\vert N-i-1,i\right\rangle ^{E}$ denotes a normalized state on an
$\left(  N-1\right)  $-qubit system that is a uniform superposition of
computational basis states with $N-i-1$ \textquotedblleft
zeros\textquotedblright\ and $i$ \textquotedblleft ones.\textquotedblright%
\ Then an isometric extension $U_{\mathcal{N}_{\text{Cl}}}^{A^{\prime
}\rightarrow BE}$ of the $1\rightarrow N$ cloning channel $\mathcal{N}$ has
the form:%
\begin{multline}
U_{\mathcal{N}_{\text{Cl}}}^{A^{\prime}\rightarrow BE}\equiv\frac{1}%
{\sqrt{\Delta_{N}}}\sum_{i=0}^{N-1}\sqrt{N-i}%
\mathop{\left|i\right>}\nolimits^{B}\left\langle 0\right\vert ^{A^{\prime}%
}\otimes\mathop{\left|i\right>}\nolimits^{E}\label{eq:ucl}\\
+\frac{1}{\sqrt{\Delta_{N}}}\sum_{i=0}^{N-1}\sqrt{i+1}%
\mathop{\left|i+1\right>}\nolimits^{B}\left\langle 1\right\vert ^{A^{\prime}%
}\otimes\mathop{\left|i\right>}\nolimits^{E},
\end{multline}
where $\Delta_{N}\equiv N\left(  N+1\right)  /2$. A set of Kraus operators for
the channel $\mathcal{N}_{\text{Cl}}$ is as follows:%
\[
\left\{  \frac{1}{\sqrt{\Delta_{N}}}\left(  \sqrt{N-i}%
\mathop{\left|i\right>}\nolimits^{B}\left\langle 0\right\vert ^{A^{\prime}%
}+\sqrt{i+1}\mathop{\left|i+1\right>}\nolimits^{B}\left\langle 1\right\vert
^{A^{\prime}}\right)  \right\}  _{i=0}^{N-1},
\]
and a set of Kraus operators for the complementary channel $\mathcal{N}%
_{\text{Cl}}^{c}$ is as follows:%
\begin{align*}
&  \sqrt{N}\left\vert 0\right\rangle ^{E}\left\langle 0\right\vert
^{A^{\prime}},\\
&  \left\{  \sqrt{N-i}\left\vert i\right\rangle ^{E}\left\langle 0\right\vert
^{A^{\prime}}+\sqrt{i}\left\vert i-1\right\rangle ^{E}\left\langle
1\right\vert ^{A^{\prime}}\right\}  _{i=1}^{N-1},\\
&  \sqrt{N}\left\vert N-1\right\rangle ^{E}\left\langle 1\right\vert
^{A^{\prime}}.
\end{align*}
We can rewrite the Kraus operators for the complementary channel of a
$1\rightarrow2$ cloning channel as follows:%
\[
\left\{
\begin{array}
[c]{c}%
\sqrt{\frac{1}{3}}\left\vert +\right\rangle \left\langle +\right\vert
,\sqrt{\frac{1}{3}}\left\vert -\right\rangle \left\langle -\right\vert
,\sqrt{\frac{1}{3}}\left\vert 0\right\rangle \left\langle 0\right\vert
,\sqrt{\frac{1}{3}}\left\vert 1\right\rangle \left\langle 1\right\vert ,\\
\sqrt{\frac{1}{3}}\left\vert +_{Y}\right\rangle \left\langle +_{Y}\right\vert
\sigma_{Z},\sqrt{\frac{1}{3}}\left\vert -_{Y}\right\rangle \left\langle
-_{Y}\right\vert \sigma_{Z}%
\end{array}
\right\}  ,
\]
where $\left\vert +\right\rangle $ and $\left\vert -\right\rangle $ are the
eigenstates of the Pauli $X$ matrix and $\left\vert +_{Y}\right\rangle $ and
$\left\vert -_{Y}\right\rangle $ are the eigenstates of the Pauli $Y$ matrix.
The representation of the complementary channel with unit rank Kraus operators
explicitly demonstrates that it is entanglement-breaking (as first observed
with a positive partial transpose argument in Ref.~\cite{B09}), and the other
arguments in Ref.~\cite{B09} demonstrate that every $1\rightarrow N$ cloning
channel is entanglement-breaking. Therefore, a $1\rightarrow N$ cloning
channel is in the class of Hadamard channels.

\subsubsection{Unruh channels}

The Unruh channel is a natural channel to consider in the context of quantum
field theory \cite{BHP09}. The output of the Unruh channel is the quantum
state detected by a uniformly accelerated observer when the input is a
dual-rail photonic qubit prepared by a Minkowski observer. The mathematical
structure of the output of an Unruh channel $\mathcal{N}_{\text{U}}$ is an
infinite-dimensional block-diagonal density matrix, where the $N^{\text{th}}$
block is an instance of a $1\rightarrow N$ cloning channel:%
\begin{equation}
\mathcal{N}_{\text{U}}(\sigma)\equiv\bigoplus_{l=2}^{\infty}p_{l}\left(
z\right)  S_{l}(\sigma), \label{eq:Unruh-channel}%
\end{equation}
where%
\[
p_{l}\left(  z\right)  \equiv(1-z)^{3}z^{l-2}\Delta_{l-1},
\]
the \textquotedblleft acceleration parameter\textquotedblright\ $z\in
\lbrack0,1)$ is a strictly increasing function of acceleration, $\Delta
_{l-1}\equiv(l-1)l/2$, and $S_{l}$ is the output of a $1\rightarrow(l-1)$
cloning channel. The complementary channel $\mathcal{N}_{\text{U}}^{c}$ is
similar, with $S_{l}$ replaced by $S_{l}^{c}$, the complementary channel of a
$1\rightarrow(l-1)$ cloning channel.

\section{Review of Capacity Regions}

\label{sec:region-review}We review several trade-off capacity regions that
have appeared in the quantum Shannon theory literature. We first review the
Devetak-Shor result for the classically-enhanced quantum (CQ) capacity region
\cite{cmp2005dev}. We then review Shor's results on the entanglement-assisted
classical (CE) capacity region \cite{arx2004shor}, followed by a review of the
natural generalization to the triple trade-off region for
entanglement-assisted classical and quantum (CQE) coding \cite{HW08a}.

\subsection{Classically-Enhanced Quantum Capacity Region}

\label{sec:devshor}Consider a protocol that exploits a noisy quantum
channel$~\mathcal{N}$\ to transmit both classical and quantum information. The
goal of such a protocol is to transmit as much of these resources
as possible with vanishing error probability and fidelity approaching unity in
the limit of a large number of uses of the channel$~\mathcal{N}$. More
precisely, the protocol transmits one classical message from a set of $M$
messages and an arbitrary quantum state of dimension $K$ using a large number
$n$ uses of the quantum channel. The classical rate of transmission is
$C\equiv\frac{\log\left(  M\right)  }{n}$ bits per channel use, and the
quantum rate of transmission is $Q\equiv\frac{\log\left(  K\right)  }{n}$
qubits per channel use. If there is a scheme that transmits classical data at
rate $C$ with vanishing error probability and quantum data at rate $Q$ with
fidelity approaching unity in the limit of a large number of uses of the
channel, we say that the rates $C$ and $Q$ form an achievable rate pair
$\left(  C,Q\right)  $.

Devetak and Shor's main result in Ref.~\cite{cmp2005dev} is that all
achievable rate pairs lie in the following classically-enhanced quantum (CQ)
capacity region of $\mathcal{N}$:%
\begin{equation}
\mathcal{C}_{\text{CQ}}(\mathcal{N})\equiv\overline{\bigcup_{k=1}^{\infty
}\frac{1}{k}\mathcal{C}_{\text{CQ}}^{(1)}(\mathcal{N}^{\otimes k})},
\label{eq:multi-letter-CEQ}%
\end{equation}
where $\overline{Z}$ is the closure of a set $Z$, and the \textquotedblleft
one-shot\textquotedblright\ region $\mathcal{C}^{(1)}(\mathcal{N})$ is as
follows:%
\begin{equation}
\mathcal{C}_{\text{CQ}}^{(1)}(\mathcal{N})\equiv\bigcup_{\rho}\mathcal{C}%
_{\text{CQ,}\rho}^{(1)}(\mathcal{N}).\nonumber
\end{equation}
The \textquotedblleft one-shot, one-state\textquotedblright\ region
$\mathcal{C}_{\text{CQ,}\rho}^{(1)}(\mathcal{N})$ is the set of all $C,Q\geq0$
such that%
\begin{align}
C  &  \leq I(X;B)_{\rho},\nonumber\\
Q  &  \leq I(A\rangle BX)_{\rho},\nonumber
\end{align}
where $\rho$ is a classical-quantum state of the following form:%
\begin{equation}
\rho^{XABE}\equiv\sum_{x}p_{X}\left(  x\right)  \left\vert x\right\rangle
\left\langle x\right\vert ^{X}\otimes U_{\mathcal{N}}^{A^{\prime}\rightarrow
BE}(\phi_{x}^{AA^{\prime}}), \label{eq:CEQ-state}%
\end{equation}
the states $\phi_{x}^{AA^{\prime}}$ are pure, and the dimension of the classical
system is finite~\cite{cmp2005dev}. For general channels, the
multi-letter characterization in (\ref{eq:multi-letter-CEQ}) is necessary, but
for certain channels, such as erasure channels and generalized dephasing
channels, the CQ\ capacity region admits a single-letter
characterization~\cite{cmp2005dev}. In
Section~\ref{sec:single-letter-CEQ-Hadamard}, we show that the CQ\ region for
all Hadamard channels admits a single-letter characterization.

\subsection{Entanglement-Assisted Classical Capacity Region}

Now consider a protocol that exploits shared noiseless entanglement and a
noisy quantum channel$~\mathcal{N}$\ to transmit classical information. The
goal of such a protocol is to transmit as much classical information as
possible with vanishing error probability while consuming as little
entanglement as possible in the limit of a large number of uses of the
channel$~\mathcal{N}$. More precisely, the protocol transmits one classical
message from a set of $M$ messages using a large number $n$ uses of the
quantum channel and a maximally entangled state $\Phi^{T_{A}T_{B}}$ of the
form in (\ref{eq:max-entangled}) where the sender possesses the system $T_{A}%
$, the receiver possesses the system $T_{B}$, and $D$ is the Schmidt rank of
the entanglement. The classical rate of transmission is $C\equiv\frac
{\log\left(  M\right)  }{n}$ bits per channel use, and the rate of
entanglement consumption is $E\equiv\frac{\log\left(  D\right)  }{n}$ ebits
per channel use. If there is a scheme that transmits classical data at rate
$C$ with vanishing error probability and consumes entanglement at rate $E$ in
the limit of a large number of uses of the channel, we say that the rates $C$
and $E$ form an achievable rate pair $\left(  C,E\right)  $.

Shor's main result in Ref.~\cite{arx2004shor} is that all achievable rate
pairs lie in the following entanglement-assisted classical capacity region of
$\mathcal{N}$:%
\begin{equation}
\mathcal{C}_{\text{CE}}(\mathcal{N})=\overline{\bigcup_{k=1}^{\infty}\frac
{1}{k}\mathcal{C}_{\text{CE}}^{(1)}(\mathcal{N}^{\otimes k})},
\label{eq:multi-letter-EAC}%
\end{equation}
where the \textquotedblleft one-shot\textquotedblright\ region $\mathcal{C}%
_{\text{CE}}^{(1)}(\mathcal{N})$ is a union of the \textquotedblleft one-shot,
one-state\textquotedblright\ regions $\mathcal{C}_{\text{CE,}\rho}%
^{(1)}(\mathcal{N})$. The one-shot, one-state region $\mathcal{C}%
_{\text{CE,}\rho}^{(1)}(\mathcal{N})$ is the set of all $C,E\geq0$ such that%
\begin{align}
C  &  \leq I(AX;B)_{\rho},\nonumber\\
E  &  \geq H\left(  A|X\right)  _{\rho},\nonumber
\end{align}
where $\rho$ is again a state of the form in (\ref{eq:CEQ-state}).

Hsieh and Wilde later gave a more refined \textquotedblleft
trapezoidal\textquotedblright\ characterization of the one-shot, one-state
region $\mathcal{C}_{\text{CE,}\rho}^{(1)}(\mathcal{N})$ \cite{HW08a}, but the
one-shot regions $\mathcal{C}_{\text{CE}}^{(1)}(\mathcal{N})$\ resulting from
both Shor's \textquotedblleft rectangular\textquotedblright\ characterization
and Hsieh and Wilde's pentagonal characterization are equivalent, and the
characterization above suffices for our purposes here.

For general channels, the multi-letter characterization in
(\ref{eq:multi-letter-EAC}) is necessary, but for certain channels, such as
the erasure channels and generalized dephasing channels, the CE\ capacity
region admits a single-letter characterization as Ref.~\cite{HW08a} stated and
as we show in full detail here for the generalized dephasing channels. In
fact, in Section~\ref{sec:single-letter-EAC-Hadamard}, we show that the
CE\ region for all Hadamard channels, of which a generalized dephasing channel
is an example, admits a single-letter characterization.

\subsection{The Capacity Region for Entanglement-Assisted Transmission of
Classical and Quantum Information}

\label{sec:EACQ-review}The natural generalization of the two scenarios we have
just considered is entanglement-assisted classical and quantum
(CQE)\ communication. The goal of such a protocol is to transmit as much
classical information with vanishing error probability and quantum information
with fidelity approaching unity while consuming as little entanglement as
possible in the limit of a large number of uses of the channel$~\mathcal{N}$.
More precisely, the protocol transmits one classical message from a set of $M$
messages and an arbitrary quantum state of dimension $K$ using a large number
$n$ uses of the quantum channel and a maximally entangled state $\Phi
^{T_{A}T_{B}}$ of dimension $D$. The classical rate of transmission is
$C\equiv\frac{\log\left(  M\right)  }{n}$ bits per channel use, the quantum
rate of transmission is $Q\equiv\frac{\log\left(  K\right)  }{n}$ qubits per
channel use, and the rate of entanglement consumption is $E\equiv\frac
{\log\left(  D\right)  }{n}$ ebits per channel use. If there is a scheme that
transmits classical data, transmits quantum data, and consumes entanglement at
respective rates $C$, $Q$, and $E$ with vanishing error probability and
fidelity approaching unity in the limit of a large number of uses of the
channel, we say that the rates $C$, $Q$, and $E$ form an achievable rate
triple $\left(  C,Q,E\right)  $.

Hsieh and Wilde's main result in Ref.~\cite{HW08a} is that all achievable rate
triples lie in the following CQE capacity region for the channel $\mathcal{N}%
$:%
\[
\mathcal{C}_{\text{CQE}}(\mathcal{N})=\overline{\bigcup_{k=1}^{\infty}\frac
{1}{k}\mathcal{C}_{\text{CQE}}^{(1)}(\mathcal{N}^{\otimes k})},
\]
where the one-shot region $\mathcal{C}_{\text{CQE}}^{(1)}(\mathcal{N})$ is the
union of the \textquotedblleft one-shot, one-state\textquotedblright\ regions
$\mathcal{C}_{\text{CQE,}\rho}^{(1)}(\mathcal{N})$. The \textquotedblleft
one-shot, one-state\textquotedblright\ region $\mathcal{C}_{\text{CQE,}\rho
}^{(1)}(\mathcal{N})$ is the set of all $C,Q,E\geq0$ such that%
\begin{align}
C+2Q  &  \leq I(AX;B)_{\rho},\label{eq:EAC-plane}\\
Q  &  \leq I(A\rangle BX)_{\rho}+E,\\
C+Q  &  \leq I(X;B)_{\rho}+I(A\rangle BX)_{\rho}+E.\label{eq:EACQ-4}
\end{align}

One of the observations in Ref.~\cite{HW08a}\ is that there is a way to
single-letterize the CQE\ capacity region provided one can show that both the
CQ and CE\ trade-off curves single-letterize. We review these arguments
briefly. Suppose that we have shown that both the CQ and CE\ trade-off curves
single-letterize. Then three surfaces specify the boundary of the
CQE\ capacity region so that we can simplify the characterization in
(\ref{eq:EAC-plane}-\ref{eq:EACQ-4}). The first surface to consider is that
formed by combining the CE\ trade-off curve with the \textquotedblleft
inverse\textquotedblright\ of the super-dense coding protocol. Recall that the
super-dense coding protocol exploits a noiseless ebit and a noiseless qubit
channel to transmit two classical bits~\cite{PhysRevLett.69.2881}. Let
$(C_{\text{CE}}\left(  s_1\right)  ,0,E_{\text{CE}}\left(  s_1\right)  )$
denote a parametrization of all points on the CE trade-off curve with respect
to some parameter $s_1\in\left[  0,1/2\right]  $, and recall that each point
on the trade-off curve has corresponding entropic quantities of the form
$\left(  I\left(  AX;B\right)  ,0,H\left(  A|X\right)  \right)  $. Then the
surface formed by combining the CE\ trade-off curve with the inverse of
super-dense coding is%
\[
\left\{
\begin{array}
[c]{c}%
(C_{\text{CE}}\left(  s_1\right)  -2Q,Q,E_{\text{CE}}\left(  s_1\right)
-Q):s_1\in\left[  0,1/2\right]  ,\\
0\leq Q\leq\min\left\{  \frac{1}{2}C_{\text{CE}}\left(  s_1\right)
,E_{\text{CE}}\left(  s_1\right)  \right\}
\end{array}
\right\}  .
\]
This surface forms an outer bound for the capacity region. Were it not so,
then one could combine points outside it with super-dense coding to outperform
points on the CE\ trade-off curve, contradicting the optimality of this
trade-off curve.

The next surface to consider is that formed by combining the CQ\ trade-off
curve with the \textquotedblleft inverse\textquotedblright\ of the
entanglement distribution protocol. Recall that the entanglement distribution
protocol exploits a noiseless qubit channel to establish a shared noiseless
ebit. Let $(C_{\text{CQ}}\left(  s_2\right)  ,Q_{\text{CQ}}\left(  s_2\right)
,0)$ denote a parametrization of all points on the CQ trade-off curve with
respect to some parameter $s_2\in\left[  0,1/2\right]  $, and recall that each
point on the trade-off curve has corresponding entropic quantities of the form
$\left(  I\left(  X;B\right)  ,I\left(  A\rangle BX\right)  ,0\right)  $. Then
the surface formed by combining the CQ trade-off curve with the inverse of
entanglement distribution is%
\[
\{(C_{\text{CQ}}\left(  s_2\right)  ,Q_{\text{CQ}}\left(  s_2\right)
+E,E):s_2\in\left[  0,1/2\right]  ,\ \ E\geq0\}.
\]
This surface also forms an outer bound for the capacity region. Were it not
so, then one could combine points outside it with entanglement distribution to
outperform points on the CQ trade-off curve, contradicting the optimality of
this trade-off curve.

The final surface to consider is the following regularization of the plane
that (\ref{eq:EAC-plane}) specifies:%
\begin{equation}
C+2Q\leq\frac{1}{n}h\left(  \mathcal{N}^{\otimes n}\right)  ,
\label{eq:EA-bound}%
\end{equation}
for all $n\geq1$, where%
\[
h\left(  \mathcal{N}\right)  \equiv\max_{\rho}I\left(  AX;B\right)  ,
\]
and $\rho$ is a state of the form in (\ref{eq:CEQ-state}).
Lemma~\ref{thm:EA-plane-single-letter}\ below states that $h\left(
\mathcal{N}^{\otimes n}\right)  $ actually single-letterizes for any quantum
channel $\mathcal{N}$:%
\[
h\left(  \mathcal{N}^{\otimes n}\right)  =nh\left(  \mathcal{N}\right)  ,
\]
so that the computation of the boundary $h\left(  \mathcal{N}^{\otimes
n}\right)  /n$ is tractable. Its proof is a consequence of the
single-letterization of the entanglement-assisted classical
capacity~\cite{ieee2002bennett}, but we provide it in
Appendix~\ref{sec:proof-EAC-plane} for completeness.

\begin{lemma}
\label{thm:EA-plane-single-letter}The plane in (\ref{eq:EA-bound})\ admits a
single-letter characterization for any noisy quantum channel $\mathcal{N}$:%
\[
h\left(  \mathcal{N}^{\otimes n}\right)  =nh\left(  \mathcal{N}\right)  .
\]

\end{lemma}

The above three surfaces all form outer bounds on the CQE capacity region, but
is it clear that we can achieve points along the boundaries? To answer this
question, we should consider the intersection of the first and second
surfaces, found by solving the following equation for $Q$ and $E$:%
\begin{multline*}
(C_{\text{CE}}\left(  s_1\right)  -2Q,Q,E_{\text{CE}}\left(  s_1\right)  -Q)\\
=(C_{\text{CQ}}\left(  s_2\right)  ,Q_{\text{CQ}}\left(  s_2\right)  +E,E).
\end{multline*}
Using the entropic expressions for the trade-off curves and solving the above
equation gives that all points along the intersection have entropic quantities
of the following form:%
\[
\left(  I\left(  X;B\right)  ,\frac{1}{2}I\left(  A;B|X\right)  ,\frac{1}%
{2}I\left(  A;E|X\right)  \right)  .
\]
Ref.~\cite{HW08a}\ constructed a protocol, dubbed the \textquotedblleft
classically-enhanced father protocol,\textquotedblright\ that can achieve the
above rates for CQE communication. Thus, by combining this protocol with
super-dense coding and entanglement distribution, we can achieve all points
inside the first and second surfaces with entanglement consumption below a
certain rate. Finally, by combining this protocol with super-dense coding,
entanglement distribution, and the wasting of entanglement, we can achieve all
points that lie inside all three surfaces, and thus we can achieve the full
CQE\ capacity region. We summarize these results as the following theorem.

\begin{theorem}
Suppose the CQ and CE\ trade-off curves of a quantum channel $\mathcal{N}$
single-letterize. Then the full CQE\ capacity region of $\mathcal{N}$
single-letterizes:%
\[
\mathcal{C}_{\emph{CQE}}\left(  \mathcal{N}\right)  =\mathcal{C}_{\emph{CQE}%
}^{\left(  1\right)  }\left(  \mathcal{N}\right)  .
\]

\end{theorem}

We apply the above theorem in the next section. We first show that both the
CQ\ and CE trade-off curves single-letterize for all Hadamard channels, and it
then follows that the full CQE\ capacity region single-letterizes for these channels.

\section{Single-Letterization of the CQ and CE\ Trade-off Curves for Hadamard
Channels}

\label{sec:single-letter-CEQ-EAC-Hadamard}

\subsection{CQ\ Trade-off Curve}

\label{sec:single-letter-CEQ-Hadamard}For the CQ\ region, we would like to
maximize both the classical and quantum communication rates, but we cannot
have both be simultaneously optimal. Thus, we must trade between these
resources. If we are willing to reduce the quantum communication rate by a
little, then we can communicate more classical information and vice versa.

Our main theorem below states and proves that the following function generates
points on the CQ trade-off curve for Hadamard channels:%
\begin{equation}
f_{\lambda}\left(  \mathcal{N}\right)  \equiv\max_{\rho}I\left(  X;B\right)
_{\rho}+\lambda I\left(  A\rangle BX\right)  _{\rho},
\label{eq:CEQ-max-function}%
\end{equation}
where the state $\rho$ is of the form in (\ref{eq:CEQ-state}) and $\lambda
\geq1$.

\begin{theorem}
\label{thm:CQ-single-letter}For any fixed $\lambda\geq1$, the function in
(\ref{eq:CEQ-max-function}) leads to a point%
\[
(I(X;B)_{\rho},I\left(  A\rangle BX\right)  _{\rho})
\]
on the CQ trade-off curve, provided $\rho$ maximizes (\ref{eq:CEQ-max-function}).
\end{theorem}

\begin{proof}
This theorem follows from the results of Lemmas~\ref{claim:Lmult},
\ref{lem:max-CQ-lambda}, and \ref{lem:CEQ-base-case} and
Corollary~\ref{thm:CEQ-single-letter} below.
\end{proof}

\begin{lemma}
\label{claim:Lmult}For any fixed $\lambda\geq0$, the function in
(\ref{eq:CEQ-max-function}) leads to a point $(I(X;B)_{\rho},I\left(  A\rangle
BX\right)  _{\rho})$ on the one-shot CQ trade-off curve in the sense of Theorem~\ref{thm:CQ-single-letter}.
\end{lemma}

\begin{proof}
We argue by contradiction. Suppose that a particular state $\rho^{XAB}$ of the
form in (\ref{eq:CEQ-state}) maximizes (\ref{eq:CEQ-max-function}).\ Suppose
further that it does not lead to a point on the trade-off curve. That is,
given the constraint $Q=I(A\rangle BX)_{\rho}$, there is some other state
$\sigma^{XAB}$ of the form in (\ref{eq:CEQ-state}) such that $I(A\rangle
BX)_{\sigma}=I(A\rangle BX)_{\rho}=Q$ where $C=I(X;B)_{\sigma}>I(X;B)_{\rho}$.
But this result implies the following inequality for all $\lambda\geq0$:%
\[
I(X;B)_{\sigma}+\lambda I(A\rangle BX)_{\sigma}>I(X;B)_{\rho}+\lambda
I(A\rangle BX)_{\rho},
\]
contradicting the fact that the state $\rho^{XAB}$ maximizes
(\ref{eq:CEQ-max-function}).
\end{proof}

\begin{lemma}
\label{lem:max-CQ-lambda}We obtain all points on the one-shot CQ trade-off
curve by considering $\lambda\geq1$ in the maximization in
(\ref{eq:CEQ-max-function}) because the maximization optimizes only the
classical capacity for all $\lambda$ such that $0\leq\lambda<1$.
\end{lemma}

\begin{proof}
Consider a state $\rho^{XABE}$ of the form in (\ref{eq:CEQ-state}). Suppose
that we perform a von Neumann measurement of the system $A$, resulting in a
classical variable$~Y$, and let $\sigma^{XYBE}$ denote the resulting state.
Then the following chain of inequalities holds for all $\lambda$ such that
$0\leq\lambda<1$:%
\begin{align*}
&  I(X;B)_{\rho}+\lambda I(A\rangle BX)_{\rho}\\
&  =H(B)_{\rho}+(\lambda-1)H(B|X)_{\rho}-\lambda H(E|X)_{\rho}\\
&  =H(B)_{\sigma}+(\lambda-1)H(B|X)_{\sigma}-\lambda H(E|X)_{\sigma}\\
&  \leq H(B)_{\sigma}+(\lambda-1)H(B|XY)_{\sigma}-\lambda H(E|XY)_{\sigma}.\\
&  =H(B)_{\sigma}+(\lambda-1)H(B|XY)_{\sigma}-\lambda H(B|XY)_{\sigma}\\
&  =H\left(  B\right)  _{\sigma}-H\left(  B|XY\right)  _{\sigma}\\
&  =I\left(  XY;B\right)  _{\sigma}.
\end{align*}
The first equality follows from definitions and because the state $\rho
^{XABE}$\ on systems $A$, $B$, and $E$ is pure when conditioned on the
classical variable $X$. The second equality follows because the von Neumann
measurement does not affect the systems involved in the entropic expressions.
The inequality follows because $0\leq\lambda<1$ and conditioning reduces
entropy. The third equality follows because the reduced state of
$\sigma^{XYBE}$ on systems $B$ and $E$ is pure when conditioned on both$~X$
and$~Y$. The last two equalities follow from algebra and the definition of the
quantum mutual information.

Thus, it becomes clear that the maximization of the original quantity for
$0\leq\lambda<1$ is always less than the classical capacity because$~I\left(
XY;B\right)  _{\sigma}\leq\max_{\rho}I\left(  X;B\right)  $. It then follows
that the trade-off curve really starts when $\lambda\geq1$.
\end{proof}

We remark that the above proof gives an alternate mathemathical justification for
considering only $\lambda\geq1$ in the maximization than does the original
operational reason given in
Ref.~\cite{cmp2005dev}.

The following lemma is the crucial one that leads to our main result in this
section:\ the single-letterization of the CQ trade-off curve for Hadamard channels.

\begin{lemma}
\label{lem:CEQ-base-case}The following additivity relation holds for a
Hadamard channel $\mathcal{N}_{1}$ and any other channel $\mathcal{N}_{2}$:%
\[
f_{\lambda}(\mathcal{N}_{1}\otimes\mathcal{N}_{2})=f_{\lambda}(\mathcal{N}%
_{1})+f_{\lambda}(\mathcal{N}_{2}).
\]

\end{lemma}

\begin{proof}
The inequality $f_{\lambda}(\mathcal{N}_{1}\otimes\mathcal{N}_{2})\geq
f_{\lambda}(\mathcal{N}_{1})+f_{\lambda}(\mathcal{N}_{2})$ trivially holds for
all quantum channels, because the maximization on the RHS\ is a restriction of
the maximization in the LHS\ to a tensor product of states of the form in
(\ref{eq:CEQ-state}). Therefore, we prove the non-trivial inequality
$f_{\lambda}(\mathcal{N}_{1}\otimes\mathcal{N}_{2})\leq f_{\lambda
}(\mathcal{N}_{1})+f_{\lambda}(\mathcal{N}_{2})$ when $\mathcal{N}_{1}$ is a
Hadamard channel.

Suppose that $\mathcal{N}_{1}^{A_{1}\rightarrow B_{1}}$ is a Hadamard channel
and $\mathcal{N}_{2}^{A_{2}\rightarrow B_{2}}$ is any quantum channel with
respective complementary channels $\left(  \mathcal{N}_{1}^{c}\right)
^{A_{1}\rightarrow E_{1}}$ and $\left(  \mathcal{N}_{2}^{c}\right)
^{A_{2}\rightarrow E_{2}}$. Then $\mathcal{N}_{1}^{c}$ is
entanglement-breaking because $\mathcal{N}_{1}$ is a Hadamard channel.
Therefore, there exists a degrading map $\mathcal{D}^{B_{1}\rightarrow E_{1}}$
because $\mathcal{N}_{1}^{c}$ is entanglement-breaking. The degrading map
consists of a measurement with classical output on system $Y$ followed by the
preparation of a state on $E_{1}$ depending on the measurement outcome (as
described in Section~\ref{sec:Hadamard-channel}). We can therefore write
$\mathcal{D}^{B_{1}\rightarrow E_{1}}=\mathcal{D}_{2}^{Y\rightarrow E_{1}%
}\circ\mathcal{D}_{1}^{B_{1}\rightarrow Y}$.

Let
\begin{align}
\psi^{XAA_{1}A_{2}} &  \equiv\sum_{x}p_{X}\left(  x\right)  \proj{x}^{X}%
\ox\proj{\phi_x}^{AA_{1}A_{2}},\nonumber\\
\theta^{XAB_{1}E_{1}A_{2}} &  \equiv U_{1}\psi U_{1}^{\dag},\nonumber\\
\rho^{XAB_{1}E_{1}B_{2}E_{2}} &  \equiv(U_{1}\otimes U_{2})\psi(U_{1}\otimes
U_{2})^{\dag},\label{eq:double-CEQ-state}%
\end{align}
where $U_{j}{}^{A_{j}\rightarrow B_{j}E_{j}}$ is the isometric extension of
$\mathcal{N}_{j}$. Suppose further that $\rho$ is the state that maximizes
$f_{\lambda}(\mathcal{N}_{1}\otimes\mathcal{N}_{2})$. Also, let $\sigma
\equiv\mathcal{D}_{1}(\rho)$ and note that this state is of the following
form:%
\begin{multline*}
\sigma^{XYAB_{2}E_{1}E_{2}}\equiv\\
\sum_{x}p_{X}\left(  x\right)  p_{Y|X}\left(  y|x\right)  \proj{x}^{X}%
\otimes\proj{y}^{Y}\otimes\proj{\phi_{x,y}}^{AB_{2}E_{1}E_{2}}.
\end{multline*}

Then the following inequalities hold for all $\lambda\geq1$:%
%TCIMACRO{\TeXButton{TeX field}{\begin{widetext}}}%
%BeginExpansion
\begin{widetext}%
%EndExpansion%
\begin{align*}
f_{\lambda}(\mathcal{N}_{1}\otimes\mathcal{N}_{2}) &  =I(X;B_{1}B_{2})_{\rho
}+\lambda I(A\rangle B_{1}B_{2}X)_{\rho}\\
&  =H(B_{1}B_{2})_{\rho}+(\lambda-1)H(B_{1}B_{2}|X)_{\rho}-\lambda
H(E_{1}E_{2}|X)_{\rho}\\
&  \leq H(B_{1})_{\rho}+(\lambda-1)H(B_{1}|X)_{\rho}-\lambda H(E_{1}|X)_{\rho
}+H(B_{2})_{\rho}+(\lambda-1)H(B_{2}|B_{1}X)_{\rho}-\lambda H(E_{2}%
|E_{1}X)_{\rho}\\
&  \leq H(B_{1})_{\rho}+(\lambda-1)H(B_{1}|X)_{\rho}-\lambda H(E_{1}|X)_{\rho
}+H(B_{2})_{\sigma}+(\lambda-1)H(B_{2}|YX)_{\sigma}-\lambda H(E_{2}%
|YX)_{\sigma}\\
&  =\left[  I(X;B_{1})_{\theta}+\lambda I(AA_{2}\rangle B_{1}X)_{\theta
}\right]  +\left[  I(XY;B_{2})_{\sigma}+\lambda I(AE_{1}\rangle B_{2}%
XY)_{\sigma}\right]  \\
&  \leq f_{\lambda}(\mathcal{N}_{1})+f_{\lambda}(\mathcal{N}_{2}).
\end{align*}%
%TCIMACRO{\TeXButton{TeX field}{\end{widetext}}}%
%BeginExpansion
\end{widetext}%
%EndExpansion
The first equality follows by definition and the assumption that $\rho$
maximizes $f_{\lambda}(\mathcal{N}_{1}\otimes\mathcal{N}_{2})$. The second
equality follows from entropic manipulations and the fact that $H(AB_{1}%
B_{2}|X)=H(E_{1}E_{2}|X)$ for the state $\rho$. The first inequality follows
from subadditivity of entropy and the chain rule \cite{book2000mikeandike}.
The second inequality uses two applications of the monotonicity of conditional
entropy with respect to quantum channels acting on the conditioned system.
Specifically, $H(B_{2}|B_{1}X)_{\rho}\leq H(B_{2}|YX)_{\sigma}$ because of the
existence of the map $\mathcal{D}_{1}$ while $H(E_{2}|YX)_{\sigma}\leq
H(E_{2}|E_{1}X)_{\rho}$ because of the existence of the map $\mathcal{D}_{2}$.
The third equality follows because $H(E_{1}|X)_{\rho}=H(AA_{2}B_{1}%
|X)_{\theta}$ and $H(E_{2}|YX)_{\rho}=H(AE_{1}B_{2}|YX)_{\sigma}$. The final
inequality follows because $\theta$ and $\sigma$ are both states of the form
in (\ref{eq:CEQ-state}).
\end{proof}

\begin{corollary}
\label{thm:CEQ-single-letter}The one-shot CQ trade-off curve is equal to the
regularized CQ\ trade-off curve when the noisy quantum channel $\mathcal{N}$
is a Hadamard channel:%
\[
f_{\lambda}\left(  \mathcal{N}^{\otimes n}\right)  =nf_{\lambda}\left(
\mathcal{N}\right)  .
\]

\end{corollary}

\begin{proof}
We prove the result using induction on $n$. The base case for $n=1$ is
trivial. Suppose the result holds for $n$: $f_{\lambda}(\mathcal{N}^{\otimes
n})=nf_{\lambda}(\mathcal{N})$. The following chain of equalities then proves
the inductive step:%
\begin{align*}
f_{\lambda}(\mathcal{N}^{\otimes n+1})  &  =f_{\lambda}(\mathcal{N}%
\otimes\mathcal{N}^{\otimes n})\\
&  =f_{\lambda}(\mathcal{N})+f_{\lambda}(\mathcal{N}^{\otimes n})\\
&  =f_{\lambda}(\mathcal{N})+nf_{\lambda}(\mathcal{N}).
\end{align*}
The first equality follows by expanding the tensor product. The second
critical equality follows from the application of
Lemma~\ref{lem:CEQ-base-case} because $\mathcal{N}$ is a Hadamard channel. The
final equality follows from the induction hypothesis.
\end{proof}

There is one last point that we should address concerning the CQ trade-off curve.
There is the possibility in this trade-off problem that the parameter $\lambda$ 
does not parametrize all points on the trade-off curve, potentially
leading to a gap in the trade-off curve. We address this concern in Appendix~\ref{sec:continuity-trade-off-curve}
by proving that one gets the entire trade-off curve by varying $\lambda$ and taking the convex hull of the resulting points. A similar proof holds for the CE trade-off curve.

\subsection{CE\ Trade-off Curve}

\label{sec:single-letter-EAC-Hadamard}For the CE region, we would like to
maximize both the classical communication rate while minimizing the
entanglement consumption rate, but we cannot have both be simultaneously
optimal. Thus, we must trade between these resources. If we are willing to
reduce the classical communication rate by a little, then the protocol does
not require as much entanglement consumption.

Our main theorem below states that the following function generates points on
the CE trade-off curve:%
\begin{equation}
g_{\lambda}\left(  \mathcal{N}\right)  \equiv\max_{\rho}I\left(  AX;B\right)
_{\rho}-\lambda H\left(  A|X\right)  _{\rho}, \label{eq:EAC-max-function}%
\end{equation}
where the state $\rho$ is of the form in (\ref{eq:CEQ-state}) and
$0\leq\lambda\leq1$.

\begin{theorem}
\label{thm:CE-single-letter}For $0\leq\lambda\leq1$, the function in
(\ref{eq:EAC-max-function}) leads to a point%
\[
(I(AX;B)_{\rho},H\left(  A|X\right)  _{\rho})
\]
on the CE trade-off curve, provided $\rho$ maximizes (\ref{eq:EAC-max-function}).
\end{theorem}

\begin{proof}
The proof of this theorem proceeds similarly to that of
Theorem~\ref{thm:CQ-single-letter} in the previous section. It follows from
the results of Lemmas~\ref{claim:LShor} and \ref{lem:max-CE-lambda} and
Corollary~\ref{thm:EAC-single-letter} in Appendix~\ref{sec:CE-trade-off-curve}.
\end{proof}

\section{Parametrization of the Trade-off Curves}

\label{sec:param-curves}%

%TCIMACRO{\FRAME{ftbpFU}{7.0534in}{3.0943in}{0pt}{\Qcb{(Color online) Plot of
%(a) the CQ trade-off curve and (b) the CE trade-off curve for a $p$-dephasing
%qubit channel for $p=0$, $0.1$, $0.2$, \ldots, $0.9$, $1$. The trade-off
%curves for $p=0$ correspond to those of a noiseless qubit channel and are the
%rightmost trade-off curve in each plot. The trade-off curves for $p=1$
%correspond to those for a classical channel, and are the leftmost trade-off
%curves in each plot. Each trade-off curve between these two extremes beats a
%time-sharing strategy, but these two extremes do not beat time-sharing.}%
%}{\Qlb{fig:CQ-EC-dephasing}}{dephasing-plot.pdf}%
%width 7.0534in;  height 3.0943in;  depth 0pt;  original-width 12.7664in;
%original-height 5.5728in;  cropleft "0";  croptop "1";  cropright "1";
%cropbottom "0";  filename 'dephasing-plot.pdf';file-properties "XNPEU";}}}%
%BeginExpansion
\begin{figure*}
[ptb]
\begin{center}
\includegraphics[
natheight=5.572800in,
natwidth=12.766400in,
height=3.0943in,
width=7.0534in
]%
{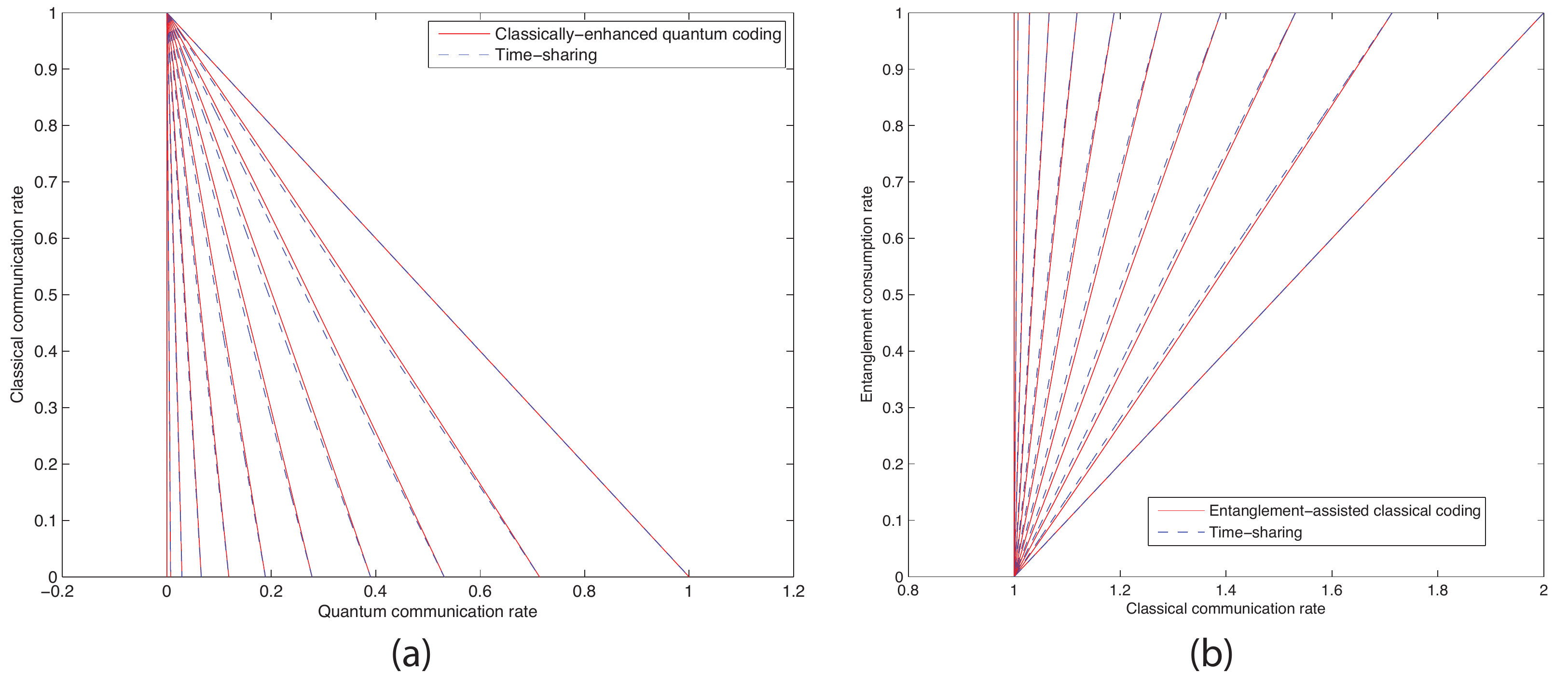}%
\caption{(Color online) Plot of (a) the CQ trade-off curve and (b) the CE
trade-off curve for a $p$-dephasing qubit channel for $p=0$, $0.1$, $0.2$,
\ldots, $0.9$, $1$. The trade-off curves for $p=0$ correspond to those of a
noiseless qubit channel and are the rightmost trade-off curve in each plot.
The trade-off curves for $p=1$ correspond to those for a classical channel,
and are the leftmost trade-off curves in each plot. Each trade-off curve
between these two extremes beats a time-sharing strategy, but these two
extremes do not beat time-sharing.}%
\label{fig:CQ-EC-dephasing}%
\end{center}
\end{figure*}
%EndExpansion

The results in the previous section demonstrate that the CQE capacity region
for all Hadamard channels single-letterizes. These results imply that we can
actually compute the CQE capacity region for these channels, by the arguments
in Section~\ref{sec:EACQ-review}. In this section, we consider several
instances of Hadamard channels and show how we can exactly characterize their
corresponding CQE capacity region. We first consider the qubit dephasing
channel and the $1\rightarrow N$ cloning channel, and then show how to apply
the results for the $1\rightarrow N$ cloning channel to the Unruh channel.

\subsection{Qubit Dephasing Channel}

This section briefly recalls the parametrizations for the CQ and CE\ trade-off
curves for a qubit dephasing channel. Devetak and Shor gave a particular
parametrization\ for the CQ trade-off curve for the case of a qubit dephasing
channel in Appendix~B of Ref.~\cite{cmp2005dev}, and Hsieh and Wilde followed
by giving a parametrization for the CE\ trade-off curve for qubit dephasing
channels~\cite{HW08a}. For the purposes of completion and comparison, we
recall these parametrizations below, and
Appendices~\ref{sec:CQ-dephasing-qubit}\ and
\ref{sec:CE-trade-off-qubit-dephasing} of this paper provide full proofs of
these assertions.

\begin{theorem}
\label{thm:dephasing-CEQ}All points on the CQ trade-off curve for a
$p$-dephasing qubit channel have the following form:%
\[
\left(  1-H_{2}(\mu),H_{2}(\mu)-H_{2}(\gamma\left(  \mu,p\right)  )\right)  ,
\]
and all points on the CE trade-off curve for a $p$-dephasing qubit channel
have the following form:%
\[
\left(  1+H_{2}(\mu)-H_{2}(\gamma\left(  \mu,p\right)  ),H_{2}(\mu)\right)  ,
\]
where $\mu\in\left[  0,1/2\right]  $, $H_{2}$ is the binary entropy function,
and%
\[
\gamma\left(  \mu,p\right)  \equiv\frac{1}{2}+\frac{1}{2}\sqrt{1-16\cdot
\frac{p}{2}\left(  1-\frac{p}{2}\right)  \mu(1-\mu)}.
\]

\end{theorem}

Figure~\ref{fig:CQ-EC-dephasing}\ plots both the CQ\ and CE trade-off curves
for several $p$-dephasing qubit channels, where $p$\ varies from zero to one
in increments of $1/10$. The figure demonstrates that both
classically-enhanced quantum coding and entanglement-assisted classical coding
beat a time-sharing strategy for any $p$ such that $0<p<1$.

We do not plot the full CQE\ capacity region for the dephasing qubit channel
but instead point the interested reader to Figure~4\ of Ref.~\cite{HW08a}.

\subsection{$1\rightarrow N$ Cloning Channels}

\label{sec:dscl}We now compute exact expressions for the CQ\ and CE\ trade-off
curves, plot them for several $1\rightarrow N$ cloning channels, and plot the
CQE\ capacity region for a particular $1\rightarrow N$ cloning channel. The
following theorem states these expressions, and the lemmas and subsequent calculation following it
provide a proof.

\begin{theorem}
\label{thm:cloning-CEQ}All points on the CQ trade-off curve for a
$1\rightarrow N$ cloning channel have the following form:%
\[
\left(  \log\left(  {N+1}\right)  -H\left(  \frac{\lambda_{i}\left(
\mu\right)  }{\Delta_{N}}\right)  ,H\left(  \frac{\lambda_{i}\left(
\mu\right)  }{\Delta_{N}}\right)  -H\left(  \frac{\eta_{i}\left(  \mu\right)
}{\Delta_{N}}\right)  \right)  ,
\]
and all points on the CE trade-off curve for a $1\rightarrow N$ cloning
channel have the following form:%
\[
\bigg(  \log\left(  {N+1}\right)  +H_{2}\left(  \mu\right)  -H\left(  \eta
_{i}\left(  \mu\right)  /\Delta_{N}\right)  ,\ H_{2}\left(  \mu\right)
\bigg)  ,
\]
where $H$ is the entropy function $H\left(  \cdot\right)  \equiv-\sum
_{i}\left(  \cdot\right)  \log\left(  \cdot\right)  $,%
\begin{align*}
\Delta_{N}  &  \equiv N\left(  N+1\right)  /2,\\
\lambda_{i}\left(  \mu\right)   &  \equiv(N-2i)\mu+i\ \ \ \text{for\ \ \ }%
0\leq i\leq N,\\
\eta_{i}\left(  \mu\right)   &  \equiv(N-1-2i)\mu+i+1\ \ \ \text{for\ \ \ }%
0\leq i\leq N-1,\\
\mu &  \in\left[  0,1/2\right]  .
\end{align*}

\end{theorem}

\begin{lemma}
\label{lem:parametrize-CEQ-cloning}An ensemble of the following form
parametrizes all points on the CQ\ trade-off curve for a $1\rightarrow N$
cloning channel:%
\begin{equation}
\frac{1}{2}\left\vert 0\right\rangle \left\langle 0\right\vert ^{X}\otimes
\psi_{0}^{AA^{\prime}}+\frac{1}{2}\left\vert 1\right\rangle \left\langle
1\right\vert ^{X}\otimes\psi_{1}^{AA^{\prime}}, \label{cq-state-cloning-mu}%
\end{equation}
where $\psi_{0}^{AA^{\prime}}$ and $\psi_{1}^{AA^{\prime}}$ are pure states,
defined as follows for $\mu\in\left[  0,1/2\right]  $:%
\begin{align}
\text{\emph{Tr}}_{A}\left\{  \psi_{0}^{AA^{\prime}}\right\}   &
=\mu\left\vert 0\right\rangle \left\langle 0\right\vert ^{A^{\prime}}+\left(
1-\mu\right)  \left\vert 1\right\rangle \left\langle 1\right\vert ^{A^{\prime
}},\label{eq:cloning-mu-1}\\
\text{\emph{Tr}}_{A}\left\{  \psi_{1}^{AA^{\prime}}\right\}   &  =\left(
1-\mu\right)  \left\vert 0\right\rangle \left\langle 0\right\vert ^{A^{\prime
}}+\mu\left\vert 1\right\rangle \left\langle 1\right\vert ^{A^{\prime}}.
\label{eq:cloning-mu-2}%
\end{align}

\end{lemma}

\begin{proof}
Consider a classical-quantum state with a finite number of conditional
density operators $\rho_{x}^{A^{\prime}}$:%
\begin{equation}
\rho^{XA^{\prime}}=\sum_{x} p_{X}\left(  x\right)  |x\rangle\langle
x|^{X}\otimes\rho_{x}^{A^{\prime}}.\nonumber
\end{equation}
Let $\mathcal{N}_{\text{Cl}}^{A^{\prime}\rightarrow B}$ denote the cloning
channel, $U_{\mathcal{N}_{\text{Cl}}}^{A^{\prime}\rightarrow BE}$ an isometric
extension, $(\mathcal{N}_{\text{Cl}}^{c})^{A^{\prime}\rightarrow E}$ a
complementary channel, and let $\rho^{XBE}$ denote the state resulting from
applying the isometric extension $U_{\mathcal{N}_{\text{Cl}}}^{A^{\prime
}\rightarrow BE}$ of the cloning channel to system $A^{\prime}$ of
$\rho^{XA^{\prime}}$.

We can form a new classical-quantum state with quadruple the number of
conditional density operators by applying all Pauli matrices to the original
conditional density operators:%
\[
\sigma^{X J A^{\prime}}\equiv\sum_{x} \sum_{j=0}^{3}\frac{1}{4}%
p_{X}\left(  x\right)  |x\rangle\langle x|^{X}\otimes\left\vert j\right\rangle
\left\langle j\right\vert ^{J}\otimes\sigma_{j}\rho_{x}^{A^{\prime}}\sigma
_{j},
\]
where $\sigma_{j}$ labels the four Pauli matrices:\ $\sigma_{0}\equiv I$,
$\sigma_{1}\equiv X$, $\sigma_{2}\equiv Y$, $\sigma_{3}\equiv Z$. Let
$\sigma^{XJBE}$ denote the state resulting from sending $A^{\prime}$ through
the isometric extension $U_{\mathcal{N}_{\text{Cl}}}^{A^{\prime}\rightarrow
BE}$ of the cloning channel.

Recall from Section~\ref{sec:cloning}\ that the cloning channel is covariant.
In fact, the following relationships hold for any input density operator
$\sigma$ and any unitary $V$ acting on the input system $A^{\prime}$:%
\begin{align*}
\mathcal{N}_{\text{Cl}}\left(  V\sigma V^{\dag}\right)   &  =R_{V}%
\mathcal{N}_{\text{Cl}}\left(  \sigma\right)  R_{V}^{\dag},\\
\mathcal{N}_{\text{Cl}}^{c}\left(  V\sigma V^{\dag}\right)   &  =S_{V}%
\mathcal{N}_{\text{Cl}}^{c}\left(  \sigma\right)  S_{V}^{\dag},
\end{align*}
where $R_{V}$ and $S_{V}$ are higher-dimensional irreducible representations
of the unitary $V$ on the respective systems $B$ and $E$. The state
$\sigma^{B}$ is equal to the maximally mixed state on the symmetric subspace
for the following reasons:%
\begin{align}
\sigma^{B}  &  = \mathcal{N}_{\text{Cl}}\left(  \sigma^{A^{\prime}}\right)
= \mathcal{N}_{\text{Cl}}\left(  \frac{I^{A^{\prime}}}%
{2}\right)  =\mathcal{N}_{\text{Cl}}\left(  \int V\omega V^{\dag}%
\ \text{d}V\right) \nonumber \\
&  =\int R_{V}\mathcal{N}\left(  \omega\right)  R_{V^{\dag}}\ \text{d}%
V=\frac{1}{N+1}\sum_{i=0}^{N}\left\vert i\right\rangle \left\langle
i\right\vert ^{B},\label{eq:cloning-unital-relation}
\end{align}
where the third equality exploits the linearity and covariance of the cloning
channel $\mathcal{N}_{\text{Cl}}$.

Then the following chain of inequalities holds for all $\lambda\geq1$:%
\begin{align}
&  I\left(  X;B\right)  _{\rho}+\lambda I\left(  A\rangle BX\right)  _{\rho
}\nonumber\\
&  =H\left(  B\right)  _{\rho}+\left(  \lambda-1\right)  H\left(  B|X\right)
_{\rho}-\lambda H\left(  E|X\right)  _{\rho}\nonumber\\
&  =H\left(  B\right)  _{\rho}+\left(  \lambda-1\right)  H\left(  B|XJ\right)
_{\sigma}-\lambda H\left(  E|XJ\right)  _{\sigma}\nonumber\\
&  \leq H\left(  B\right)  _{\sigma}+\left(  \lambda-1\right)  H\left(
B|XJ\right)  _{\sigma}-\lambda H\left(  E|XJ\right)  _{\sigma}\nonumber\\
&  =\log\left(  N+1\right)  +\left(  \lambda-1\right)  H\left(  B|XJ\right)
_{\sigma}-\lambda H\left(  E|XJ\right)  _{\sigma}\nonumber\\
&  =\log\left(  N+1\right)  +\sum_{x}p_{X}\left(  x\right)  \left[  \left(
\lambda-1\right)  H\left(  B\right)  _{\rho_{x}}-\lambda H\left(  E\right)
_{\rho_{x}}\right] \nonumber\\
&  \leq\log\left(  N+1\right)  +\max_{x}\left[  \left(  \lambda-1\right)
H\left(  B\right)  _{\rho_{x}}-\lambda H\left(  E\right)  _{\rho_{x}}\right]
\nonumber\\
&  =\log\left(  N+1\right)  +\left(  \lambda-1\right)  H\left(  B\right)
_{\rho_{x}^{\ast}}-\lambda H\left(  E\right)  _{\rho_{x}^{\ast}}.
\label{eq:last-line}%
\end{align}
The first equality follows by standard entropic manipulations. The second
equality follows because the conditional entropies are invariant under
unitary transformations:%
\begin{align*}
H(B)_{R_{\sigma_{j}}\rho_{x}^{B}R_{\sigma_{j}}^{\dag}}  &  =H(B)_{\rho_{x}%
^{B}},\\
H(E)_{S_{\sigma_{j}}\rho_{x}^{E}S_{\sigma_{j}}^{\dag}}  &  =H(E)_{\rho_{x}%
^{E}},
\end{align*}
where $R_{\sigma_{j}}$ and $S_{\sigma_{j}}$ are higher-dimensional
representations of $\sigma_{j}$ on systems $B$ and $E$, respectively. The
first inequality follows because entropy is concave, i.e., the local state
$\sigma^{B}$ is a mixed version of $\rho^{B}$. The third equality follows
because (\ref{eq:cloning-unital-relation}) implies that $H(B)_{\sigma^{B}%
}=\log\left(  N+1\right)  $. The fourth equality follows because the systems
$X$ and $J$ are classical. The second inequality follows because the maximum
value of a realization of a random variable is not less than its expectation.
The final equality simply follows by defining $\rho_{x}^{\ast}$ to be the
conditional density operator on systems $B$ and $E$ that arises from sending
an arbitrary state through the channel.

The entropies $H(B)_{\rho_{x}^{*}}$ and $H(E)_{\rho_{x}^{*}}$ depend only on
the eigenvalues of the input state $\rho_{x}^{*}$ by the covariance of both
the cloning channel and its complement. We can therefore choose $\rho_{x}^{*}$
to be diagonal in the $\{ |0\rangle,|1\rangle\}$ basis of $A^{\prime}$. The
ensemble defined to consist of the purifications of $\rho_{x}^{*}$ and
$\sigma_{x} \rho_{x}^{*} \sigma_{x}^{*}$ assigned equal probabilities then
saturates the upper bound on $I(X;B)_{\rho}+ \lambda I(A\rangle BX )_{\rho}$
in (\ref{eq:last-line}), which concludes the proof.
\end{proof}
\smallskip
%

%TCIMACRO{\FRAME{ftbpFU}{7.0534in}{2.9386in}{0pt}{\Qcb{(Color online) Plot of
%(a) the CQ trade-off curve and (b) the CE trade-off curve of a $1\rightarrow
%N$ cloning channel for $N=1$, $2$, $3$, $5$, $8$, $12$, and $24$. The
%trade-off curves for $N=1$ correspond to those for the noiseless qubit channel
%and are the rightmost trade-off curves in each plot. In both plots, proceeding
%left from the $N=1$ curve, we obtain trade-off curves for $N=2$, $3$, $5$,
%$8$, $12$, and $24$ and notice that they all beat a time-sharing strategy by a
%larger relative proportion when $N$ increases.}}{\Qlb{fig:cloning-CQ-CE}%
%}{cloning-plot.pdf}{\special{ language "Scientific Word";  type "GRAPHIC";
%width 7.0534in;  height 2.9386in;  depth 0pt;  original-width 12.7802in;
%original-height 5.2935in;  cropleft "0";  croptop "1";  cropright "1";
%cropbottom "0";  filename 'cloning-plot.pdf';file-properties "XNPEU";}}}%
%BeginExpansion
\begin{figure*}
[ptb]
\begin{center}
\includegraphics[
natheight=5.293500in,
natwidth=12.780200in,
height=2.9386in,
width=7.0534in
]%
{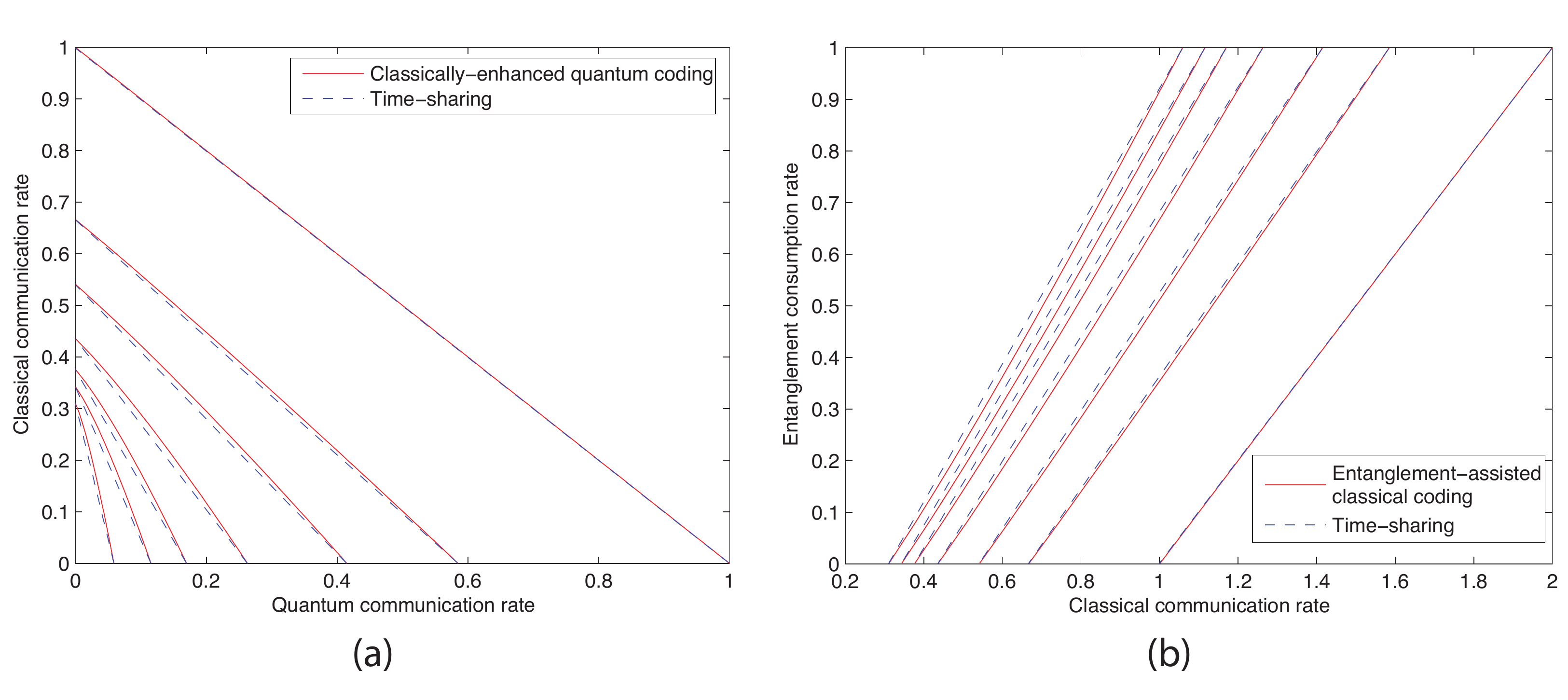}%
\caption{(Color online) Plot of (a) the CQ trade-off curve and (b) the CE
trade-off curve of a $1\rightarrow N$ cloning channel for $N=1$, $2$, $3$,
$5$, $8$, $12$, and $24$. The trade-off curves for $N=1$ correspond to those
for the noiseless qubit channel and are the rightmost trade-off curves in each
plot. In both plots, proceeding left from the $N=1$ curve, we obtain trade-off
curves for $N=2$, $3$, $5$, $8$, $12$, and $24$ and notice that they all beat
a time-sharing strategy by a larger relative proportion when $N$ increases.}%
\label{fig:cloning-CQ-CE}%
\end{center}
\end{figure*}
%EndExpansion

\begin{proof}
[Proof of CQ trade-off in Theorem~\ref{thm:cloning-CEQ}]We simply need to
compute the Holevo information $I\left(  X;B\right)  $ and the coherent
information $I\left(  A\rangle BX\right)  $ for an ensemble of the form in the
statement of Lemma~\ref{lem:parametrize-CEQ-cloning}.

The purifications of the states in (\ref{eq:cloning-mu-1}%
-\ref{eq:cloning-mu-2}) are the same as in (\ref{eq:dephasing-mu-1-pur}%
-\ref{eq:dephasing-mu-2-pur}). These states lead to classical-quantum states
of the form in (\ref{cq-state-cloning-mu}). An isometric extension
$U_{\mathcal{N}_{\text{Cl}}}$ of the $1\rightarrow N$ cloning channel acts as follows on
these states:%
\begin{multline*}
\mathop{\left|\psi_0\right>}\nolimits^{ABE}=\frac{1}{\sqrt{\Delta_{N}}%
}\bigg[\sum_{i=0}^{N-1}\bigg(\sqrt{\mu}\sqrt{N-i}%
\mathop{\left|0\right>}\nolimits^{A}\mathop{\left|i\right>}\nolimits^{B}\\
+\sqrt{1-\mu}\sqrt{i+1}\mathop{\left|1\right>}\nolimits^{A}%
\mathop{\left|i+1\right>}\nolimits^{B}%
\bigg)\mathop{\left|i\right>}\nolimits^{E}\bigg],
\end{multline*}%
\begin{multline*}
\mathop{\left|\psi_1\right>}\nolimits^{ABE}=\frac{1}{\sqrt{\Delta_{N}}%
}\bigg[\sum_{i=0}^{N-1}\bigg(\sqrt{1-\mu}\sqrt{N-i}%
\mathop{\left|0\right>}\nolimits^{A}\mathop{\left|i\right>}\nolimits^{B}\\
+\sqrt{\mu}\sqrt{i+1}\mathop{\left|1\right>}\nolimits^{A}%
\mathop{\left|i+1\right>}\nolimits^{B}%
\bigg)\mathop{\left|i\right>}\nolimits^{E}\bigg].
\end{multline*}
The state at the output of the isometric extension is as follows:%
\begin{equation}
\rho^{XABE}\equiv\frac{1}{2}\bigg[|0\rangle\!\langle0|^{X}\otimes\psi
_{0}^{ABE}+|1\rangle\!\langle1|^{X}\otimes\psi_{1}{}^{ABE}\bigg].
\nonumber\label{eq:outcl}%
\end{equation}
Defining $\lambda_{i}\left(  \mu\right)  \equiv(N-2i)\mu+i$ and $\Delta
_{N}=N\left(  N+1\right)  /2$ gives%
\begin{align}
\rho^{XB}  &  =\frac{1}{2}\bigg[|0\rangle\!\langle0|^{X}\otimes\psi_{0}%
^{B}+|1\rangle\!\langle1|^{X}\otimes\psi_{1}^{B}\bigg],\nonumber\\
\psi_{0}^{B}  &  =\sum_{i=0}^{N}\frac{\lambda_{i}\left(  \mu\right)  }%
{\Delta_{N}}|i\rangle\!\langle i|^{B},\nonumber\\
\psi_{1}^{B}  &  =\sum_{i=0}^{N}\frac{\lambda_{i}\left(  \mu\right)  }%
{\Delta_{N}}|N-i\rangle\!\langle N-i|^{B},\\
\rho^{B}  &  =\frac{1}{N+1}\sum_{i=0}^{N}|i\rangle\!\langle i|^{B}.
\end{align}
We can then compute the entropies $H\left(  B\right)  $ and $H\left(
B|X\right)  $:%
\begin{align*}
H(B)  &  =\log{(N+1),}\\
H(B|X)  &  =H\left(  \lambda_{i}\left(  \mu\right)  /\Delta_{N}\right)  ,
\end{align*}
giving the following expression for the Holevo information:%
\begin{align}
I(X;B)  &  =H(B)-H(B|X)\nonumber\\
&  =\log{(N+1)-}H\left(  \lambda_{i}\left(  \mu\right)  /\Delta_{N}\right)  ,
\label{eq:holevo-cloning}%
\end{align}
The above expression coincides with the expression for the classical capacity
of the $1\rightarrow N$ cloning channel when $\mu=0$ (Corollary~2 in
Ref.~\cite{B09}):%
\begin{equation}
I(X;B)_{\mu=0}=1-\log{N}+\frac{1}{\Delta_{N}}\sum_{i=0}^{N}i\log{i}\nonumber
\end{equation}
The expression in (\ref{eq:holevo-cloning})\ vanishes when $\mu=\frac{1}{2}$.

We now compute the coherent information $I(A\rangle BX)=H(B|X)-H(E|X)$. The
following states are important in this computation:%
\begin{align}
\rho^{XE}  &  =\frac{1}{2}\bigg[|0\rangle\!\langle0|^{X}\otimes\psi_{0}%
^{E}+|1\rangle\!\langle1|^{X}\otimes\psi_{1}^{E}\bigg],\nonumber\\
\psi_{0}^{E}  &  =\sum_{i=0}^{N-1}\frac{\eta_{i}\left(  \mu\right)  }%
{\Delta_{N}}|i\rangle\!\langle i|^{E},\\
\psi_{1}^{E}  &  =\sum_{i=0}^{N-1}\frac{\eta_{i}\left(  \mu\right)  }%
{\Delta_{N}}|N-1-i\rangle\!\langle N-1-i|^{E},
\end{align}
where $\eta_{i}\left(  \mu\right)  \equiv(N-1-2i)\mu+i+1$. We can then compute
the entropy $H\left(  E|X\right)  $:%
\[
H(E|X)=H\left(  \eta_{i}\left(  \mu\right)  /\Delta_{N}\right)  {,}%
\]
and the coherent information is as follows:%
\begin{equation}
I(A\rangle BX)=H\left(  \lambda_{i}\left(  \mu\right)  /\Delta_{N}\right)
-H\left(  \eta_{i}\left(  \mu\right)  /\Delta_{N}\right)  .\nonumber
\end{equation}
The above expression vanishes when $\mu=0$. It coincides with the quantum
capacity of a $1\rightarrow N$ cloning channel when $\mu=\frac{1}{2}$
(Equation~30 in Ref.~\cite{BDHM09}):%
\[
I(A\rangle BX)_{\mu=1/2}=\log\left(  {\frac{N+1}{N}}\right)  {.}%
\]

\end{proof}

We now turn to the proof of the CE\ trade-off curve for the $1\rightarrow N$
cloning channel.

\begin{lemma}
\label{lem:parametrize-EAC-cloning}An ensemble of the same form as in
Lemma~\ref{lem:parametrize-CEQ-cloning} parametrizes all points on the CE
trade-off curve for a $1\rightarrow N$ cloning channel.
\end{lemma}

\begin{proof}
The proof of this lemma proceeds along similar lines as the proof of
Lemma~\ref{lem:parametrize-CEQ-cloning}, using the same states and ideas
concerning the cloning channel. We highlight the major differences by giving
the following chain of inequalities that holds for all $\lambda$ such that
$0\leq\lambda<1$:%
\begin{align*}
&  I\left(  AX;B\right)  _{\rho}-\lambda H\left(  A|X\right)  _{\rho}\\
&  =\left(  1-\lambda\right)  H\left(  A|X\right)  _{\rho}+H\left(  B\right)
_{\rho}-H\left(  E|X\right)  _{\rho}\\
&  =H\left(  B\right)  _{\rho}+\left(  1-\lambda\right)  H\left(  A|XJ\right)
_{\sigma}-H\left(  E|XJ\right)  _{\sigma}\\
&  \leq H\left(  B\right)  _{\sigma}+\left(  1-\lambda\right)  H\left(
A|XJ\right)  _{\sigma}-H\left(  E|XJ\right)  _{\sigma}\\
&  =\log\left(  N+1\right)  +\left(  1-\lambda\right)  H\left(  A|XJ\right)
_{\sigma}-H\left(  E|XJ\right)  _{\sigma}\\
&  =\log\left(  N+1\right)  +\sum_{x}p_{X}\left(  x\right)  \left[  \left(
1-\lambda\right)  H\left(  A\right)  _{\psi_{x}}-H\left(  E\right)  _{\psi
_{x}}\right] \\
&  \leq\log\left(  N+1\right)  +\max_{x}\left[  \left(  1-\lambda\right)
H\left(  A\right)  _{\psi_{x}}-H\left(  E\right)  _{\psi_{x}}\right] \\
&  =\log\left(  N+1\right)  +\left(  1-\lambda\right)  H\left(  A\right)
_{\psi_{x}^{\ast}}-H\left(  E\right)  _{\psi_{x}^{\ast}}.
\end{align*}
We do not provide justifications for the above chain of inequalities because
it follows for similar reasons to the chain of inequalities in
Lemma~\ref{lem:parametrize-CEQ-cloning}.
\end{proof}

\smallskip

\begin{proof}
[Proof of CE trade-off in Theorem~\ref{thm:cloning-CEQ}]The proof follows by
noting that $I\left(  AX;B\right)  =H\left(  A|X\right)  +H\left(  B\right)
-H\left(  E|X\right)  $, $H\left(  A|X\right)  =H_{2}\left(  \mu\right)  $,
and that we have already computed $H\left(  B\right)  $ and $H\left(
E|X\right)  $ in the proof of the CQ trade-off in
Theorem~\ref{thm:cloning-CEQ}.
\end{proof}

Figure~\ref{fig:cloning-CQ-CE}\ plots the CQ\ and CE\ trade-off curves for a
$1\rightarrow N$ cloning channel, for $N=1$, $2$, $3$, $5$, $8$, $12$, and
$24$. The figure demonstrates that both classically-enhanced quantum coding
and entanglement-assisted classical coding beat a time-sharing strategy for a
cloning channel when $N>1$.

Figure~\ref{fig:EACQ-cloning}\ plots the full CQE capacity region for a
$1\rightarrow10$ cloning channel. The figure demonstrates that the
classically-enhanced father protocol, combined with entanglement distribution
and super-dense coding beats a time-sharing strategy because the first two
surfaces described in Section~\ref{sec:EACQ-review}\ are strictly concave for
the cloning channels when $N>1$.%
%TCIMACRO{\FRAME{ftbpFU}{3.4402in}{2.4215in}{0pt}{\Qcb{(Color online) The
%figure plots the CQE capacity region for a $1\rightarrow10$ cloning channel.
%It features three distinct surfaces. The first is the flat vertical plane that
%arises from the bound $R+2Q\leq\log\left(  \frac{N+1}{N}\right)  +1$, which is
%the entanglement-assisted classical capacity of a $1\rightarrow N$ cloning
%channel. The plane extends infinitely upward because we can always achieve
%these rate triples simply by wasting entanglement. The second surface is that
%below and to the left of the plane, formed by combining the CE trade-off curve
%with the inverse of super-dense coding, as described in
%Section~\ref{sec:EACQ-review}. The final surface is that below and to the
%right of the plane, formed by combining the CQ trade-off curve with the
%inverse of entanglement distribution, as described in
%Section~\ref{sec:EACQ-review}. }}{\Qlb{fig:EACQ-cloning}}%
%{triple-cloning-1-10.pdf}{\special{ language "Scientific Word";
%type "GRAPHIC";  maintain-aspect-ratio TRUE;  display "USEDEF";
%valid_file "F";  width 3.4402in;  height 2.4215in;  depth 0pt;
%original-width 5.8219in;  original-height 4.3708in;  cropleft "0";
%filename 'triple-cloning-1-10.pdf';file-properties "XNPEU";}}}%
%BeginExpansion
\begin{figure}
[ptb]
\begin{center}
\includegraphics[
natheight=4.370800in,
natwidth=5.821900in,
height=2.4215in,
width=3.4402in
]%
{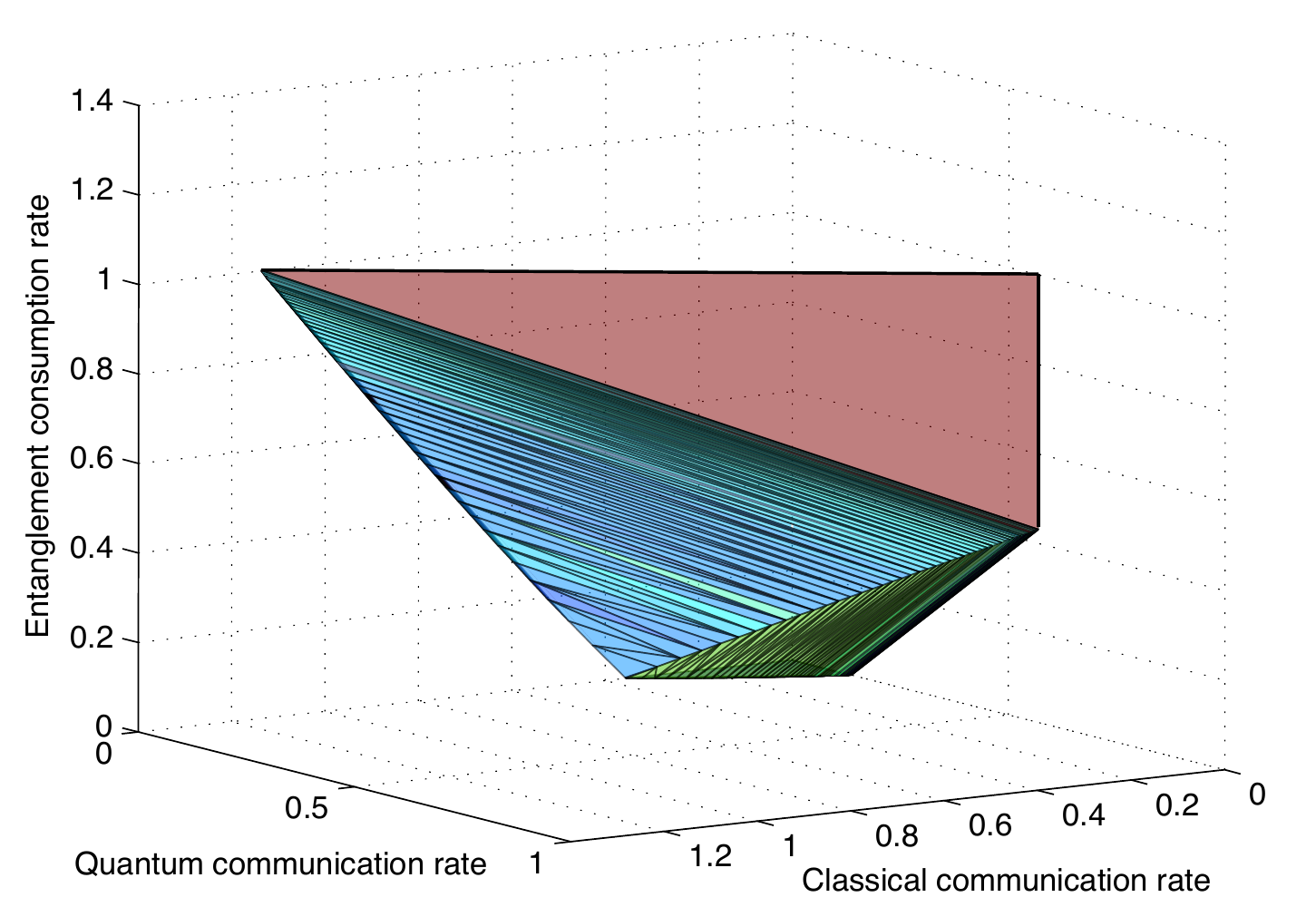}%
\caption{(Color online) The figure plots the CQE capacity region for a
$1\rightarrow10$ cloning channel. It features three distinct surfaces. The
first is the flat vertical plane that arises from the bound $R+2Q\leq
\log\left(  \frac{N+1}{N}\right)  +1$, which is the entanglement-assisted
classical capacity of a $1\rightarrow N$ cloning channel. The plane extends
infinitely upward because we can always achieve these rate triples simply by
wasting entanglement. The second surface is that below and to the left of the
plane, formed by combining the CE trade-off curve with the inverse of
super-dense coding, as described in Section~\ref{sec:EACQ-review}. The final
surface is that below and to the right of the plane, formed by combining the
CQ trade-off curve with the inverse of entanglement distribution, as described
in Section~\ref{sec:EACQ-review}. }%
\label{fig:EACQ-cloning}%
\end{center}
\end{figure}
%EndExpansion

\subsection{Unruh Channel}

\label{sec:dsunruh}We now compute the trade-off curves for the Unruh channel,
defined in (\ref{eq:Unruh-channel}). Capacity questions are directly linked to
those of the cloning channel because the mathematical structure of the output
of the Unruh channel is an infinite-dimensional block-diagonal density matrix,
with each block containing an occurence of a $1\rightarrow N$ cloning channel.
In fact, maximizing the rate of transmission is equivalent to maximizing it in
each block. This observation gives us the following theorem, and its proof
appears in Appendix~\ref{sec:cq-ce-unruh}.

\begin{theorem}
\label{thm:Unruh-CEQ}All points on the CQ trade-off curve for an Unruh channel
have the following form:%
\[
\left(  C_{\text{\emph{CQ}}}\left(  z,\mu\right)  ,\ Q_{\text{\emph{CQ}}%
}\left(  z,\mu\right)  \right)  ,
\]
and all points on the CE\ trade-off curve for an Unruh channel have the
following form:%
\[
\left(  C_{\text{\emph{CE}}}\left(  z,\mu\right)  ,\ E_{\text{\emph{CE}}%
}\left(  \mu\right)  \right)  ,
\]
where%
\begin{align*}
C_{\text{\emph{CQ}}}\left(  z,\mu\right)   &  \equiv1-\sum_{l=2}^{\infty}%
p_{l}\left(  z\right)  \log{(l-1)}\\
&  +\sum_{l=2}^{\infty}\frac{p_{l}\left(  z\right)  }{\Delta_{l-1}}\sum
_{i=0}^{l-1}\lambda_{i}^{\left(  l-1\right)  }\left(  \mu\right)  \log
{(}\lambda_{i}^{\left(  l-1\right)  }\left(  \mu\right)  {)},
\end{align*}%
\begin{align*}
Q_{\text{\emph{CQ}}}\left(  z,\mu\right)   &  \equiv-\sum_{l=2}^{\infty}%
\frac{p_{l}\left(  z\right)  }{\Delta_{l-1}}\sum_{i=0}^{l-1}\lambda
_{i}^{\left(  l-1\right)  }\left(  \mu\right)  \log{(\lambda_{i}^{\left(
l-1\right)  }\left(  \mu\right)  )}\\
&  +\sum_{l=2}^{\infty}\frac{p_{l}\left(  z\right)  }{\Delta_{l-1}}\sum
_{i=0}^{l-2}\eta_{i}^{\left(  l-1\right)  }\left(  \mu\right)  \log{(\eta
_{i}^{\left(  l-1\right)  }\left(  \mu\right)  )},
\end{align*}%
\begin{align*}
C_{\text{\emph{CE}}}\left(  z,\mu\right)   &  \equiv H_{2}\left(  \mu\right)
+1-\sum_{l=2}^{\infty}p_{l}\left(  z\right)  \log{(l-1)}\\
&  +\sum_{l=2}^{\infty}\frac{p_{l}\left(  z\right)  }{\Delta_{l-1}}\sum
_{i=0}^{l-2}\eta_{i}^{\left(  l-1\right)  }\left(  \mu\right)  \log{(\eta
_{i}^{\left(  l-1\right)  }\left(  \mu\right)  )},\\
E_{\text{\emph{CE}}}\left(  \mu\right)   &  \equiv H_{2}\left(  \mu\right)
,\\
\lambda_{i}^{\left(  l-1\right)  }(\mu)  &  \equiv(l-1-2i)\mu+i,\\
\eta_{i}^{\left(  l-1\right)  }(\mu)  &  \equiv(l-2-2i)\mu+i+1,
\end{align*}
and $H_{2}$ is the binary entropy function.
\end{theorem}

Figure~\ref{fig:unruh-plots}\ plots both the CQ and CE trade-off curves of the
Unruh channel for several values of the acceleration parameter $z=0$, $0.2$,
$0.4$, $0.6$, $0.8$, $0.95$. The figure demonstrates that both
classically-enhanced quantum coding and entanglement-assisted classical coding
beat a time-sharing strategy for an Unruh channel when $z>0$.

Figure~\ref{fig:unruh-triple} plots the full CQE capacity region for an Unruh
channel with acceleration parameter $z=0.95$. The figure demonstrates that the
classically-enhanced father protocol, combined with entanglement distribution
and super-dense coding beats a time-sharing strategy because the first two
surfaces described in Section~\ref{sec:EACQ-review}\ are strictly concave for
the Unruh channel when $z>0$.%
%TCIMACRO{\FRAME{ftbpFU}{7.0534in}{2.9672in}{0pt}{\Qcb{(Color online) Plot of
%the CQ trade-off curve (a) and the CE trade-off curve (b) of an Unruh channel
%for $z=0$, $0.2$, $0.4$, $0.6$, $0.8$, and $0.95$. The trade-off curves for
%$z=0$ correspond to those for the noiseless qubit channel and are the
%rightmost trade-off curves in each plot. In both plots, proceeding left from
%the $z=0$ curve, we obtain trade-off curves for $z=0.2$, $0.4$, $0.6$, $0.8$,
%and $0.95$ and notice that they all beat a time-sharing strategy by a larger
%relative proportion as $z$ increases.}}{\Qlb{fig:unruh-plots}}{unruh-plot.pdf}%
%width 7.0534in;  height 2.9672in;  depth 0pt;  original-width 12.6539in;
%original-height 5.2935in;  cropleft "0";  croptop "1";  cropright "1";
%cropbottom "0";  filename 'unruh-plot.pdf';file-properties "XNPEU";}}}%
%BeginExpansion
\begin{figure*}
[ptb]
\begin{center}
\includegraphics[
natheight=5.293500in,
natwidth=12.653900in,
height=2.9672in,
width=7.0534in
]%
{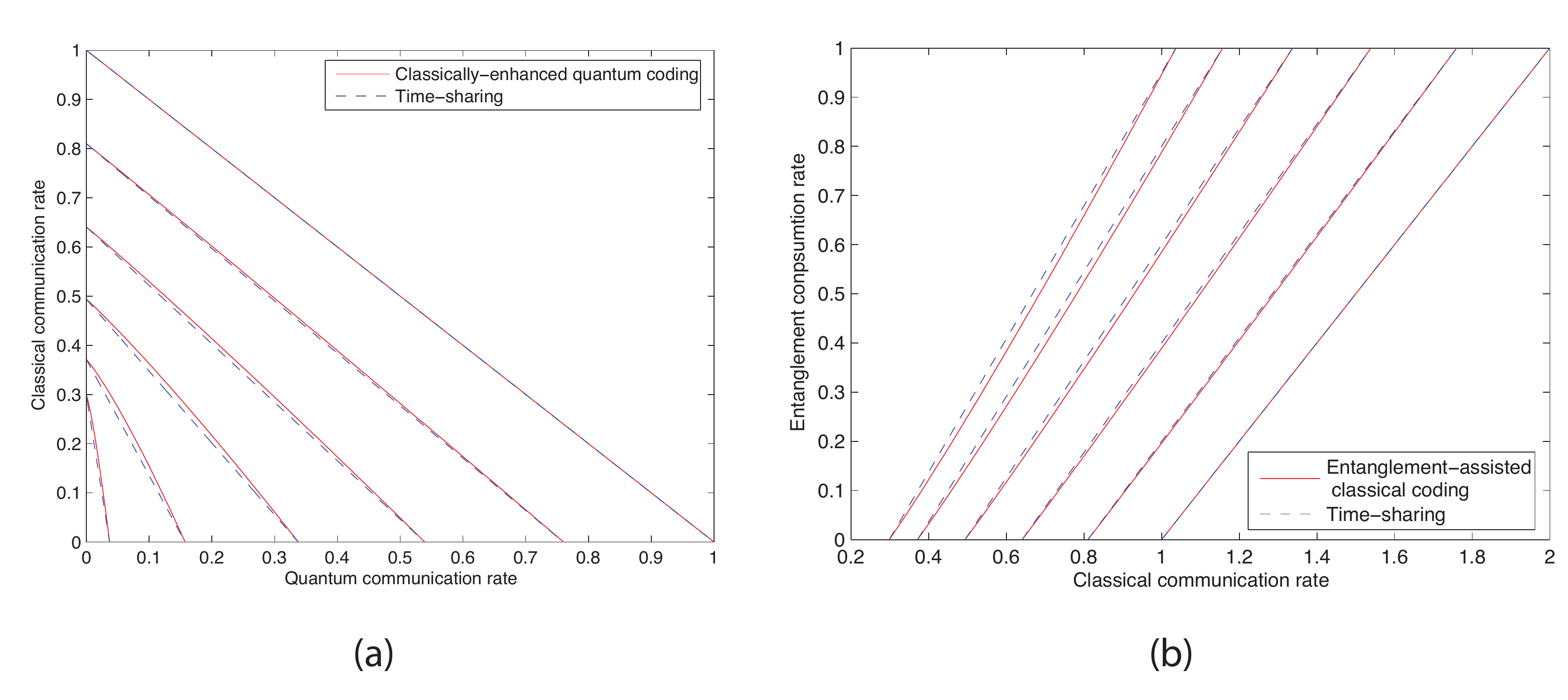}%
\caption{(Color online) Plot of the CQ trade-off curve (a) and the CE
trade-off curve (b) of an Unruh channel for $z=0$, $0.2$, $0.4$, $0.6$, $0.8$,
and $0.95$. The trade-off curves for $z=0$ correspond to those for the
noiseless qubit channel and are the rightmost trade-off curves in each plot.
In both plots, proceeding left from the $z=0$ curve, we obtain trade-off
curves for $z=0.2$, $0.4$, $0.6$, $0.8$, and $0.95$ and notice that they all
beat a time-sharing strategy by a larger relative proportion as $z$
increases.}%
\label{fig:unruh-plots}%
\end{center}
\end{figure*}
%EndExpansion%
%TCIMACRO{\FRAME{ftbpFU}{3.4411in}{2.6308in}{0pt}{\Qcb{(Color online) The
%figure plots the CQE capacity region for an Unruh channel with acceleration
%parameter $z=0.95$. It features three distinct surfaces. The first is the flat
%vertical plane that arises from the bound $R+2Q\leq\QTR{cal}{C}_{\text{EAC}}$,
%where $\QTR{cal}{C}_{\text{EAC}}$ is the entanglement-assisted classical
%capacity of an Unruh channel. The plane extends infinitely upward because we
%can always achieve these rate triples simply by wasting entanglement. The
%second surface is that below and to the left of the plane, formed by combining
%the CE trade-off curve with the inverse of super-dense coding, as described in
%Section~\ref{sec:EACQ-review}. The final surface is that below and to the
%right of the plane, formed by combining the CQ trade-off curve with the
%inverse of entanglement distribution, as described in
%Section~\ref{sec:EACQ-review}. }}{\Qlb{fig:unruh-triple}}{unruh-triple.pdf}%
%width 3.4411in;  height 2.6308in;  depth 0pt;  original-width 5.54in;
%original-height 4.2263in;  cropleft "0";  croptop "1";  cropright "1";
%cropbottom "0";  filename 'unruh-triple.pdf';file-properties "XNPEU";}}}%
%BeginExpansion
\begin{figure}
[ptb]
\begin{center}
\includegraphics[
natheight=4.226300in,
natwidth=5.540000in,
height=2.6308in,
width=3.4411in
]%
{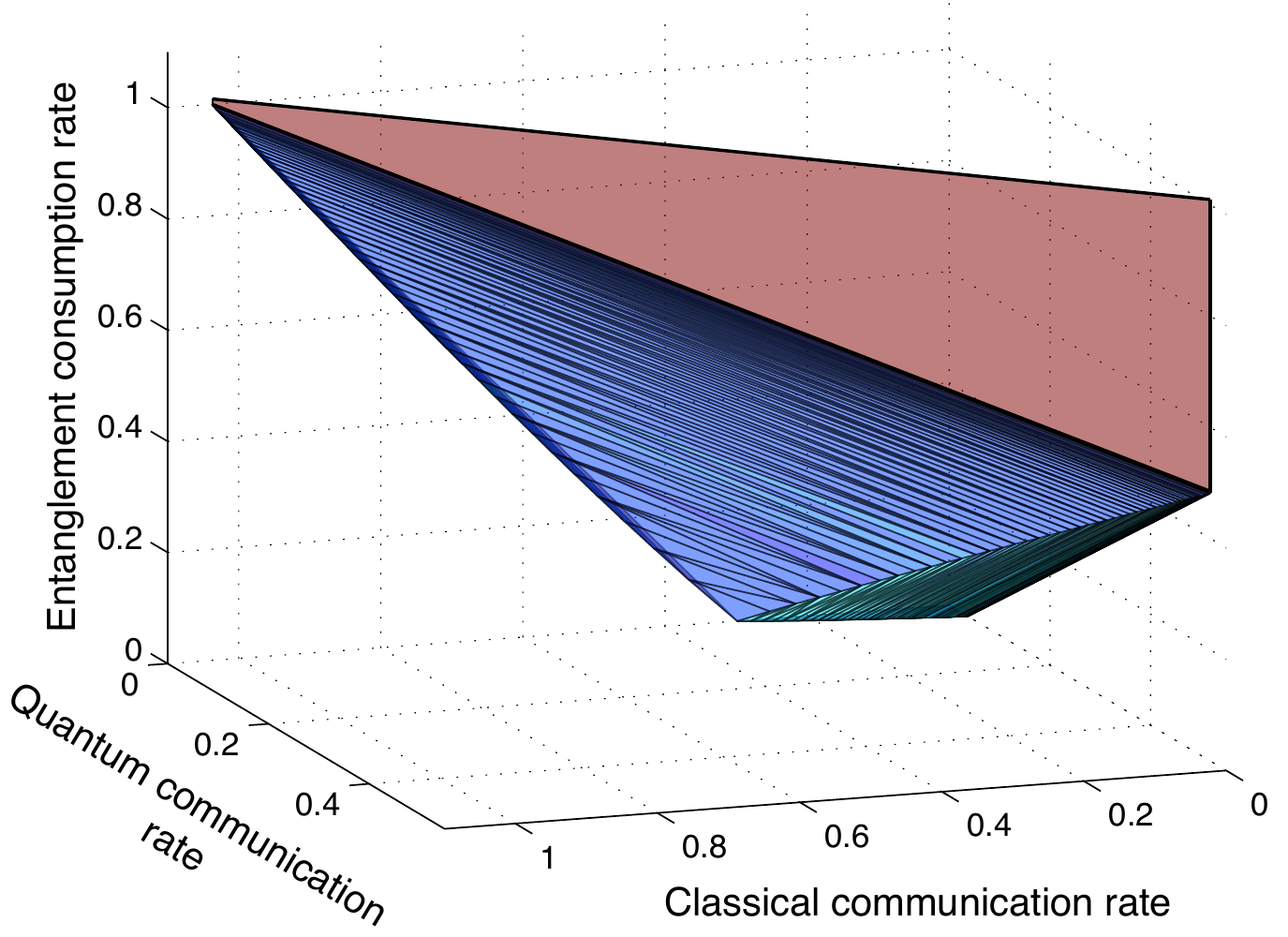}%
\caption{(Color online) The figure plots the CQE capacity region for an Unruh
channel with acceleration parameter $z=0.95$. It features three distinct
surfaces. The first is the flat vertical plane that arises from the bound
$R+2Q\leq\mathcal{C}_{\text{EAC}}$, where $\mathcal{C}_{\text{EAC}}$ is the
entanglement-assisted classical capacity of an Unruh channel. The plane
extends infinitely upward because we can always achieve these rate triples
simply by wasting entanglement. The second surface is that below and to the
left of the plane, formed by combining the CE trade-off curve with the inverse
of super-dense coding, as described in Section~\ref{sec:EACQ-review}. The
final surface is that below and to the right of the plane, formed by combining
the CQ trade-off curve with the inverse of entanglement distribution, as
described in Section~\ref{sec:EACQ-review}. }%
\label{fig:unruh-triple}%
\end{center}
\end{figure}
%EndExpansion

\section{Measuring the Gain over a Time-Sharing Strategy}

\label{sec:disc}Figures~\ref{fig:CQ-EC-dephasing}, \ref{fig:cloning-CQ-CE},
and \ref{fig:unruh-plots} demonstrate that classically-enhanced quantum coding
and entanglement-assisted classical coding both beat the time-sharing strategy
for the dephasing, cloning, and Unruh channels. Ref.~\cite{cmp2005dev}%
\ provided a simple way to compute the benefit of \textquotedblleft
specialized coding\textquotedblright\ over the time-sharing strategy, simply by
plotting the difference between a trade-off curve and the line of time-sharing
as a function of one of the rates in a trade-off scenario.

The above gain measure illustrates the benefit of specialized coding, but it
ignores the relative improvement that specialized coding may give over
time-sharing for very noisy channels. As a result, that gain measure tends to
zero as one of the capacities tends to zero and thus loses meaning in this
asymptotic limit. For example, consider a cloning channel with $N=1,000,000$.
This channel is particularly noisy for quantum transmission with a low quantum
capacity at approximately $1.5\times10^{-6}~$qubits / channel use, but the
classical capacity is approximately $0.27~$bits / channel use. Suppose that
the sender would like to transmit classical data at a rate of $0.165~$bits /
channel use in order to transmit more quantum information. With a time-sharing
strategy, the sender can transmit quantum data at approximately the rate
$5.9\times10^{-7}~$qubits / channel use, while classically-enhanced quantum
coding transmits quantum data at approximately the rate $7.2\times10^{-7}$
qubits / channel use. This improvement in transmission appears low in absolute
terms, but the relative increase of transmission is substantial, and a measure
that captures this relative increase may be more useful for studying the gain.

We suggest an alternate gain measure that highlights the relative improvement
and is simple to compute numerically for both the CQ\ and CE\ trade-off
curves. Let $A_{\text{CQ}}$ denote the area under the CQ~trade-off curve and
let $A_{\text{TSCQ}}$ denote the area under the line of time-sharing. Then the
relative gain $G_{\text{CQ}}$ for CQ~trading is the ratio of $A_{\text{CQ}}$
to $A_{\text{TSCQ}}$:%
\[
G_{\text{CQ}}\equiv\frac{A_{\text{CQ}}}{A_{\text{TSCQ}}}.
\]
The relative gain measure for the CE trade-off curve is similar. Let
$A_{\text{CE}}$ denote the area under the CE\ trade-off curve and let
$A_{\text{TSCE}}$ denote the area under the line of time-sharing. Then the
relative gain $G_{\text{CE}}$ for CE~trading is the ratio of $A_{\text{TSCQ}}$
to $A_{\text{CE}}$:%
\[
G_{\text{CE}}\equiv\frac{A_{\text{TSCE}}}{A_{\text{CE}}}.
\]
Each of the above relative gains exhibits non-trivial behavior even if one of
the capacities vanishes as the noise of a channel increases. These measures
are also average gains because the area involves an integration over all
points on the trade-off curve. One could generalize these relative gain
measures to the CQE~capacity region by taking the ratio of the volume enclosed
by the bounding surfaces for CQE~capacity region to the volume enclosed by
surfaces obtained by time-sharing.

Figure~\ref{fig:payoff} plots the relative gains~$G_{\text{CQ}}$ and
$G_{\text{CE}}$ as a function of the dephasing parameter $p$ for the dephasing
qubit channel, as a function of the acceleration parameter $z$ for the Unruh
channel, and as a function of the number of clones $N$ for the $1\rightarrow
N$ cloning channel. The accompanying caption features an interpretation of the
results.%
%TCIMACRO{\FRAME{ftbpFU}{7.0534in}{3.0978in}{0pt}{\Qcb{(Color online) The
%figures plot (a) the relative gain $G_{\text{CQ}}$ for the CQ trade-off curve
%and (b) the relative gain $G_{\text{CE}}$ for the CE\ trade-off curve for the
%dephasing channel, the cloning channel, and the Unruh channel. The figures
%plot these relative gains as a function of the dephasing parameter
%$p\in\left[  0,1\right]  $ for the $p$-dephasing qubit channel (bottom
%horizontal axis), as a function of the acceleration parameter $z\in\left[
%0,1\right]  $ for the Unruh channel (bottom horizontal axis), and as a
%function of the number of clones$\ N$ for the $1\rightarrow N$ cloning channel
%(top horizontal axis). The plot on the left demonstrates that the relative
%gain $G_{\text{CQ}}$ for the cloning channels is best as $N$ increases, while
%the Unruh channel features an improved relative gain over a dephasing channel
%for large accelerations. The plot on the right features different
%behavior---the relative gain $G_{\text{CE}}$ of the dephasing channel
%outperforms that for the Unruh channel if we consider the parameters $p$ and
%$z$ on equal footing, in spite of their drastically different physical
%interpretations.}}{\Qlb{fig:payoff}}{payoff-plot.pdf}%
%width 7.0534in;  height 3.0978in;  depth 0pt;  original-width 12.9999in;
%original-height 5.6801in;  cropleft "0";  croptop "1";  cropright "1";
%cropbottom "0";  filename 'payoff-plot.pdf';file-properties "XNPEU";}}}%
%BeginExpansion
\begin{figure*}
[ptb]
\begin{center}
\includegraphics[
natheight=5.680100in,
natwidth=12.999900in,
height=3.0978in,
width=7.0534in
]%
{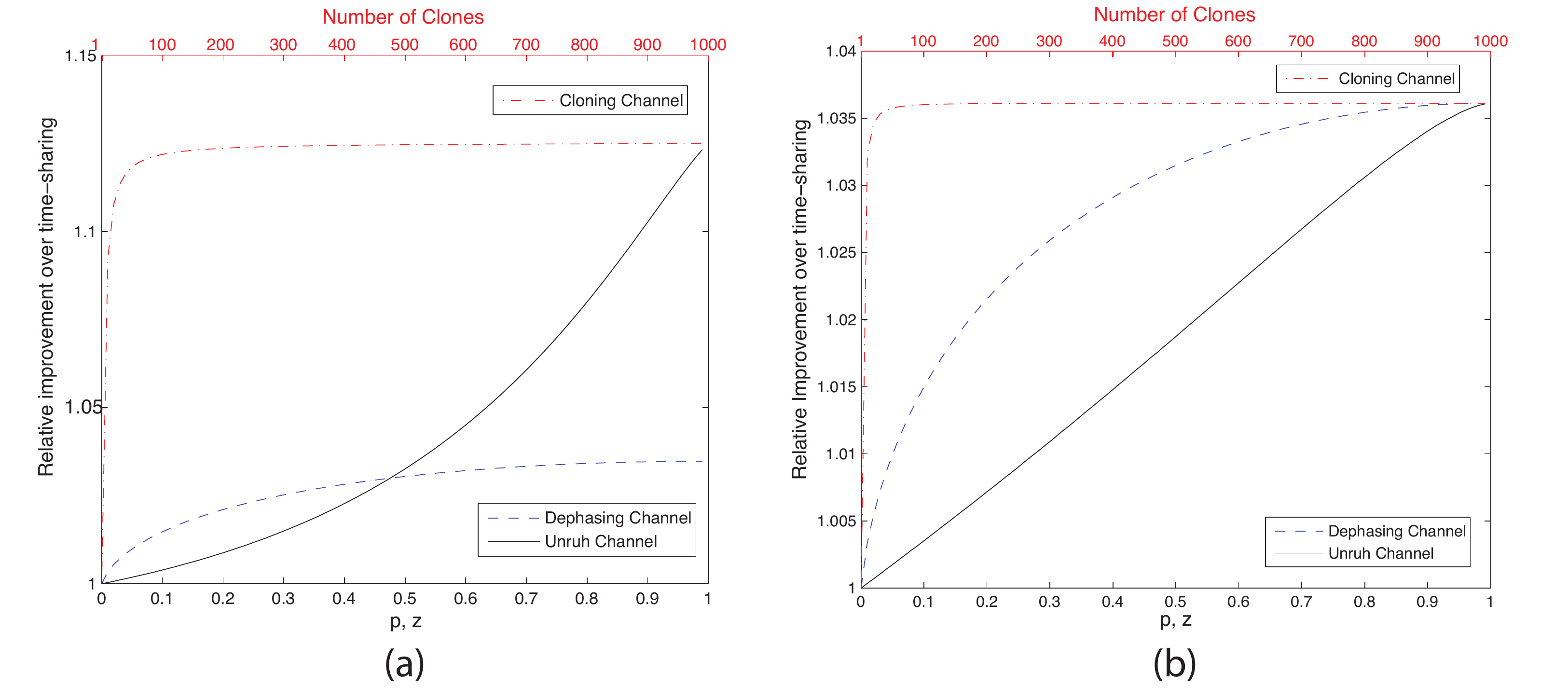}%
\caption{(Color online) The figures plot (a) the relative gain $G_{\text{CQ}}$
for the CQ trade-off curve and (b) the relative gain $G_{\text{CE}}$ for the
CE\ trade-off curve for the dephasing channel, the cloning channel, and the
Unruh channel. The figures plot these relative gains as a function of the
dephasing parameter $p\in\left[  0,1\right]  $ for the $p$-dephasing qubit
channel (bottom horizontal axis), as a function of the acceleration parameter
$z\in\left[  0,1\right]  $ for the Unruh channel (bottom horizontal axis), and
as a function of the number of clones$\ N$ for the $1\rightarrow N$ cloning
channel (top horizontal axis). The plot on the left demonstrates that the
relative gain $G_{\text{CQ}}$ for the cloning channels is best as $N$
increases, while the Unruh channel features an improved relative gain over a
dephasing channel for large accelerations. The plot on the right features
different behavior---the relative gain $G_{\text{CE}}$ of the dephasing
channel outperforms that for the Unruh channel if we consider the parameters
$p$ and $z$ on equal footing, in spite of their drastically different physical
interpretations.}%
\label{fig:payoff}%
\end{center}
\end{figure*}
%EndExpansion

One criterion that seems to be necessary in order to obtain a large relative
gain for the CQ\ trade-off curve is that the quantum capacity of the channel
should be much smaller than the classical capacity, so that the area between
the inner bounding line of time-sharing and the trade-off curve is relatively large.

\section{Conclusion}

We have proven that the CQE capacity region for all Hadamard channels admits a
single-letter characterization. Particular examples of the Hadamard channels
are generalized dephasing channels and cloning channels, and we have computed
exact formulas that specify their CQE capacity regions. Furthermore, we have
obtained expressions for the CQE capacity region of an Unruh channel because
of its close connection with the cloning channels. The classically-enhanced
father protocol beats a simple time-sharing strategy for all of these
channels, stressing the need for non-trivial coding techniques when trading
multiple resources.

It is interesting to ponder the reason why a particular channel obtains an
improvement over time-sharing. The relative improvements are most significant
for the CQ trade-off, in which case the cloning and Unruh channels exhibit
much more substantial gains than the dephasing channels. In retrospect, it is
perhaps surprising that the dephasing channels exhibit any improvement at all.
Since these channels can transmit classical data noiselessly, it would be
natural to expect that any optimal strategy for sending classical bits would
directly exploit this capability. For CQ trade-off coding, that would entail
allocating some fraction of channel uses to noiseless classical data
transmission and the rest to quantum, which is exactly the time-sharing
strategy. The existence of a nontrivial CQ trade-off indicates that this
strategy is actually not optimal. In contrast, the cloning and Unruh channels
are incapable of sending classical data noiselessly when $N>1$ or $z>0$ so any
communication strategy requires error correction with the attendant
opportunity for nontrivial trade-off coding.

Some future directions for this work are in order. It would be desirable to
discover other channels for which the full CQE\ capacity region
single-letterizes, but it is likely that the technique for proving
single-letterization would be completely different. The ideas exploited here
are that Hadamard channels are degradable and have entanglement-breaking
complementary channels, allowing us to generalize the Devetak-Shor proof
technique in Appendix~B of Ref.~\cite{cmp2005dev}. If other single-letter
examples do exist, one should then determine if such a channel obtains an
improvement over time-sharing and perhaps attempt to uncover a general method
that determines if a channel obtains a gain over time-sharing.

It may be interesting to explore the static case, where two parties share a
bipartite state and exploit this state and some noiseless resources to extract
other noiseless resources. One might consider bipartite states that arise from
Hadamard channels and attempt to single-letterize the static capacity region.
Hsieh and Wilde found formulas for the triple trade-off static capacity region
in Ref.~\cite{HW09}, but the task of single-letterization for the static case
is more difficult than that for the dynamic case because one must consider
quantum instruments applied to many copies of the bipartite state.

It would also be interesting to consider versions of the Unruh channel other
than the original definition in Ref.~\cite{BHP09}\ and determine if the
corresponding CQE capacity region can single-letterize (we refer to this
channel as \textquotedblleft the\textquotedblright\ Unruh channel, but it may
be more appropriate to consider it as \textquotedblleft an\textquotedblright%
\ Unruh channel). One can consider the CQE\ capacity region for a qudit Unruh
channel by exploiting some of the insights in Ref.~\cite{BHP10}. Finally, one
might consider a trade-off capacity region for an Unruh channel that includes
the resource of private quantum information, given that this noiseless
resource appears in relativistic quantum information theory, or one could
consider trading public classical information, private classical information,
and secret key by exploiting the ideas in
Refs.~\cite{hsieh:042306,PhysRevA.80.022306}.

\begin{acknowledgments}
The authors acknowledge useful discussions with Min-Hsiu Hsieh concerning the
CQ and CE trade-off curves and the CQE capacity region. K.~B.~and
P.~H.~acknowledge support from the Office of Naval Research under grant
No.~N000140811249. P.~H.~acknowledges support from the Canada Research Chairs
program, CIFAR, FQRNT, MITACS, NSERC, QuantumWorks, and the Sloan Foundation.
M.~M.~W.~acknowledges support from the MDEIE (Qu\'{e}bec) PSR-SIIRI
international collaboration grant.
\end{acknowledgments}

\appendix

\section{Single-letterization of the Plane in (\ref{eq:EAC-plane})}

\label{sec:proof-EAC-plane}The following additivity lemma aids in proving the
additivity of the plane in (\ref{eq:EAC-plane}).

\begin{lemma}
\label{lem:additivity-EA-plane}The following additivity relation holds for any
two quantum channels $\mathcal{N}_{1}$ and $\mathcal{N}_{2}$:%
\[
h\left(  \mathcal{N}_{1}\otimes\mathcal{N}_{2}\right)  =h\left(
\mathcal{N}_{1}\right)  +h\left(  \mathcal{N}_{2}\right)  .
\]

\end{lemma}

\begin{proof}
The inequality $h(\mathcal{N}_{1}\otimes\mathcal{N}_{2})\geq h(\mathcal{N}%
_{1})+h(\mathcal{N}_{2})$ trivially holds for all quantum channels, because
the maximization on the RHS\ is a restriction of the maximization in the
LHS\ to a tensor product of states of the form in (\ref{eq:CEQ-state}).
Therefore, we prove the non-trivial inequality $h(\mathcal{N}_{1}%
\otimes\mathcal{N}_{2})\leq h(\mathcal{N}_{1})+h(\mathcal{N}_{2})$.

Let%
\begin{align*}
\psi^{XAA_{1}A_{2}}  &  \equiv\sum_{x}p_{X}\left(  x\right)  \proj{x}^{X}%
\ox\proj{\ph_x}^{AA_{1}A_{2}},\\
\theta^{XAB_{1}E_{1}A_{2}}  &  \equiv U_{1}\psi U_{1}^{\dag},\\
\omega^{XAA_{1}B_{2}E_{2}}  &  \equiv U_{2}\psi U_{2}^{\dag},\\
\rho^{XAB_{1}E_{1}B_{2}E_{2}}  &  \equiv(U_{1}\otimes U_{2})\psi(U_{1}\otimes
U_{2})^{\dag},
\end{align*}
where $U_{j}^{A_{j}\rightarrow B_{j}E_{j}}$ is the isometric extension of
$\mathcal{N}_{j}$. Suppose further that $\rho$ is the state that maximizes
$h\left(  \mathcal{N}_{1}\otimes\mathcal{N}_{2}\right)  $.

The following chain of inequalities holds for any two channels $\mathcal{N}%
_{1}$ and $\mathcal{N}_{2}$:%
%TCIMACRO{\TeXButton{TeX field}{\begin{widetext}}}%
%BeginExpansion
\begin{widetext}%
%EndExpansion%
\begin{align*}
h\left(  \mathcal{N}_{1}\otimes\mathcal{N}_{2}\right)   &  =I\left(
AX;B_{1}B_{2}\right)  _{\rho}\\
&  =H\left(  B_{1}B_{2}|E_{1}E_{2}X\right)  _{\rho}+H\left(  B_{1}%
B_{2}\right)  _{\rho}\\
&  \leq H\left(  B_{1}|E_{1}X\right)  _{\rho}+H\left(  B_{1}\right)  _{\rho
}+H\left(  B_{2}|E_{2}X\right)  _{\rho}+H\left(  B_{2}\right)  _{\rho},\\
&  =H\left(  B_{1}E_{1}|X\right)  _{\rho}-H\left(  E_{1}|X\right)  _{\rho
}+H\left(  B_{1}\right)  _{\rho}+H\left(  B_{2}E_{2}|X\right)  _{\rho
}-H\left(  E_{2}|X\right)  _{\rho}+H\left(  B_{2}\right)  _{\rho}\\
&  =H\left(  AA_{2}|X\right)  _{\theta}-H\left(  AA_{2}B_{1}|X\right)
_{\theta}+H\left(  B_{1}\right)  _{\theta}+H\left(  AA_{1}|X\right)  _{\omega
}-H\left(  AA_{1}B_{1}|X\right)  _{\omega}+H\left(  B_{2}\right)  _{\omega}\\
&  =I(AA_{2}X;B_{1})_{\theta}+I(AA_{1}X;B_{2})_{\omega}\\
&  \leq h\left(  \mathcal{N}_{1}\right)  +h\left(  \mathcal{N}_{2}\right)  .
\end{align*}%
%TCIMACRO{\TeXButton{TeX field}{\end{widetext}}}%
%BeginExpansion
\end{widetext}%
%EndExpansion
The first equality holds because $\rho$ is the state that maximizes $h\left(
\mathcal{N}_{1}\otimes\mathcal{N}_{2}\right)  $. The second equality follows
from straightforward entropic manipulations by noting that the state $\rho
$\ on systems $A$, $B_{1}$, $B_{2}$, $E_{1}$, and $E_{2}$ is pure when
conditioned on $X$. The first inequality follows from an application of strong
subadditivity~\cite{book2000mikeandike}\ and subadditivity. The next three
equalities follow by straightforward entropic manipulations. The last
inequality follows from the definition of $h$ in (\ref{eq:EA-bound}) and the
fact that both $\theta$ and $\omega$ are states of the form in
(\ref{eq:CEQ-state}).
\end{proof}

\begin{proof}
[Proof of Lemma~\ref{thm:EA-plane-single-letter}]The proof follows from
Lemma~\ref{lem:additivity-EA-plane} in the same way as
Corollary~\ref{thm:CEQ-single-letter} does from Lemma~\ref{lem:CEQ-base-case}.
\end{proof}

We remark that the above lemma follows from the single-letterization of the
entanglement-assisted classical capacity~\cite{ieee2002bennett}, but we have
included the result for completeness.

\section{Proof of the Single-letterization of the CE Trade-off Curve for
Hadamard channels}

\label{sec:CE-trade-off-curve}

\begin{lemma}
\label{claim:LShor}For any fixed $\lambda\geq0$, the function in
(\ref{eq:EAC-max-function}) leads to a point $(I(AX;B)_{\rho},H\left(
A|X\right)  _{\rho})$ on the one-shot CE trade-off curve in the sense of Theorem~\ref{thm:CE-single-letter}.
\end{lemma}

\begin{proof}
We argue by contradiction. Suppose that $\rho^{XAB}$ maximizes
(\ref{eq:EAC-max-function}). Suppose further that it does not lead to a point
on the CE trade-off curve. That is, given the constraint $C=I(AX;B)_{\rho}$,
there exists some other state $\sigma^{XAB}$ of the form in
(\ref{eq:CEQ-state}) such that $I(AX;B)_{\sigma}=I(AX;B)_{\rho}=C$, but
$E=H(A|X)_{\sigma}<H(A|X)_{\rho}$. That is, the state $\sigma$ allows for
communication of as much classical information as the state $\rho$ does, but
consumes less entanglement. But then the following inequality holds for all
$\lambda\geq0$:%
\[
I(AX;B)_{\sigma}-\lambda H(A|X)_{\sigma}>I(AX;B)_{\rho}-\lambda H(A|X)_{\rho
},
\]
contradicting the fact that the state $\rho^{XAB}$ maximizes
(\ref{eq:EAC-max-function}).
\end{proof}

\begin{lemma}
\label{lem:max-CE-lambda}We get all points on the CE trade-off curve by
considering $0\leq\lambda\leq1$ in the maximization in
(\ref{eq:EAC-max-function}) because the maximization optimizes only the
classical capacity when $\lambda>1$.
\end{lemma}

\begin{proof}
Consider a state $\rho^{XABE}$ of the form in (\ref{eq:CEQ-state}). Suppose
that we perform a von Neumann measurement of the system $A$, resulting in a
classical variable$~Y$, and let $\sigma^{XYBE}$ denote the resulting state.
The following chain of inequalities then holds for all $\lambda>1$:%
\begin{align*}
&  I(AX;B)_{\rho}-\lambda H(A|X)_{\rho}\\
&  =H\left(  A|X\right)  _{\rho}-H\left(  E|X\right)  _{\rho}+H\left(
B\right)  _{\rho}-\lambda H(A|X)_{\rho}\\
&  =\left(  1-\lambda\right)  H\left(  BE|X\right)  _{\rho}-H\left(
E|X\right)  _{\rho}+H\left(  B\right)  _{\rho}\\
&  =\left(  1-\lambda\right)  H\left(  BE|X\right)  _{\sigma}-H\left(
E|X\right)  _{\sigma}+H\left(  B\right)  _{\sigma}\\
&  \leq H\left(  B\right)  _{\sigma}-H\left(  E|X\right)  _{\sigma}\\
&  \leq H\left(  B\right)  _{\sigma}-H\left(  E|XY\right)  _{\sigma}\\
&  =H\left(  B\right)  _{\sigma}-H\left(  B|XY\right)  _{\sigma}\\
&  =I\left(  XY;B\right)  _{\sigma}.
\end{align*}
The first equality follows from definitions and from the equality $H\left(
AB|X\right)  _{\rho}=H\left(  E|X\right)  _{\rho}$ for the state $\rho^{XABE}%
$, and the second equality follows from algebra and the fact that $H\left(
BE|X\right)  _{\rho}=H\left(  A|X\right)  _{\rho}$. The third equality follows
because the von Neumann measurement does not affect the systems in the
entropic quantities. The first inequality follows because $\lambda>1$ and the
entropy $H\left(  BE|X\right)  _{\rho}$ is always positive. The second
inequality follows because conditioning reduces entropy. The fourth equality
follows because the reduced state of $\sigma^{XYBE}$ on systems $B$ and $E$ is
pure when conditioned on both$~X$ and$~Y$. The last equality follows from the
definition of the quantum mutual information.

Thus, it becomes clear that the maximization of the original quantity when
$\lambda>1$ is always less than the classical capacity because$~I\left(
XY;B\right)  _{\sigma}\leq\max_{\rho}I\left(  X;B\right)  $. It then follows
that the trade-off curve really occurs for the interval $0\leq\lambda\leq1$.
\end{proof}

The following lemma is the crucial one that leads to our main result in this
section:\ the single-letterization of the CE trade-off curve for Hadamard channels.

\begin{lemma}
\label{lem:additivity-EAC}The following additivity relation holds for a
Hadamard channel $\mathcal{N}_{1}$ and any other channel $\mathcal{N}_{2}$:%
\[
g_{\lambda}(\mathcal{N}_{1}\otimes\mathcal{N}_{2})=g_{\lambda}(\mathcal{N}%
_{1})+g_{\lambda}(\mathcal{N}_{2})
\]

\end{lemma}

\begin{proof}
The inequality $g_{\lambda}(\mathcal{N}_{1}\otimes\mathcal{N}_{2})\geq
g_{\lambda}(\mathcal{N}_{1})+g_{\lambda}(\mathcal{N}_{2})$ trivially holds for
all quantum channels, because the maximization on the RHS\ is a restriction of
the maximization in the LHS\ to a tensor product of states of the form in
(\ref{eq:CEQ-state}). Therefore, we prove the non-trivial inequality
$g_{\lambda}(\mathcal{N}_{1}\otimes\mathcal{N}_{2})\leq g_{\lambda
}(\mathcal{N}_{1})+g_{\lambda}(\mathcal{N}_{2})$ when $\mathcal{N}_{1}$ is a
Hadamard channel.

The channels $\mathcal{N}_{1}^{A_{1}\rightarrow B_{1}}$ and $\mathcal{N}%
_{2}^{A_{2}\rightarrow B_{2}}$ and their respective complementary channels
$\left(  \mathcal{N}_{1}^{c}\right)  ^{A_{1}\rightarrow E_{1}}$ and $\left(
\mathcal{N}_{2}^{c}\right)  ^{A_{2}\rightarrow E_{2}}$ have the same
properties as they do in the proof of Theorem~\ref{thm:CEQ-single-letter}. The
state that is the output of the channels is the state $\rho^{XAB_{1}E_{1}%
B_{2}E_{2}}$ in (\ref{eq:double-CEQ-state}), but we now define it to be the
state that maximizes $g_{\lambda}(\mathcal{N}_{1}\otimes\mathcal{N}_{2})$.
Define $\theta$ as before and $\sigma$ again to be the state after processing
system $B_{1}$ of $\rho$ with $\mathcal{D}_{1}$.

The following chain of inequalities holds when $\lambda\leq1$:%
%TCIMACRO{\TeXButton{TeX field}{\begin{widetext}}}%
%BeginExpansion
\begin{widetext}%
%EndExpansion%
\begin{align*}
g_{\lambda}(\mathcal{N}_{1}\otimes\mathcal{N}_{2})  &  =I\left(  AX;B_{1}%
B_{2}\right)  _{\rho}-\lambda H\left(  A|X\right)  _{\rho}\\
&  =\left(  1-\lambda\right)  H\left(  B_{1}B_{2}E_{1}E_{2}|X\right)  _{\rho
}-H\left(  E_{1}E_{2}|X\right)  _{\rho}+H\left(  B_{1}B_{2}\right)  _{\rho}\\
&  =\left(  1-\lambda\right)  H\left(  B_{1}E_{1}|X\right)  _{\rho}-H\left(
E_{1}|X\right)  _{\rho}+H\left(  B_{1}\right)  _{\rho}+\left(  1-\lambda
\right)  H\left(  B_{2}E_{2}|B_{1}E_{1}X\right)  _{\rho}-H\left(  E_{2}%
|E_{1}X\right)  _{\rho}+H\left(  B_{2}|B_{1}\right)  _{\rho}\\
&  \leq\left(  1-\lambda\right)  H\left(  B_{1}E_{1}|X\right)  _{\rho
}-H\left(  E_{1}|X\right)  _{\rho}+H\left(  B_{1}\right)  _{\rho}+\left(
1-\lambda\right)  H\left(  B_{2}E_{2}|YX\right)  _{\sigma}-H\left(
E_{2}|YX\right)  _{\sigma}+H\left(  B_{2}\right)  _{\sigma}\\
&  =\left(  1-\lambda\right)  H\left(  AA_{2}|X\right)  _{\theta}-H\left(
AA_{2}B_{1}|X\right)  _{\theta}+H\left(  B_{1}\right)  _{\theta}+\left(
1-\lambda\right)  H\left(  AE_{1}|YX\right)  _{\sigma}-H\left(  AE_{1}%
B_{2}|YX\right)  _{\sigma}+H\left(  B_{2}\right)  _{\sigma}\\
&  =\left[  I\left(  AA_{2}X;B_{1}\right)  _{\theta}-\lambda H\left(
AA_{2}|X\right)  _{\theta}\right]  +\left[  I\left(  AE_{1}XY;B_{2}\right)
_{\sigma}-\lambda H\left(  AE_{1}|XY\right)  _{\sigma}\right] \\
&  \leq g_{\lambda}\left(  \mathcal{N}_{1}\right)  +g_{\lambda}\left(
\mathcal{N}_{2}\right)  .
\end{align*}%
%TCIMACRO{\TeXButton{TeX field}{\end{widetext}}}%
%BeginExpansion
\end{widetext}%
%EndExpansion
The first equality follows because $\rho$ is the state that maximizes
$g_{\lambda}(\mathcal{N}_{1}\otimes\mathcal{N}_{2})$. The second equality
follows from entropic manipulations. The third equality follows from the chain
rule. The first and crucial inequality follows from monotonicity of the
conditional entropy $H\left(  B_{2}E_{2}|B_{1}E_{1}X\right)  _{\rho}$\ under
the map $\mathcal{D}_{1}$ and the discarding of system $E_{1}$, monotonicity
of the conditional entropy $H\left(  E_{2}|E_{1}X\right)  _{\rho}$ under the
map $\mathcal{D}_{2}$, and conditioning does not increase entropy. The fourth
equality follows because the reduced state of $\theta$ on systems $A$, $A_{2}%
$, $B_{1}$, and $E_{1}$ is pure when conditioned on $X$, and the reduced state
of $\sigma$ on systems $A$, $E_{1}$, $B_{2}$, and $E_{2}$ is pure when
conditioned on both$~X$ and$~Y$. The fifth equality follows from entropic
manipulations, and the final inequality follows because $\theta$ and $\sigma$
are both states of the form in (\ref{eq:CEQ-state}).
\end{proof}

\begin{corollary}
\label{thm:EAC-single-letter}The one-shot CE trade-off curve is equal to the
regularized CE\ trade-off curve when the noisy quantum channel $\mathcal{N}$
is a Hadamard channel:%
\[
g_{\lambda}\left(  \mathcal{N}^{\otimes n}\right)  =ng_{\lambda}\left(
\mathcal{N}\right)  .
\]

\end{corollary}

\begin{proof}
The proof exploits the same induction technique that
Corollary~\ref{thm:CEQ-single-letter}\ does, but it applies the result of
Lemma~\ref{lem:additivity-EAC}.
\end{proof}

\section{On the parametrization of the trade-off curve}
\label{sec:continuity-trade-off-curve}
\begin{lemma}
$\lambda$ parametrizes all points on the CQ\ and CE\ trade-off
curves with the possible exception of those lying on segments of
constant slope.
\end{lemma}

\begin{proof}
We prove the lemma for the case of the CQ trade-off curve. The proof
for the CE trade-off curve is similar. Let $(C(t),Q(t))$ for $0 \leq
t \leq 1$ be a parametrization of the trade-off curve with $C(0)$
equal to the classical capacity and $Q(1)$ equal to the quantum
capacity. The function $C\left(  t\right)  $ is monotonically
decreasing and the function $Q\left( t\right) $ is monotonically
increasing. The graph of the trade-off curve is convex and,
therefore, has one-sided directional derivatives at all points~\cite{RV73}. It
is also monotonically decreasing.

Consider the function $f_{\lambda
}\left(  \mathcal{N}\right)  $ where%
\[
f_{\lambda}\left(  \mathcal{N}\right)  \equiv\max_{\rho^{XBE}}I\left(
X;B\right)  +\lambda I\left(  A\rangle BX\right)  .
\]
For any $0 \leq t \leq 1$, we have $f_\lambda(\mathcal{N}) = C(t) +
\lambda Q(t) $ if and only if
\begin{equation} \label{eqn:tradeoff.slope}
C\left(  t\right)  +\lambda  Q\left(  t\right)  \geq C\left(
s\right)  +\lambda  Q\left(  s\right)  .
\end{equation}
for all $0 \leq s \leq 1$. Perhaps more instructively, if $s < t$
and $Q(s)\neq Q(t)$, this inequality can be
written as %
\[
\frac{C\left(  s\right)  -C\left(  t\right)  }{Q\left(  s\right)
-Q\left( t\right)  }\geq-\lambda
\]
because of the monotonicity of the functions $C$ and $Q$. Likewise,
when $s > t$, it has the form
\[
\frac{C\left(  s\right)  -C\left(  t\right)  }{Q\left(  s\right)
-Q\left( t\right)  }\leq-\lambda.
\]
If $(C(t),Q(t))$ is a point on the graph at which the derivative is
not constant, then setting $-\lambda$ to be the slope of the graph
will lead to Eq.~(\ref{eqn:tradeoff.slope}) being satisfied. If the
graph is not differentiable at $(C(t),Q(t))$, then the slope must
drop discontinuously at that point. Setting $-\lambda$ to any value
in the gap will again lead to the condition being satisfied.
\end{proof}

\smallskip
At points where the graph is differentiable but the slope is
constant, $\lambda$ might not be a good parameter. These points,
however, are in the convex hull of the points that $\lambda$ does
parametrize.

\section{Form of the CQ\ Trade-off Curve for Qubit Dephasing Channels}

\label{sec:CQ-dephasing-qubit}We first prove two important lemmas and then
prove a theorem that gives the exact form of the CQ trade-off curve.

\begin{lemma}
\label{lem:diagonal-dephasing-CEQ}Let $\mathcal{N}$ be a generalized dephasing
channel. In the optimization task for the CQ\ trade-off curve, it suffices to
consider a classical-quantum state with diagonal conditional density
operators, in the sense that the following inequality holds:%
\[
I\left(  X;B\right)  _{\rho}+\lambda I\left(  A\rangle BX\right)  _{\rho}\leq
I\left(  X;B\right)  _{\theta}+\lambda I\left(  A\rangle BX\right)  _{\theta
},
\]
where%
\begin{align*}
\rho^{XABE}  &  \equiv\sum_{x}p_{X}\left(  x\right)  \left\vert x\right\rangle
\left\langle x\right\vert ^{X}\otimes U_{\mathcal{N}}^{A^{\prime}\rightarrow
BE}(\phi_{x}^{AA^{\prime}}),\\
\theta^{XABE}  &  \equiv\sum_{x}p_{X}\left(  x\right)  \left\vert
x\right\rangle \left\langle x\right\vert ^{X}\otimes U_{\mathcal{N}%
}^{A^{\prime}\rightarrow BE}(\varphi_{x}^{AA^{\prime}}),
\end{align*}
$U_{\mathcal{N}}^{A^{\prime}\rightarrow BE}$ is an isometric extension of
$\mathcal{N}$, $\varphi_{x}^{A^{\prime}}=\Delta(\varphi_{x}^{A^{\prime}%
})=\Delta(\phi_{x}^{A^{\prime}})$, and $\Delta$ is the completely dephasing channel.
\end{lemma}

\begin{proof}
The proof of this lemma is similar to the proof of Lemma~9 in
Ref.~\cite{itit2008hsieh}. Consider another classical-quantum state $\sigma$
in addition to the two presented in the statement of the theorem:%
\[
\sigma^{XAYE}\equiv\sum_{x}p_{X}\left(  x\right)  \left\vert x\right\rangle
\left\langle x\right\vert ^{X}\otimes(\Delta^{B\rightarrow Y}\circ
U_{\mathcal{N}}^{A^{\prime}\rightarrow BE})(\phi_{x}^{AA^{\prime}}).
\]
Then the following chain of inequalities holds for all $\lambda\geq1$:%
\begin{align*}
&  I\left(  X;B\right)  _{\rho}+\lambda I\left(  A\rangle BX\right)  _{\rho}\\
&  =H\left(  B\right)  _{\rho}+\left(  \lambda-1\right)  H\left(  B|X\right)
_{\rho}-\lambda H\left(  E|X\right)  _{\rho}\\
&  \leq H\left(  Y\right)  _{\sigma}+\left(  \lambda-1\right)  H\left(
Y|X\right)  _{\sigma}-\lambda H\left(  E|X\right)  _{\sigma}\\
&  =H\left(  B\right)  _{\theta}+\left(  \lambda-1\right)  H\left(
B|X\right)  _{\theta}-\lambda H\left(  E|X\right)  _{\theta}\\
&  =I\left(  X;B\right)  _{\theta}+\lambda I\left(  A\rangle BX\right)
_{\theta}.
\end{align*}
The first equality follows from entropic manipulations. The inequality follows
because the entropies $H\left(  B\right)  _{\rho}$ and $H\left(  B|X\right)
_{\rho}$ can only increase under a complete dephasing
\cite{book2000mikeandike}. The second equality follows because $\mathcal{N}%
\circ\Delta=\Delta\circ\mathcal{N}$ and $\mathcal{N}^{c}\circ\Delta
=\mathcal{N}^{c}$ for a generalized dephasing channel $\mathcal{N}$, and the
final equality follows from entropic manipulations.
\end{proof}

\begin{lemma}
\label{lem:mu-CEQ}An ensemble of the following form parametrizes all points on
the CQ\ trade-off curve for a qubit dephasing channel:%
\begin{equation}
\frac{1}{2}\left\vert 0\right\rangle \left\langle 0\right\vert ^{X}\otimes
\psi_{0}^{AA^{\prime}}+\frac{1}{2}\left\vert 1\right\rangle \left\langle
1\right\vert ^{X}\otimes\psi_{1}^{AA^{\prime}}, \label{eq:mu-cq-state-CEQ}%
\end{equation}
where $\psi_{0}^{AA^{\prime}}$ and $\psi_{1}^{AA^{\prime}}$ are pure states,
defined as follows for $\mu\in\left[  0,1/2\right]  $:%
\begin{align}
\text{\emph{Tr}}_{A}\left\{  \psi_{0}^{AA^{\prime}}\right\}   &
=\mu\left\vert 0\right\rangle \left\langle 0\right\vert ^{A^{\prime}}+\left(
1-\mu\right)  \left\vert 1\right\rangle \left\langle 1\right\vert ^{A^{\prime
}},\label{eq:1st-mu-state-CEQ}\\
\text{\emph{Tr}}_{A}\left\{  \psi_{1}^{AA^{\prime}}\right\}   &  =\left(
1-\mu\right)  \left\vert 0\right\rangle \left\langle 0\right\vert ^{A^{\prime
}}+\mu\left\vert 1\right\rangle \left\langle 1\right\vert ^{A^{\prime}}.
\label{eq:2nd-mu-state-CEQ}%
\end{align}

\end{lemma}

\begin{proof}
We assume without loss of generality that the dephasing basis is the
computational basis. Consider a classical-quantum state with a finite number
$N$ of diagonal conditional density operators $\rho_{x}^{A^{\prime}}$:%
\begin{equation}
\rho^{XA^{\prime}}\equiv\sum_{x=0}^{N-1}p_{X}\left(  x\right)  |x\rangle
\langle x|^{X}\otimes\rho_{x}^{A^{\prime}}.\nonumber
\end{equation}
We can form a new classical-quantum state with double the number of
conditional density operators by \textquotedblleft
bit-flipping\textquotedblright\ the original conditional density operators:%
\begin{multline*}
\sigma^{XA^{\prime}}\equiv\frac{1}{2}\sum_{x=0}^{N-1}p_{X}\left(  x\right)
|x\rangle\langle x|^{X}\otimes\rho_{x}^{A^{\prime}}\\
+\frac{1}{2}\sum_{x=0}^{N-1}p_{X}\left(  x\right)  |x+N\rangle\langle
x+N|^{X}\otimes X\rho_{x}^{A^{\prime}}X,
\end{multline*}
where $X$ is the $\sigma_{X}$ \textquotedblleft bit-flip\textquotedblright%
\ Pauli operator. Consider the following chain of inequalities that holds for
all $\lambda\geq1$:%
\begin{align*}
&  I\left(  X;B\right)  _{\rho}+\lambda I\left(  A\rangle BX\right)  _{\rho}\\
&  =H\left(  B\right)  _{\rho}+\left(  \lambda-1\right)  H\left(  B|X\right)
_{\rho}-\lambda H\left(  E|X\right)  _{\rho}\\
&  =H\left(  B\right)  _{\rho}+\left(  \lambda-1\right)  H\left(  B|X\right)
_{\sigma}-\lambda H\left(  E|X\right)  _{\sigma}\\
&  \leq H\left(  B\right)  _{\sigma}+\left(  \lambda-1\right)  H\left(
B|X\right)  _{\sigma}-\lambda H\left(  E|X\right)  _{\sigma}\\
&  =1+\left(  \lambda-1\right)  H\left(  B|X\right)  _{\sigma}-\lambda
H\left(  E|X\right)  _{\sigma}\\
&  =1+\sum_{x}p_{X}\left(  x\right)  \left[  \left(  \lambda-1\right)
H\left(  B\right)  _{\rho_{x}}-\lambda H\left(  E\right)  _{\rho_{x}}\right]
\\
&  \leq1+\max_{x}\left[  \left(  \lambda-1\right)  H\left(  B\right)
_{\rho_{x}}-\lambda H\left(  E\right)  _{\rho_{x}}\right] \\
&  =1+\left(  \lambda-1\right)  H\left(  B\right)  _{\rho_{x}^{\ast}}-\lambda
H\left(  E\right)  _{\rho_{x}^{\ast}}.
\end{align*}
The first equality follows by standard entropic manipulations. The second
equality follows because the conditional entropy $H\left(  B|X\right)  $ is
invariant under a bit-flipping unitary on the input state that commutes with
the channel: $H(B)_{X\rho_{x}^{B}X}=H(B)_{\rho_{x}^{B}}$. Furthermore, a bit
flip on the input state does not change the eigenvalues for the output of the
dephasing channel's complementary channel (as observed in
Section~\ref{sec:generalized-dephasing}): $H(E)_{\mathcal{N}^{c}(X\rho
_{x}^{A^{\prime}}X)}=H(E)_{\mathcal{N}^{c}(\rho_{x}^{A^{\prime}})}$. The first
inequality follows because entropy is concave, i.e., the local state
$\sigma^{B}$ is a mixed version of $\rho^{B}$. The third equality follows
because $H(B)_{\sigma^{B}}=H\left(  \sum_{x}\frac{1}{2}p_{X}\left(  x\right)
(\rho_{x}^{B}+X\rho_{x}^{B}X)\right)  =H\left(  \frac{1}{2}\sum_{x}%
p_{X}\left(  x\right)  I\right)  =1$. The fourth equality follows because the
system $X$ is classical. The second inequality follows because the maximum
value of a realization of a random variable is not less than its expectation.
The final equality simply follows by defining $\rho_{x}^{\ast}$ to be the
conditional density operator on systems $B$ and $E$ that arises from sending a
diagonal state of the form $\mu\left\vert 0\right\rangle \left\langle
0\right\vert ^{A^{\prime}}+\left(  1-\mu\right)  \left\vert 1\right\rangle
\left\langle 1\right\vert ^{A^{\prime}}$ through the channel. Thus, an
ensemble of the kind in (\ref{eq:mu-cq-state-CEQ}) is sufficient to attain a
point on the CQ\ trade-off curve.

In the last step above, we observe that there is a direct correspondence
between $\mu$ that parametrizes the ensemble and $\lambda$ that parametrizes a
point on the CQ\ trade-off curve.
\end{proof}

\begin{proof}
[Proof of CQ\ trade-off in Theorem~\ref{thm:dephasing-CEQ}]We simply need to
compute the Holevo information $I\left(  X;B\right)  $ and the coherent
information $I\left(  A\rangle BX\right)  $ for an ensemble of the form in the
statement of Lemma~\ref{lem:mu-CEQ}, due to the results of
Lemmas~\ref{lem:diagonal-dephasing-CEQ}\ and~\ref{lem:mu-CEQ}.

First consider respective purifications of the states in
(\ref{eq:1st-mu-state-CEQ}-\ref{eq:2nd-mu-state-CEQ}):%
\begin{align}
\mathop{\left|\psi_0\right>}\nolimits^{AA^{\prime}}  &  =\sqrt{\mu
}\mathop{\left|0\right>}\nolimits^{A}%
\mathop{\left|0\right>}\nolimits^{A^{\prime}}+\sqrt{1-\mu}%
\mathop{\left|1\right>}\nolimits^{A}%
\mathop{\left|1\right>}\nolimits^{A^{\prime}},\label{eq:dephasing-mu-1-pur}\\
\mathop{\left|\psi_1\right>}\nolimits^{AA^{\prime}}  &  =\sqrt{1-\mu
}\mathop{\left|0\right>}\nolimits^{A}%
\mathop{\left|0\right>}\nolimits^{A^{\prime}}+\sqrt{\mu}%
\mathop{\left|1\right>}\nolimits^{A}%
\mathop{\left|1\right>}\nolimits^{A^{\prime}}. \label{eq:dephasing-mu-2-pur}%
\end{align}
The above states lead to a classical-quantum state of the form in
(\ref{eq:mu-cq-state-CEQ}). An isometric extension $U_{\mathcal{N}}%
=\sqrt{1-\frac{p}{2}}I\otimes\left\vert 0\right\rangle ^{E}+\sqrt{\frac{p}{2}%
}Z\otimes\left\vert 1\right\rangle ^{E}$ of the qubit dephasing channel acts
as follows on the above states:%
\begin{align*}
\mathop{\left|\psi_0\right>}\nolimits^{ABE}  &  \equiv U_{\mathcal{N}%
}\mathop{\left|\psi_0\right>}\nolimits^{AA^{\prime}}\\
&  =\sqrt{\mu}\mathop{\left|0\right>}\nolimits^{A}%
\mathop{\left|0\right>}\nolimits^{B}\bigg(\sqrt{1-\frac{p}{2}}%
\mathop{\left|0\right>}\nolimits^{E}+\sqrt{\frac{p}{2}}%
\mathop{\left|1\right>}\nolimits^{E}\bigg)\\
&  +\sqrt{1-\mu}\mathop{\left|1\right>}\nolimits^{A}%
\mathop{\left|1\right>}\nolimits^{B}\bigg(\sqrt{1-\frac{p}{2}}%
\mathop{\left|0\right>}\nolimits^{E}-\sqrt{\frac{p}{2}}%
\mathop{\left|1\right>}\nolimits^{E}\bigg),
\end{align*}%
\begin{align*}
\mathop{\left|\psi_1\right>}\nolimits^{ABE}  &  \equiv U_{\mathcal{N}%
}\mathop{\left|\psi_1\right>}\nolimits^{AA^{\prime}}\\
&  =\sqrt{1-\mu}\mathop{\left|0\right>}\nolimits^{A}%
\mathop{\left|0\right>}\nolimits^{B}\bigg(\sqrt{1-\frac{p}{2}}%
\mathop{\left|0\right>}\nolimits^{E}+\sqrt{\frac{p}{2}}%
\mathop{\left|1\right>}\nolimits^{E}\bigg)\\
&  +\sqrt{\mu}\mathop{\left|1\right>}\nolimits^{A}%
\mathop{\left|1\right>}\nolimits^{B}\bigg(\sqrt{1-\frac{p}{2}}%
\mathop{\left|0\right>}\nolimits^{E}-\sqrt{\frac{p}{2}}%
\mathop{\left|1\right>}\nolimits^{E}\bigg).
\end{align*}
The classical-quantum state at the output of the channel is as follows:%
\begin{equation}
\rho^{XABE}\equiv\frac{1}{2}\bigg[|0\rangle\!\langle0|^{X}\otimes|\psi
_{0}\rangle\!\langle\psi_{0}|^{ABE}+|1\rangle\!\langle1|^{X}\otimes|\psi
_{1}\rangle\!\langle\psi_{1}|^{ABE}\bigg]. \label{eq:outdep}%
\end{equation}
The following states are useful for computing the entropies $H\left(
B\right)  $ and $H\left(  B|X\right)  $:%
\begin{align}
\rho^{XB}  &  =\frac{1}{2}\bigg[|0\rangle\!\langle0|^{X}\otimes\psi_{0}%
^{B}+|1\rangle\!\langle1|^{X}\otimes\psi_{1}^{B}\bigg],\nonumber\\
\psi_{0}^{B}  &  =\mu|0\rangle\!\langle0|^{B}+(1-\mu)|1\rangle\!\langle
1|^{B},\nonumber\\
\psi_{1}^{B}  &  =(1-\mu)|0\rangle\!\langle0|^{B}+(\mu)|1\rangle
\!\langle1|^{B},\nonumber\\
\rho^{B}  &  =\frac{1}{2}\bigg[|0\rangle\!\langle0|^{B}+|1\rangle
\!\langle1|^{B}\bigg].\nonumber
\end{align}
The Holevo information is then as follows:%
\begin{equation}
I(X;B)=H(B)-H(B|X)=1-H_{2}(\mu).\nonumber
\end{equation}

We now compute the coherent information $I(A\rangle BX)=H(B|X)-H(E|X)$. The
following states are important in this computation:%
\begin{align*}
\rho^{XE}  &  =\frac{1}{2}\bigg[|0\rangle\!\langle0|^{X}\otimes\psi_{0}%
^{E}+|1\rangle\!\langle1|^{X}\otimes\psi_{1}^{E}\bigg],\\
\psi_{0}^{E}  &  =\left(  1-\frac{p}{2}\right)  |0\rangle\!\langle0|^{E}%
+\frac{p}{2}|1\rangle\!\langle1|^{E}\\
&  +\sqrt{1-\frac{p}{2}}\sqrt{\frac{p}{2}}(2\mu-1)(|0\rangle\!\langle
1|^{E}+|1\rangle\!\langle0|^{E}),\\
\psi_{1}^{E}  &  =\left(  1-\frac{p}{2}\right)  |0\rangle\!\langle0|^{E}%
+\frac{p}{2}|1\rangle\!\langle1|^{E}\\
&  -\sqrt{1-\frac{p}{2}}\sqrt{\frac{p}{2}}(2\mu-1)(|0\rangle\!\langle
1|^{E}+|1\rangle\!\langle0|^{E}).
\end{align*}
We compute the determinants of the density operators $\psi_{0}^{E}$ and
$\psi_{1}^{E}$:%
\begin{align*}
\text{Det}(\psi_{0}^{E})  &  =\text{Det}(\psi_{1}^{E})\\
&  =\left(  1-\frac{p}{2}\right)  \frac{p}{2}(1-(2\mu-1)^{2})\\
&  =2p\mu\left(  1-\frac{p}{2}\right)  (1-\mu).
\end{align*}
These determinants lead to the same eigenvalues for both $\psi_{0}^{E}$ and
$\psi_{1}^{E}$:%
\begin{align}
\lambda_{\pm}  &  \equiv\frac{1}{2}\pm\sqrt{\frac{1}{4}-\text{Det}(\psi
_{o}^{E})}\label{eq:lpm}\\
&  =\frac{1}{2}\pm\frac{1}{2}\sqrt{1-16\cdot\frac{p}{2}\left(  1-\frac{p}%
{2}\right)  \mu(1-\mu)}.\nonumber
\end{align}
Thus, the coherent information is as stated in the theorem: $I(A\rangle
BX)=H_{2}(\mu)-H_{2}(\lambda_{+})$.
\end{proof}

\section{Form of the CE Trade-off Curve for Qubit Dephasing Channels}

\label{sec:CE-trade-off-qubit-dephasing}We first prove two important lemmas
and then prove a theorem that gives the exact form of the CE trade-off curve.

\begin{lemma}
\label{lem:diagonal-dephasing-EAC}Let $\mathcal{N}$ be a generalized dephasing
channel. In the optimization task for the CE trade-off curve, it suffices to
consider a classical-quantum state with diagonal conditional density
operators, in the sense that the following inequality holds when $0\leq
\lambda\leq1$:%
\[
I\left(  AX;B\right)  _{\rho}-\lambda H\left(  A|X\right)  _{\rho}\leq
I\left(  AX;B\right)  _{\theta}-\lambda H\left(  A|X\right)  _{\theta},
\]
where%
\begin{align*}
\rho^{XABE}  &  \equiv\sum_{x}p_{X}\left(  x\right)  \left\vert x\right\rangle
\left\langle x\right\vert ^{X}\otimes U_{\mathcal{N}}^{A^{\prime}\rightarrow
BE}(\phi_{x}^{AA^{\prime}}),\\
\theta^{XABE}  &  \equiv\sum_{x}p_{X}\left(  x\right)  \left\vert
x\right\rangle \left\langle x\right\vert ^{X}\otimes U_{\mathcal{N}%
}^{A^{\prime}\rightarrow BE}(\varphi_{x}^{AA^{\prime}}),
\end{align*}
$U_{\mathcal{N}}^{A^{\prime}\rightarrow BE}$ is an isometric extension of
$\mathcal{N}$, $\varphi_{x}^{A^{\prime}}=\Delta(\varphi_{x}^{A^{\prime}%
})=\Delta(\phi_{x}^{A^{\prime}})$, and $\Delta$ is the completely dephasing channel.
\end{lemma}

\begin{proof}
The proof of this lemma is similar to the proof of Lemma~9 in
Ref.~\cite{itit2008hsieh}. Consider another classical-quantum state $\sigma
$\ in addition to the two presented in the statement of the theorem:%
\[
\sigma^{XY_{A}Y_{B}E}\equiv\sum_{x}p_{X}\left(  x\right)  \left\vert
x\right\rangle \left\langle x\right\vert ^{X}\otimes\Delta^{B\rightarrow
Y}(U_{\mathcal{N}}^{A^{\prime}\rightarrow BE}(\phi_{x}^{AA^{\prime}})).
\]
Then the following chain of inequalities holds when $0\leq\lambda\leq1$:%
\begin{align*}
&  I\left(  AX;B\right)  _{\rho}-\lambda H\left(  A|X\right)  _{\rho}\\
&  =\left(  1-\lambda\right)  H\left(  A|X\right)  _{\rho}+H\left(  B\right)
_{\rho}-H\left(  E|X\right)  _{\rho}\\
&  =\left(  1-\lambda\right)  \sum_{x}p_{X}\left(  x\right)  H\left(
A^{\prime}\right)  _{\phi_{x}}+H\left(  B\right)  _{\rho}-H\left(  E|X\right)
_{\rho}\\
&  \leq\left(  1-\lambda\right)  \sum_{x}p_{X}\left(  x\right)  H\left(
A^{\prime}\right)  _{\Delta(\phi_{x}^{A^{\prime}})}+H\left(  Y\right)
_{\sigma}-H\left(  E|X\right)  _{\sigma}\\
&  =\left(  1-\lambda\right)  H\left(  A|X\right)  _{\theta}+H\left(
B\right)  _{\theta}-H\left(  E|X\right)  _{\theta}\\
&  =I\left(  AX;B\right)  _{\theta}-\lambda H\left(  A|X\right)  _{\theta}.
\end{align*}
The first equality follows from entropic manipulations. The second equality
follows because the system $X$ is classical and the states $\phi
_{x}^{AA^{\prime}}$ are pure. The inequality follows because the entropies
$H\left(  A^{\prime}\right)  $ and $H\left(  B\right)  _{\rho}$ can only
increase under a complete dephasing. The third equality follows because
$\mathcal{N}\circ\Delta=\Delta\circ\mathcal{N}$ and $\mathcal{N}^{c}%
\circ\Delta=\mathcal{N}^{c}$ for a generalized dephasing channel $\mathcal{N}%
$, and the final equality follows from entropic manipulations.
\end{proof}

\begin{lemma}
\label{lem:mu-EAC}An ensemble of the following form parametrizes all points on
the CE trade-off curve for a qubit dephasing channel:%
\begin{equation}
\frac{1}{2}(\left\vert 0\right\rangle \left\langle 0\right\vert ^{X}%
\otimes\psi_{0}^{AA^{\prime}}+\left\vert 1\right\rangle \left\langle
1\right\vert ^{X}\otimes\psi_{1}^{AA^{\prime}}),
\end{equation}
where the states $\psi_{0}$ and $\psi_{1}$ are the same as they are in the
statement of Lemma~\ref{lem:mu-CEQ}.
\end{lemma}

\begin{proof}
The proof proceeds similarly to the proof of Lemma~\ref{lem:mu-CEQ}, with the
same definitions of states $\rho$ and $\sigma$, but with the following
different chain of inequalities for $0\leq\lambda\leq1$:%
\begin{align*}
&  I\left(  AX;B\right)  _{\rho}-\lambda H\left(  A|X\right)  _{\rho}\\
&  =\left(  1-\lambda\right)  H\left(  A|X\right)  _{\rho}+H\left(  B\right)
_{\rho}-H\left(  E|X\right)  _{\rho}\\
&  =\left(  1-\lambda\right)  H\left(  A|X\right)  _{\sigma}+H\left(
B\right)  _{\rho}-H\left(  E|X\right)  _{\sigma}\\
&  \leq\left(  1-\lambda\right)  H\left(  A|X\right)  _{\sigma}+H\left(
B\right)  _{\sigma}-H\left(  E|X\right)  _{\sigma}\\
&  =\left(  1-\lambda\right)  H\left(  A|X\right)  _{\sigma}+1-H\left(
E|X\right)  _{\sigma}\\
&  =1+\sum_{x}p_{X}\left(  x\right)  \left[  \left(  1-\lambda\right)
H\left(  A\right)  _{\psi_{x}}-H\left(  E\right)  _{\psi_{x}}\right] \\
&  \leq1+\max_{x}\left[  \left(  1-\lambda\right)  H\left(  A\right)
_{\psi_{x}}-H\left(  E\right)  _{\psi_{x}}\right] \\
&  =1+\left(  1-\lambda\right)  H\left(  A\right)  _{\psi_{x}^{\ast}}-H\left(
E\right)  _{\psi_{x}^{\ast}}.
\end{align*}
We do not provide justifications for the above steps because they are
identical those in the proof of Lemma~\ref{lem:mu-CEQ}.
\end{proof}

\begin{proof}
[Proof of CE\ trade-off in Theorem~\ref{thm:dephasing-CEQ}]The proof follows
by noting that $I\left(  AX;B\right)  =H\left(  A|X\right)  +H\left(
B\right)  -H\left(  E|X\right)  $, $H\left(  A|X\right)  =H_{2}\left(
\mu\right)  $, and that we have already computed $H\left(  B\right)  $ and
$H\left(  E|X\right)  $ in the first part of the proof of
Theorem~\ref{thm:dephasing-CEQ}.
\end{proof}

\section{CQ and CE\ Trade-off Curves for the Unruh Channel}

\label{sec:cq-ce-unruh}

\begin{proof}
The input state in (\ref{cq-state-cloning-mu}) that traces out the CQ curve
for the cloning channels also does so for the Unruh channel. We consider the
purification of this state, so that the output state is as follows:%
\begin{equation*}
\rho^{XABE}  \equiv\frac{1}{2}\bigg[|0\rangle\!\langle0|^{X}\otimes\psi
_{0}{}^{ABE}+|1\rangle\!\langle1|^{X}\otimes\psi_{1}^{ABE}\bigg].
\end{equation*}
Let $\mathcal{N}(\left\vert 0\right\rangle \left\langle 0\right\vert )$ and
$\mathcal{N}(\left\vert 1\right\rangle \left\langle 1\right\vert )$ denote the
Unruh channel outputs corresponding to the respective input states
$|0\rangle\!\langle0|^{A^{\prime}}$ and $|1\rangle\!\langle1|^{A^{\prime}}$:%
\begin{align*}
\mathcal{N}(\left\vert 0\right\rangle \left\langle 0\right\vert )  &
\equiv\bigoplus_{l=2}^{\infty}p_{l}\left(  z\right)  S_{l}(\left\vert
0\right\rangle \left\langle 0\right\vert )\\
&  =\bigoplus_{l=2}^{\infty}\frac{p_{l}\left(  z\right)  }{\Delta_{l-1}}%
\sum_{i=0}^{l-2}(l-1-i)|i\rangle\!\langle i|^{B},\\
\mathcal{N}(\left\vert 1\right\rangle \left\langle 1\right\vert )  &
\equiv\bigoplus_{l=2}^{\infty}p_{l}\left(  z\right)  S_{l}(\left\vert
1\right\rangle \left\langle 1\right\vert )\\
&  =\bigoplus_{l=2}^{\infty}\frac{p_{l}\left(  z\right)  }{\Delta_{l-1}}%
\sum_{i=0}^{l-2}\left(  i+1\right)  |i+1\rangle\!\langle i+1|^{B}.
\end{align*}
The following states are useful for calculating the Holevo information:%
\begin{align*}
\rho^{XB}  &  =\frac{1}{2}\bigg[|0\rangle\!\langle0|^{X}\otimes\psi_{0}%
^{B}+|1\rangle\!\langle1|^{X}\otimes\psi_{1}^{B}\bigg].\\
\psi_{0}^{B}  &  =\mu\mathcal{N}(\left\vert 0\right\rangle \left\langle
0\right\vert )+(1-\mu)\mathcal{N}(\left\vert 1\right\rangle \left\langle
1\right\vert )\\
&  =\bigoplus_{l=2}^{\infty}\frac{p_{l}\left(  z\right)  }{\Delta_{l-1}}%
\sum_{i=0}^{l-1}\lambda_{i}^{l-1}\left(  \mu\right)  |i\rangle\!\langle
i|^{B},\\
\psi_{1}^{B}  &  =(1-\mu)\mathcal{N}(\left\vert 0\right\rangle \left\langle
0\right\vert )+\mu\mathcal{N}(\left\vert 1\right\rangle \left\langle
1\right\vert )\\
&  =\bigoplus_{l=2}^{\infty}\frac{p_{l}\left(  z\right)  }{\Delta_{l-1}}%
\sum_{i=0}^{l-1}\lambda_{i}^{l-1}\left(  1-\mu\right)  |i\rangle\!\langle
i|^{B},\\
\rho^{B}  &  =\frac{1}{2}\bigg(\mathcal{N}(\left\vert 0\right\rangle \left\langle
0\right\vert )+\mathcal{N}(\left\vert 1\right\rangle \left\langle 1\right\vert
)\bigg)\\
&  =\bigoplus_{l=2}^{\infty}\frac{p_{l}\left(  z\right)  }{l}\sum_{i=0}%
^{l-1}|i\rangle\!\langle i|^{B},
\end{align*}
where $\lambda_{i}^{N}(\mu)\equiv(N-2i)\mu+i$. We then compute the entropies
$H\left(  B\right)  $ and $H\left(  B|X\right)  $:%
\begin{align}
H(B)_{\rho}  &  =-\sum_{l=2}^{\infty}\sum_{i=0}^{l-1}\frac{p_{l}\left(
z\right)  }{l}\log\left(  {\frac{p_{l}\left(  z\right)  }{l}}\right)
\nonumber\\
&  =-\sum_{l=2}^{\infty}p_{l}\left(  z\right)  \log\left(  {\frac{p_{l}\left(
z\right)  }{l}}\right)  {,}\nonumber\\
H(B|X)_{\rho}  &  =\frac{1}{2}\left(  H\left(  B\right)  _{\psi_{0}}+H\left(
B\right)  _{\psi_{1}}\right) \nonumber\\
&  =-\sum_{l=2}^{\infty}\sum_{i=0}^{l-1}\frac{p_{l}\left(  z\right)
\lambda_{i}^{l-1}\left(  \mu\right)  }{\Delta_{l-1}}\log\left(  {\frac
{p_{l}\left(  z\right)  \lambda_{i}^{l-1}\left(  \mu\right)  }{\Delta_{l-1}}%
}\right)  . \label{eq:HBX_unruh}%
\end{align}
Observe that the following relationships hold:%
\begin{align*}
\sum_{i=0}^{l-1}\lambda_{i}^{\left(  l-1\right)  }\left(  \mu\right)   &
=\Delta_{l-1},\\
\sum_{l=2}^{\infty}p_{l}\left(  z\right)   &  =1.
\end{align*}
These relationships allow us to rewrite the expression in (\ref{eq:HBX_unruh})
for $H\left(  B|X\right)  $ as follows:%
\begin{multline}
H(B|X)=H(B)-1+\sum_{l=2}^{\infty}p_{l}\left(  z\right)  \log{(l-1)}%
\label{eq:HBX-Unruh}\\
-\sum_{l=2}^{\infty}\sum_{i=0}^{l-1}\frac{p_{l}\left(  z\right)  }%
{\Delta_{l-1}}\lambda_{i}^{\left(  l-1\right)  }\left(  \mu\right)
\log{\bigg(\lambda_{i}^{\left(  l-1\right)  }\left(  \mu\right)  \bigg).}%
\end{multline}
and we get the Holevo information:%
\begin{multline}
I(X;B)=1-\sum_{l=2}^{\infty}p_{l}\left(  z\right)  \log{(l-1)}\nonumber\\
+\sum_{l=2}^{\infty}\frac{p_{l}\left(  z\right)  }{\Delta_{l-1}}\sum
_{i=0}^{l-1}\lambda_{i}^{\left(  l-1\right)  }\left(  \mu\right)
\log{\bigg(\lambda_{i}^{\left(  l-1\right)  }}\left(  \mu\right)  \bigg) .%
\end{multline}
The Holevo information $I(X;B)$ coincides with the expression for the
classical capacity of an Unruh channel when $\mu=0$ (Corollary~3,
Equation~(19) in Ref.~\cite{B09}, though note the slightly different
definition of $\Delta_{l}$ in that paper):%
\begin{equation}
I(X;B)_{\mu=0}=1-\sum_{l=2}^{\infty}p_{l}\left(  z\right)  \log{(l-1)}%
+\sum_{l=2}^{\infty}\frac{p_{l}\left(  z\right)  }{\Delta_{l-1}}\sum
_{i=0}^{l-1}i\log{i}.\nonumber
\end{equation}
The Holevo information vanishes when $\mu=1/2$.

We now compute the coherent information. Let $\mathcal{N}^{c}(\left\vert
0\right\rangle \left\langle 0\right\vert )$ and $\mathcal{N}^{c}(\left\vert
1\right\rangle \left\langle 1\right\vert )$ denote the outputs of the
complementary channel of the Unruh channel corresponding to the respective
input states $|0\rangle\!\langle0|^{A^{\prime}}$ and $|1\rangle\!\langle
1|^{A^{\prime}}$:%
\begin{align*}
\mathcal{N}^{c}(\left\vert 0\right\rangle \left\langle 0\right\vert )  &
\equiv\bigoplus_{l=2}^{\infty}p_{l}\left(  z\right)  S_{l}^{c}(\left\vert
0\right\rangle \left\langle 0\right\vert )\\
&  =\bigoplus_{l=2}^{\infty}\frac{p_{l}\left(  z\right)  }{\Delta_{l-1}}%
\sum_{i=0}^{l-2}(l-1-i)|i\rangle\!\langle i|^{E}\\
\mathcal{N}^{c}(\left\vert 1\right\rangle \left\langle 1\right\vert )  &
\equiv\bigoplus_{l=2}^{\infty}p_{l}\left(  z\right)  S_{l}^{c}(\left\vert
1\right\rangle \left\langle 1\right\vert )\\
&  =\bigoplus_{l=2}^{\infty}\frac{p_{l}\left(  z\right)  }{\Delta_{l-1}}%
\sum_{i=0}^{l-2}\left(  i+1\right)  |i\rangle\!\langle i|^{E}.
\end{align*}
The following states are important in this calculation:%
\begin{align*}
\rho^{XE}  &  =\frac{1}{2}\bigg[|0\rangle\!\langle0|^{X}\otimes\psi_{0}%
^{E}+|1\rangle\!\langle1|^{X}\otimes\psi_{1}^{E}\bigg],\\
\psi_{0}^{E}  &  =\mu\mathcal{N}^{c}(\left\vert 0\right\rangle \left\langle
0\right\vert )+(1-\mu)\mathcal{N}^{c}(\left\vert 1\right\rangle \left\langle
1\right\vert )\\
&  =\bigoplus_{l=2}^{\infty}\frac{p_{l}\left(  z\right)  }{\Delta_{l-1}}%
\sum_{i=0}^{l-2}\eta_{i}^{\left(  l-1\right)  }\left(  \mu\right)
|i\rangle\!\langle i|^{E},\\
\psi_{1}^{E}  &  =(1-\mu)\mathcal{N}^{c}(\left\vert 0\right\rangle
\left\langle 0\right\vert )+\mu\mathcal{N}^{c}(\left\vert 1\right\rangle
\left\langle 1\right\vert )\\
&  =\bigoplus_{l=2}^{\infty}\frac{p_{l}\left(  z\right)  }{\Delta_{l-1}}%
\sum_{i=0}^{l-2}\eta_{i}^{\left(  l-1\right)  }\left(  1-\mu\right)
|i\rangle\!\langle i|^{E},
\end{align*}
where $\eta_{i}^{\left(  l-1\right)  }(\mu)\equiv(l-2-2i)\mu+i+1$. Then the
conditional entropy $H\left(  E|X\right)  $ is as follows:%
\[
H(E|X)=-\sum_{l=2}^{\infty} \frac{p_{l}\left(  z\right)}{\Delta_{l-1}}
\sum_{i=0}^{l-2}
\eta_{i}^{\left(  l-1\right)  \left(  \mu\right)  }\log\left(
{\frac{p_{l}\left(  z\right)  \eta_{i}^{\left(  l-1\right)  }\left(
\mu\right)  }{\Delta_{l-1}}}\right)  .
\]
The relation $\sum_{i=0}^{l-2}\eta_{i}^{\left(  l-1\right)  }\left(
\mu\right)  =\Delta_{l-1}$ allows us to simplify the above expression for the
conditional entropy $H\left(  E|X\right)  $:%
\begin{multline}
H(E|X)=H(B)-1+\sum_{l=2}^{\infty}p_{l}\left(  z\right)  \log{(l-1)}%
\label{eq:HEX-unruh}\\
-\sum_{l=2}^{\infty}\sum_{i=0}^{l-2}\frac{p_{l}\left(  z\right)  }%
{\Delta_{l-1}}\eta_{i}^{\left(  l-1\right)  }\left(  \mu\right)  \log
{(\eta_{i}^{\left(  l-1\right)  }\left(  \mu\right)  ).}%
\end{multline}
We finally obtain the coherent information as the difference of
(\ref{eq:HBX-Unruh}) and (\ref{eq:HEX-unruh}):%
\begin{multline*}
I(A\rangle BX)=-\sum_{l=2}^{\infty}\frac{p_{l}\left(  z\right)  }{\Delta
_{l-1}}\sum_{i=0}^{l-1}\lambda_{i}^{\left(  l-1\right)  }\left(  \mu\right)
\log{(\lambda_{i}^{\left(  l-1\right)  }\left(  \mu\right)  )}\\
+\sum_{l=2}^{\infty}\frac{p_{l}\left(  z\right)  }{\Delta_{l-1}}\sum
_{i=0}^{l-2}\eta_{i}^{\left(  l-1\right)  }\left(  \mu\right)  \log{(\eta
_{i}^{\left(  l-1\right)  }\left(  \mu\right)  )}.
\end{multline*}
The above expression vanishes when $\mu=0$, and it coincides with the
expression for the quantum capacity of the Unruh channel when $\mu=1/2$ (in
Section III B of Ref.~\cite{BDHM09}):%
\begin{equation}
I(A\rangle BX)_{\mu=1/2}=\sum_{k=0}^{\infty}p_{k+2}\left(  z\right)
\log\left(  {\frac{k+2}{k+1}}\right)  .\nonumber
\end{equation}

We now trace out the CE\ trade-off curve for a single use of the Unruh
channel. We use the same input state as in
Lemma~\ref{lem:parametrize-EAC-cloning}\ because that lemma proves that this
input state traces out both the CQ\ curve and the CE curve. We then have here
that $H(A|X)=H_{2}(\mu)$, and we obtain the expression in the statement of the
theorem by noting that $I\left(  AX;B\right)  =H\left(  A|X\right)  +H\left(
B\right)  -H\left(  E|X\right)  $.
\end{proof}

\bibliographystyle{unsrt}
\bibliography{Ref}

\end{document}